\RequirePackage{fix-cm}

\documentclass[numbook, envcountsect, envcountsame, envcountreset, runningheads, smallextended]{svjour3}
\smartqed

\usepackage{amsmath,amssymb,amsfonts,latexsym}
\usepackage[dvips]{graphicx}
\usepackage{color}
\usepackage{array}
\usepackage{booktabs, multicol, multirow}
\usepackage{verbatim}

\usepackage{bbm}
\usepackage{tikz}
\usepackage{caption}
\usepackage{subcaption}
\usetikzlibrary{trees,shapes,snakes}
\usepackage{enumitem}

\newcommand{\bibmod}[1]{}

\newcommand{\cred}{\color{black}}
\newcommand{\cb}{\color{black}}

\def\JELcode{\textbf{JEL Classification}\enspace}
\def\JEL#1{\par\addvspace\medskipamount{\rightskip=0pt plus1cm
		\def\and{\ifhmode\unskip\nobreak\fi\ $\cdot$
		}\noindent\JELcode\enspace\ignorespaces#1\par}}

\begin{document}

\title{Speculative Trading, Prospect Theory and Transaction Costs}
\titlerunning{Speculative Trading, Prospect Theory and Transaction Costs}

\author{Alex S.L. Tse \and Harry Zheng
}

\authorrunning{Alex S.L. Tse and Harry Zheng}

\institute{Alex S.L. Tse \at
	Department of Mathematics, University College London, London WC1H 0AY, UK. \\
	\email{alex.tse@ucl.ac.uk}           
	\and
	Harry Zheng \at
	Department of Mathematics, Imperial College London, London SW7 2AZ, UK.\\
	\email{harry.zheng@imperial.ac.uk}
}

\date{\today}
\maketitle

\begin{abstract}
A speculative agent with Prospect Theory preference chooses the optimal time to purchase and then to sell an indivisible risky asset to maximize the expected utility of the round-trip profit net of transaction costs. The optimization problem is formulated as a sequential optimal stopping problem and we provide a complete characterization of the solution. Depending on the preference and market parameters, the optimal strategy can be ``buy and hold'', ``buy low sell high'', ``buy high sell higher'' or ``no trading''. Behavioral preference and market friction interact in a subtle way which yields surprising implications on the agent's trading patterns. For example, increasing the market entry fee does not necessarily curb speculative trading, but instead it may induce a higher reference point under which the agent becomes more risk-seeking and in turn is more likely to trade. 
\end{abstract}
\vspace{\baselineskip}




\section{Introduction}
\label{sect:intro}

When it comes to modeling of trading behaviors, the standard economic paradigm is the maximization of risk-averse agents' expected utility in a frictionless market. This criterion however has been criticized on many levels. In terms of trading environment, financial friction is omnipresent in reality where transactions are subject to various costs. In terms of agents' preferences, behavioral economics literature suggests that many individuals do not make decisions in accordance to expected utility theory. First, utilities are not necessarily derived from final wealth but typically what matters is the change in wealth relative to some reference point. Second, individuals are usually risk-averse over the domain of gains but risk-seeking over the domain of losses - this can be captured by an S-shaped utility function. Finally, individuals may fail to take portfolio effect into account when making investment decision and this phenomenon is known as narrow framing. These psychological ideas are explored for example in the seminal work of Kahneman and Tversky~\cite{kahneman-tversky79}, Tversky and Kahneman~\cite{tversky-kahneman81, tversky-kahneman92} and Kahneman and Lovallo~\cite{kahneman-lovallo93}.

In this paper, we develop a tractable dynamic trading model which captures a number of stylized behavioral biases of individuals as well as market frictions. In our setup, trading is costly due to proportional transaction costs as well as a fixed market entry fee. The goal of an agent is to find the optimal time to buy and then to sell an indivisible risky asset to maximize the expected utility of the round-trip profit under Prospect Theory preference of Tversky and Kahneman~\cite{tversky-kahneman92}. While a realistic economy can consist of multiple assets, we can interpret the assumption of a single indivisible asset as a manifestation of narrow framing such that the trading decision associated with one particular unit of indivisible asset can be completely isolated from the other investment opportunities. We believe the model is the best suitable to describe the trading behaviors of speculative agents. These ``less-than-fully rational'' agents purchase and sell an asset with a narrow objective of making a one-off round-trip profit rather than supporting consumption or stipulating a long term portfolio growth. 

A sequential optimal stopping problem featuring an S-shaped utility function is solved to identify the entry and exit time of the market by the agent. The solution approach is based on a backward induction idea. In the first stage, we focus on the exit strategy of the agent: Conditional on the ownership of the asset purchased at a given price level (which determines the agent’s reference point), the optimal liquidation problem is solved. Then the value function of this exit problem reflects the utility value of purchasing the asset at different price level. Upon comparison against the utility value of inaction, we obtain the payoff function of the real option to purchase the asset which is then used in the second stage problem concerning the entry decision of the agent: The agent picks the optimal time to enter the trade as to maximize the expected payoff of this real option to purchase the asset.

The traditional route to analyze an optimal stopping problem is to first conjecture a candidate optimal stopping rule and then the dynamic programming principle is invoked to derive a free boundary value problem that the value function should satisfy. Then one can attempt to solve for the free boundaries via value matching and smooth pasting. For this approach to work, we need to correctly identify the form of the optimal stopping rule but this exercise may not be trivial. As it turns out, the optimal continuation region of our entry problem can either be connected or disconnected depending on the model parameters. It is thus difficult to adopt such a guess-and-verify approach since we do not know the correct form of the optimal stopping rule upfront. In our analysis, martingale method is employed to solve the underlying optimal stopping problems, which has an important advantage that no priori conjecture on the optimal strategy is required. The optimal continuation/stopping set can be deduced directly by studying the smallest concave majorant to a suitably scaled payoff function.

Despite its relatively simple nature, our model is capable of generating a rich variety of trading behaviors such as ``buy and hold'', ``buy low sell high'', ``buy high sell higher'' and ``no trading''. The risk-seeking preference of a behavioral agent over the loss domain will typically encourage him to enter the trade but his precise trading behaviors depend crucially on the level of transaction costs relative to his preference parameters. Generally speaking, a high proportional (fixed) transaction cost discourages trading at a high (low) nominal price. When proportional costs are high and the asset is expensive, the agent prefers waiting until the price level declines and hence he is more inclined to consider a ``buy low sell high’’ strategy. But if instead the fixed entry fee is high and the asset is cheap, the agent might prefer delaying the purchase decision until the asset reaches a higher price level, and this leads to a trading pattern of ``buy high sell higher’’. 

Both behavioral preferences and market frictions are studied extensively as separate topics in the mathematical finance literature. To the best of our knowledge, however, their interaction has not been explored to date. Under Prospect Theory, the risk attitude of the agent is heavily influenced by the reference point. Since the reference point is endogenized in our model which depends on the cost of purchase including the transaction cost paid, the level of transaction cost has a direct impact on the agent's risk preference. This subtle interaction between risk preference and transaction cost leads to interesting policy implications on how speculative trading can be curbed effectively. For example, a surprisingly result is that imposing a fixed market entry fee might indeed accelerate rather than cool down trading participation.

Our paper is closely related to the literature of optimal stopping under S-shaped utility function. Kyle et al.~\cite{kyle-ouyang-xiong06} and Henderson~\cite{henderson12} consider a one-off optimal liquidation problem in which the agent solves for the optimal time to liquidate an endowed risky asset to maximize the expected Prospect Theory utility. They do not consider the purchase decision and the reference point is taken as some exogenously given status quo. A main contribution of our paper is that we further endogenize the reference point which depends on the purchase price of the asset, and the optimal purchase price must be determined as a part of the optimization problem. The recent work of Henderson and Muscat~\cite{henderson-muscat20} extends the model of Henderson~\cite{henderson12} by considering partial liquidation of multiple indivisible assets. Both this paper and Henderson and Muscat~\cite{henderson-muscat20} consider a sequential optimal stopping problem as the underlying mathematical framework. However, the economic natures of the problems are completely different where we study the sequential decision of purchase and sale while they exclusively focus on sales.

Another relevant class of works is the realization utility model which further incorporates reinvestment possibility within a behavioral optimal stopping model such as Barberis and Xiong~\cite{barberis-xiong12}, Ingersoll and Jin~\cite{ingersoll-jin13}, He and Yang~\cite{he-yang19}, Kong et al.~\cite{kong-qin-yue22} and Dai et al.~\cite{dai-qin-wang22}. In such models, the agent repeatedly purchases and sells an asset to maximize the sum of utility bursts realized from the gain and loss associated with each round-trip transaction. In a certain sense, their models consider endogenized reference point which is continuously updated based on the historical prices within each trading episode. However, the purchase decision is exogenously given in many of these models where the agent is simply assumed to buy the asset again immediately after a sale. These cited papers on realization utility models all feature transaction costs which are required to make the problems well-posed. As a result, the purchase pattern is not entirely realistic: If the agent is willing to sell an asset and then instantaneously repurchase an identical (or a different but statistically identical) asset, then the agent is essentially throwing away money in form of transaction costs without altering his own financial position. Only He and Yang~\cite{he-yang19} carefully analyze the purchase decision of the agent, but in any case they find that the purchase strategy is trivial where the agent either buys the asset immediately after a sale or never enters the trade again. Our model differs from the realization utility model in a way that we do not consider perpetual reinvestment opportunities (which can be understood as narrow framing that the agent only focuses on a single episode of the trading experience when evaluating the entry and exit strategies). Nonetheless, the optimal purchase region of our model is non-trivial under typical parameters which encapsulates many realistic trading strategies. 

This casts doubts over whether such models are the best suitable in terms of explaining the purchase-and-sale behaviors of an individual. In our baseline model, the optimal purchase strategy can be non-trivial and this might hint how existing realization utility models can be extended to produce more realistic purchase behavior. See the discussion in Section \ref{sect:realization_uti}.

Beyond the context of behavioral economics, there are a few works attempting to model the sequential purchase and sale decisions under optimal switching framework. However, identification of a modeling setup which can generate reasonable trading patterns proves to be much more difficult than expected. On the one hand, Zervos et al. (\cite{zervos-johnson-alazemi13}, p.561) report that ``...the prime example of an asset price process, namely, the geometric Brownian motion, does not allow for optimal buying and selling strategies that have a sequential nature''. Indeed, existing literature which gives ``buy low sell high'' as an optimal trading strategy often relies on extra statistical features of the asset price process such as mean reversion. See for example Zhang and Zhang~\cite{zhang-zhang08}, Song et al.~\cite{song-yin-zhang09}, Leung et al.~\cite{leung-li-wang15} and Leung and Li~\cite{leung-li15}. On the other hand, momentum-based trading strategy is also rarely studied in mathematical finance literature. The scarce examples include the work of Dai, Zhang and Zhu~\cite{dai-zhang-zhu10} and Dai, Yang, Zhang and Zhu~\cite{dai-yang-zhang-zhu16} who find that trend-following strategy is optimal under a regime-switching model of asset price. We contribute to this strand of literature by showing that a trading model based on a simple geometric Brownian motion can also generate many realistic trading patterns including both reversal strategy (buy low sell high) and momentum strategy (buy high sell higher). This is achieved via incorporating standard elements of behavioral preferences and market frictions.

The optimal investment rule in the classical Merton~\cite{merton69, merton71} portfolio selection problem can also be viewed as a buy low sell high strategy: Since the agent keeps a constant fraction of wealth invested in the risky asset, extra units of risky asset are sold (purchased) when the price increases (falls), ceteris paribus. In our paper, we focus on a single indivisible asset and do not consider portfolio effect.

The rest of the paper is organized as follows. Section \ref{sect:prob} provides a description of the model and the underlying optimization problem. In Section \ref{sect:solmethod}, we outline the solution methods to a standard optimal stopping problem and discuss heuristically how the solution to our sequential optimal stopping problem shall be characterized via the idea of backward induction. The main results are collected in Section \ref{sect:main}. Some comparative statics results and their policy implications are discussed in Section \ref{sect:compstat}. {\cred Several extensions of the baseline model are discussed in Section \ref{sect:extend}}. Section \ref{sect:conc} concludes. A few technical proofs are deferred to the appendix.



\section{Problem description}
\label{sect:prob}

Let $(\Omega, \mathcal{F}, \{\mathcal{F}_t\},\mathbb{P})$ be a filtered probability space satisfying the usual conditions which supports a one-dimensional Brownian motion $B=(B_t)_{t\geq 0}$. There is a single indivisible risky asset in the economy. Its price process $P=(P_t)_{t\geq 0}$ is modeled by a one-dimensional diffusion with state space $\mathcal{J}\subseteq \mathbb{R}_{+}$ and dynamics of
\begin{align*}
dP_t=\mu(P_t)dt+\sigma(P_t)dB_t,
\end{align*}
where $\mu:\mathcal{J}\to\mathbb{R}$ and $\sigma:\mathcal{J}\to(0,\infty)$ are Borel functions. $\mathcal{J}$ is assumed to be an interval with endpoints $0\leq a_{\mathcal{J}}<b_{\mathcal{J}}\leq \infty$ and that $P$ is regular in $(a_{\mathcal{J}},b_{\mathcal{J}})$, {\cred i.e. for any $p,y\in(a_{\mathcal{J}},b_{\mathcal{J}})$ we have $\mathbb{P}[\tau_y<\infty|P_0=p]>0$ where $\tau_y:=\inf\{t\geq 0: P_t=y\}$.}

We assume that interest rate is zero in our exposition. For the non-zero interest rate case one can interpret the process $P$ as the numeraire-adjusted price of the asset. Then the drift term $\mu(\cdot)$ can be viewed as the instantaneous excess return of the risky asset. 

Trading in the asset is costly. If the agent wants to purchase the asset at its current price $p$, he will need to pay $\lambda p + \Psi$ to initiate the trade where $\lambda\in[1,\infty)$ is the proportional transaction cost on purchase and $\Psi\geq 0$ represents a fixed market entry fee. When the agent sells the asset at price $p$, he will only receive $\gamma p$ where $\gamma\in(0,1]$ is the proportional transaction cost on sale.

Preference of the agent is described by Prospect Theory of Tversky and Kahneman~\cite{tversky-kahneman92}. Under this framework, utility is derived from gains and losses relative to some reference point rather than the total wealth. Individuals are typically risk-averse over the domain of gains and risk-seeking over the domain of losses. This can be captured by an S-shaped utility function $U:\mathbb{R}\to\mathbb{R}$ with $U(0)=0$ and that $U$ is concave (resp. convex) over $\mathbb{R}_{+}$ (resp. $\mathbb{R}_{-}$). Finally, individuals also exhibit loss-aversion such that the negative utility brought by a unit of loss is much larger in magnitude than the positive utility from a unit of gain. 

In behavioral optimal liquidation literature such as Kyle et al.~\cite{kyle-ouyang-xiong06} and Henderson~\cite{henderson12}, the liquidation payoff is always compared against some exogenously given constant reference point. In our setup, we assume the reference point depends on both an exogenous constant $R$ as well as the amount paid by the agent to purchase the asset. Suppose the agent has executed a speculative round-trip trade where he has bought and then sold the asset at stopping times $\tau$ and $\nu$ (with $\tau\leq \nu$) respectively. The liquidation payoff $\gamma P_{\nu}$ is evaluated against $\lambda P_{\tau}+\Psi+R$ as the reference point, where $\lambda P_{\tau}+\Psi$ is the capital spent on purchasing the asset and $R$ is a constant outside the model specification. The parameter $R$ can be interpreted as a preference parameter of the agent which reflects his ``aspiration level'' in the sense of Lopes and Oden~\cite{lopes-oden99} where a more motivated agent will set a higher economic benchmark as a profit target to beat. The realized utility of this round-trip trade is $U(\gamma P_{\nu}-\lambda P_{\tau}-\Psi-R)$.

A caveat, however, is that the agent is not obligated to enter or exit the trade at all if it is undesirable to do so. A realization of $\tau=\infty$ refers to the case that the purchase decision is deferred indefinitely, which is economically equivalent to not entering the trade at all. The liquidation value is zero because there is nothing to be sold, and the reference point becomes $R$ since the required cash outflow for purchase $\lambda P_{\tau}+\Psi$ has never materialized. Thus the Prospect Theory value under such strategy is simply $U(-R)$. Similarly, the agent may enter the trade at some time point but never liquidate the asset. This corresponds to a realization of $\tau<\infty$ and $\nu=\infty$. In this case, the liquidation value is again zero which is evaluated against the reference point $\lambda P_{\tau}+\Psi+R$. To summarize all the possibilities, the realized Prospect Theory utility associated with trading strategy $(\tau,\nu)$ shall be written as
\begin{align}
\begin{cases}
U(\gamma P_{\nu}-\lambda P_{\tau}-\Psi-R),& \tau<\infty,\nu<\infty;\\
U(-R),& \tau=\infty;\\
U(-\lambda P_{\tau}-\Psi-R),& \tau<\infty,\nu=\infty.
\end{cases}
\label{eq:realisedPT}
\end{align}

The objective of the agent is to find the optimal purchase time $\tau$ and sale time $\nu$ to maximize the expected value of \eqref{eq:realisedPT}.
Define the objective function as
\begin{align}
J(p;\tau,\nu):=\mathbb{E}\left[U\left(\gamma P_{\nu} {\mathbbm 1}_{\{\tau<\infty,\nu<\infty\}}-\left(\lambda P_{\tau}+\Psi\right) {\mathbbm 1}_{\{\tau<\infty\}}-R\right)\Bigl | P_0=p\right].
\label{eq:obj}
\end{align}
Formally, the agent is solving the sequential optimal stopping problem
\begin{align}
\mathcal{V}(p):=\sup_{\tau,\nu\in\mathcal{T}:\tau \leq \nu}J(p;\tau,\nu)
\label{eq:valfun}
\end{align}
where $\mathcal{T}$ is the set of $\{\mathcal{F}_t\}$-stopping times valued in $\mathbb{R}_{+}\cup \{+\infty\}$. Problem \eqref{eq:valfun} has two features which make it a non-standard one relative to a typical optimal stopping problem. First, the decision space is two-dimensional. Second, the objective function has an explicit dependence on the stopping times $\tau,\nu$ via the indicator functions which further complicate the analysis. 

\begin{remark}
Similar to Henderson~\cite{henderson12}, Xu and Zhou~\cite{xu-zhou13} and Henderson et al.~\cite{henderson-hobson-tse18}, we do not explicitly consider subjective discounting {\cred in our baseline model}. {\cred On the one hand, our model features cash flows at different time points and it is not entirely clear what is the most appropriate way to apply subjective discounting because the standard Prospect Theory framework is not directly applicable to intertemporal choices. On the other hand,} under discounting an impatient agent is much more inclined to delay losses and to realize profits earlier, this will lead to an extreme \textit{disposition effect} which is not consistent with the empirical trading pattern of retail investors. See the discussion in Henderson~\cite{henderson12}. {\cred At a mathematical level, introducing subjective discounting will also make our problem harder to be analyzed under full generality. We will briefly discuss in Section \ref{sect:discount} how subjective discounting might be incorporated and explore (numerically in some cases) how it affects the optimal trading behaviors.}
\end{remark}




\section{The solution methods}
\label{sect:solmethod}

In this section, we give an overview of how problem \eqref{eq:valfun} can be solved. We begin by offering a short summary about the solution approach to solve a standard optimal stopping problem for one-dimensional diffusion.

\subsection{The martingale methods for optimal stopping problems}
\label{sect:mgmethod}

We review the martingale methods to solve an undiscounted optimal stopping problem, which is based on Dynkin and Yushkevich~\cite{dynkin-yushkevich69} and Dayanik and Karatzas~\cite{dayanik-karatzas03}. 

Consider a general problem in form of
\begin{align*}
V(p)=\sup_{\tau\in\mathcal{T}}\mathbb{E}\left[G\left(P_{\tau}\right)|P_0=p\right]
\end{align*}
for some payoff function $G$. Under standard theory of optimal stopping, the optimal stopping time can be characterized by the first exit time of the process from some open set $\mathcal{C}$ such that $\tau=\inf\{t\geq 0: P_t\notin C\}$. In a one-dimensional diffusion setting, it is sufficient to consider stopping times which have the form $\tau_{a,b}:=\tau_a \wedge \tau_b$ where $\tau_a:=\inf\{t\geq 0: P_t=a\}$ and $\tau_b:=\inf\{t\geq 0: P_t=b\}$ with $a\leq p\leq b$. Here $[a,b]\subseteq \mathcal{J}$ is the unknown interval to be identified (and it depends on $p$ in general).

Let $s(\cdot)$ be the scale function of process $P$ (which is unique up to affine transformation) defined as a strictly increasing function such that $\Theta_t:=s(P_t)$ is a local martingale. A simple application of Ito's lemma shows that $s(\cdot)$ should solve the second order differential equation
\begin{align}
\frac{\sigma^2(p)}{2}s''(p)+\mu(p)s'(p)=0.
\label{eq:scaleode}
\end{align}
Let $\theta:=s(p)$. Then
\begin{align*}
J(p;\tau_{a,b}):=\mathbb{E}\bigl[G(P_{\tau_{a,b}})|P_{0}=p\bigr]
&=\mathbb{E}\Bigl[G\bigl(s^{-1}(\Theta_{\tau_{a,b}})\bigr)\Bigl|\Theta_{0}=\theta\Bigr]\\
&=\mathbb{E}\bigl[\phi(\Theta_{\tau_{a,b}})\bigl|\Theta_{0}=\theta\bigr]\\
&=\mathbb{P}[\tau_a<\tau_b|\Theta_{0}=\theta]\phi\bigl(s(a)\bigr)\\
&\qquad +\mathbb{P}[\tau_b<\tau_a|\Theta_{0}=\theta]\phi\bigl(s(b)\bigr)\\
&=\frac{s(b)-\theta}{s(b)-s(a)}\phi\bigl(s(a)\bigr)+\frac{\theta-s(a)}{s(b)-s(a)}\phi\bigl(s(b)\bigr)
\end{align*}
where $\phi:=G\circ s^{-1}$. The above can be maximized with respect to $a$ and $b$. Moreover, the dummy variables $a$ and $b$ can be replaced by $a'=s(a)$ and $b'=s(b)$. Hence
\begin{align*}
V(p)=\sup_{a,b:a\leq p\leq b}J(p;\tau_{a,b})=\sup_{a',b':a'\leq \theta\leq b'}\left[\frac{b'-\theta}{b'-a'}\phi(a')+\frac{\theta-a'}{b'-a'}\phi(b')\right]=:v(\theta)
\end{align*}
and thus $V(p)=v(s(p))$. The scaled value function $v(\theta)$ can be characterized by the smallest concave majorant to $\phi(\theta)=G(s^{-1}(\theta))$ the scaled payoff function over $s(\mathcal{J})$ which is defined as an interval with endpoints $s(a_\mathcal{J})$ and $s(b_\mathcal{J})$. The continuation and stopping set associated with the optimal stopping rule are given by $\mathcal{C}=\{p\in\mathcal{J}:v(s(p))>\phi(s(p))\}$ and $\mathcal{S}=\{p\in\mathcal{J}:v(s(p))=\phi(s(p))\}$ respectively.

\subsection{Decomposition of the sequential optimal stopping problem}

In the following two subsections, we discuss heuristically how the value function of problem \eqref{eq:valfun} can be constructed by considering two sub-problems based on the idea of backward induction. The well-posedness conditions as well as formal verification of optimality will be explored in Section \ref{sect:main} when we specialize the modeling setup.

\subsubsection{Exit problem}
\label{sect:exit}

Suppose for the moment that the agent has already purchased the asset at some known price level $q$ via paying $\lambda q+\Psi$ at some time point in the past. Conditional on this information, the reference point has been fixed at a known constant level $H:=\lambda q + \Psi+R$. Suppose the current time is labeled as $t=0$ and the current price of the asset is $P_0=p$. The goal of the agent in the exit problem is to find the optimal time to sell this owned asset to maximize the expected Prospect Theory value of the sale proceed relative to the reference point $H$. If the asset is (ever) sold at time $\nu$, the utility of gain and loss relative to the reference point is $U(\gamma P_{\nu} {\mathbbm 1}_{\{\nu<\infty\}}-H)$ after taking the transaction cost on sale into account. Since the realized utility is increasing in $P_{\nu}$ and the process $P$ is non-negative, in general there is no incentive for the agent to forgo the sale opportunity. Hence heuristically one can drop the indicator function ${\mathbbm 1}_{\{\nu<\infty\}}$ and it is sufficient to consider the objective function $G_1(P_{\nu};H):=U(\gamma P_{\nu}-H)$. The agent solves an optimal stopping problem
\begin{align}
V_1(p; H):=\sup_{\nu\in\mathcal{T}}\mathbb{E}\Bigl[G_1(P_{\nu};H)\Bigl | P_0=p\Bigr]=\sup_{\nu\in\mathcal{T}}\mathbb{E}\Bigl[U(\gamma P_{\nu}-H)\Bigl | P_0=p\Bigr]
\label{eq:exit}
\end{align}
to determine the optimal time of the asset sale. The value function of the exit problem is then given by $V_1(p; H)=\bar{g}_1(s(p); H)$ where $\bar{g}_1=\bar{g}_1(\theta;H)$ is the smallest concave majorant of
\begin{align*}
g_1(\theta;H):=G_1\Bigl(s^{-1}(\theta); H\Bigr)=U\Bigl(\gamma s^{-1}(\theta)-H\Bigr).
\end{align*}
We write the optimizer to problem \eqref{eq:exit} as $\nu^*(p;H)$ which depends on the initial price level $p$ and the given reference point $H$.

\subsubsection{Entry problem}
\label{sect:entry}

Now we assume that the agent does not own any asset to begin with. His economic objective is to determine the optimal time to purchase (and then to sell) the asset to maximize the expected utility of the liquidation proceed relative to the endogenized reference point.

At a given current asset price level $p$, there are two possible actions for the agent. First, he can opt to initiate the speculative trade by buying the asset now which fixes the reference point as $\lambda p+\Psi+R$, and then sell it later in the future. When the asset is liquidated at his choice of the sale time $\nu$, the realized utility is $U(\gamma P_{\nu}-\lambda p-\Psi-R)$. Conditional on the decision to purchase the asset today at price $p$, the agent can find the best time of sale to maximize his expected utility by solving problem \eqref{eq:exit} on setting $H=\lambda p + \Psi+R$. Then the best possible expected utility he can attain is
\begin{align*}
\sup_{\nu\in\mathcal{T}}\mathbb{E}\Bigl[U(\gamma P_{\nu}-\lambda p-\Psi-R)\Bigl | P_0=p\Bigr]=V_{1}(p; \lambda p+\Psi+R)
\end{align*}
provided that he decides to enter the trade at the given price of $p$.

Alternatively, the agent can forgo the opportunity to enter the trade and stay away from the market forever. In this case, the payoff must be zero and the reference point is simply equal to $R$. The utility he will receive is just a constant of $U(-R)$.

Therefore, the opportunity to enter the speculative trade can be viewed as a real option. At a given price level $p$ the agent is willing to enter the trade only if the maximal expected utility of trading is not less than that of inaction, i.e. $V_{1}(p; \lambda p+\Psi+R)\geq U(-R)$. This is similar to a financial option being in-the-money. The payoff of this real option to the agent in utility terms as a function of price level $p$ is given by
\begin{align}
G_2(p):=\max\Bigl\{V_1(p;\lambda p+\Psi+R),U(-R)\Bigr\}.
\label{eq:G2}
\end{align}
The entry problem for the agent is to find the optimal time to initiate the trade as to maximize the expected value of \eqref{eq:G2}. It is equivalent to solving
\begin{align}
V_2(p)&:=\sup_{\tau\in\mathcal{T}}\mathbb{E}\Bigl[G_2(P_{\tau})\Bigl|P_0=p\Bigr]\nonumber\\
&=\sup_{\tau\in\mathcal{T}}\mathbb{E}\biggl[\max\Bigl\{V_1(P_\tau;\lambda P_\tau+\Psi+R),U(-R)\Bigr\}\biggl|P_0=p\biggr]
\label{eq:entry}
\end{align}
provided that the exit problem value function $V_1$ is well-defined. We identify $\bar{g}_2=\bar{g}_2(\theta)$ as the smallest concave majorant of
\begin{align*}
g_2(\theta):=G_2(s^{-1}(\theta))&=\max\Bigl\{V_1\bigl(s^{-1}(\theta); \lambda s^{-1}(\theta)+\Psi + R\bigr),U(-R)\Bigr\}\\
&=\max\Bigl\{\bar{g}_1\bigl(\theta; \lambda s^{-1}(\theta)+\Psi + R\bigr),U(-R)\Bigr\}.
\end{align*}
Then the value function of the entry problem is $V_2(p)=\bar{g}_2(s(p))$.

Let the optimizer to problem \eqref{eq:entry} be $\tau^*(p)$. With $p$ being the initial price of the asset at $t=0$, the agent will purchase the asset at the stopping time $t=\tau^*(p)$. Then conditional on the realization of the entry price level $P_{\tau^*(p)}$, the agent solves the exit problem \eqref{eq:exit} under initial value $P_{\tau^*(p)}$ and reference point $H=\lambda P_{\tau^*(p)}+ \Psi+R$. The corresponding optimizer is given by $\nu^*(P_{\tau^*(p)};\lambda P_{\tau^*(p)}+ \Psi+R)$ which reflects the time lapse between the initiation and closure of the trade. In particular, the agent will sell the asset at the stopping time $t=\tau^*(p)+\nu^*(P_{\tau^*(p)};\lambda P_{\tau^*(p)}+ \Psi+R)$. This gives the complete characterization of the optimal entry and exit strategy of the agent. Note that it is possible to have $\mathbb{P}[\tau^*(p)<\infty]<1$. In other words, there is a possibility that the entry strategy is not executed in finite time, and hence there is no decision to sell (see the discussion in Section \ref{sect:main}).

Intuitively, we expect $\mathcal{V}(p)=V_2(p)$ where $\mathcal{V}$ is the value function of the original sequential optimal stopping problem \eqref{eq:valfun}. This claim has to be verified formally. Without any further specifications of the utility function $U$ and the underlying price process $P$, however, it is hard to make any further progress. For example, it is not even clear upfront whether \eqref{eq:valfun} is a well-posed problem which yields a finite value function.

\section{Main results}
\label{sect:main}

The procedures described in Section \ref{sect:solmethod} are very generic and can guide us to write down the value function of the sequential optimal stopping problem under a range of model specifications. To derive stronger analytical results, in the rest of this paper we specialize to the piecewise power utility function of Tversky and Kahneman~\cite{tversky-kahneman92} in form of
\begin{align*}
U(x)=
\begin{cases}
x^\alpha,&x>0; \\
-k |x|^\alpha,& x\leq 0.
\end{cases}
\end{align*}
Here $\alpha\in(0,1)$ such that $1-\alpha$ is the level of risk-aversion and risk-seeking on the domain of gains and losses, and $k>1$ controls the degree of loss-aversion. Experimental results of \cite{tversky-kahneman92} give an estimation of $\alpha=0.88$ and $k=2.25$.

The price process of the risky asset $P=(P_t)_{t\geq 0}$ is assumed to be a geometric Brownian motion 
\begin{align*}
dP_t=P_t(\mu dt +\sigma dB_t)
\end{align*}
with $\mu\geq 0$ and $\sigma>0$ being the constant drift and volatility of the asset. Define $\beta:=1-\frac{2\mu}{\sigma^2}\leq 1$, then by substituting $\mu(p)=\mu p$ and $\sigma(p)=\sigma p$ in \eqref{eq:scaleode} a scale function of $P$ can be found as
\begin{align*}
s(x)=
\begin{cases}
x^{\beta},& \beta > 0;\\
\ln x,& \beta=0;\\
x^{-\beta},& \beta<0.
\end{cases}
\end{align*}

Finally, we assume $R>0$ so that the aspiration level of the agent is always positive. This is not unreasonable since this parameter can be understood as some performance benchmark that an agent wants to outperform and such a goal is typically a positive one. 

We begin by looking at the necessary condition under which problem \eqref{eq:valfun} is well-posed.
\begin{lemma}
	If $\beta\leq 0$ or $0<\beta<\alpha$, then there exists a sequence of stopping times $(\tau_n,\nu_n)_{n=1,2,...}$ such that $J(p;\tau_n,\nu_n)\to +\infty$ as $n \uparrow \infty$, where $J(p;\cdot,\cdot)$ is defined in \eqref{eq:obj}.
	\label{lem:wellpose}
\end{lemma}
\begin{proof}
	Consider a sequence of stopping times defined by
	\begin{align}
	\tau_n:=0,\qquad \nu_n:=\inf\{t\geq 0: P_t \geq n\}.
	\label{eq:buy_and_hold}
	\end{align}
	If $\beta\leq 0$, then $$P_t=P_0\exp\biggl[\Bigl(\mu-\frac{\sigma^2}{2}\Bigr)t+\sigma B_t\biggr]=P_0\exp\biggl[\sigma\Bigl(-\frac{\sigma}{2}\beta t+ B_t\Bigr)\biggr]$$ such that the Brownian motion in the exponent has non-negative drift and hence $P$ can reach any arbitrarily high level in finite time, and as such $\nu_n<\infty$ and $P_{\nu_n}=n$ almost surely. Hence 
	\begin{align*}
	J(p;\tau_n,\nu_n)&=\mathbb{E}\left[U\left(\gamma P_{\nu_n} {\mathbbm 1}_{\{\tau_n<\infty,\nu_n<\infty\}}-\left(\lambda P_{\tau_n}+\Psi\right) {\mathbbm 1}_{\{\tau_n<\infty\}}-R\right)\Bigl | P_0=p\right]\\
	&=U\left(\gamma n-\lambda p-\Psi-R\right)\to +\infty
	\end{align*}
	as $n\uparrow \infty$.
	
	If $\beta>0$, then instead $P$ may not reach any arbitrarily given level above its starting value in finite time and we have $\lim_{t\to\infty}P_t=0$ almost surely. Then $\{v_n<\infty\}=\{P_{\nu_n}=n\}$, and for sufficiently large $n$ such that $n>p$ one can compute that
	\begin{align*}
	\mathbb{P}[\nu_n=\infty|P_0=p]=\frac{n^\beta-p^\beta}{n^\beta},\qquad \mathbb{P}[\nu_n<\infty|P_0=p]=\frac{p^\beta}{n^\beta}.
	\end{align*}
	Then
	\begin{align*}
	J(p;\tau_n,\nu_n)&=\mathbb{P}[\nu_n=\infty|P_0=p]U\left(-\lambda p-\Psi-R\right)\\
	&\qquad+\mathbb{P}[\nu_n<\infty|P_0=p]U\left(\gamma n-\lambda p-\Psi-R\right) \\
	&=-\frac{n^\beta-p^\beta}{n^\beta}k(\lambda p +\Psi+R)^{\alpha}+\frac{p^\beta}{n^\beta}(\gamma n-\lambda p-\Psi-R)^{\alpha}\to +\infty
	\end{align*}
	as $n\uparrow \infty$ if $\beta<\alpha$. \qed
\end{proof}

Mathematically speaking, the sequential optimal stopping problem \eqref{eq:valfun} is ill-posed under the parameters combination in Lemma \ref{lem:wellpose} where the value function diverges to infinity. This arises when the performance of the asset is too good relative to the agent's risk-aversion level over gains. \eqref{eq:buy_and_hold} corresponds to a ``buy and hold'' trading rule as a possible optimal strategy: the agent purchases the asset immediately from the outset and the profit-target level of sale can be set arbitrarily high. 

\begin{remark}
 Empirically, historical returns on equities are excessively high relative to their risk level. For example, the annualized mean and standard deviation of the equity premium (i.e. excess return above the riskfree rate) of the U.S. market over the time period 1889-1978 are 6.18\% and 16.67\% respectively (see Mehra and Prescott\cite{mehra-prescott85}) such that $\beta=1-2\mu/\sigma\approx -3.45$, while the estimates based on a more recent time period 1950-2015 are 7.15\% and 16.83\% (see Bai and Lu~\cite{bai-lu22}) such that $\beta\approx -4.05$. Although this may cast doubt over the practical relevance of the condition $0<\alpha\leq\beta$ (where buy-and-hold is an optimal strategy if the condition does not hold), we would like to emphasize $\beta$ in general is a noisy statistical quantity which is hard to be forecasted. In our model, $\beta$ should be interpreted as the agent's subjective assessment of the asset performance which can be much more conservative than the historical estimates. $\mu$ may also encapsulate subjective discounting which further lowers its value. See Section \ref{sect:discount}.
\end{remark}

From here onwards, we focus on the case $0<\alpha\leq \beta$ which is the necessary condition for problem \eqref{eq:valfun} to be well-posed. The form of the solutions to the exit problem \eqref{eq:exit} and entry problem \eqref{eq:entry} are first provided and then we discuss the economic intuitions behind the associated trading strategies. Towards the end of this section, the optimality of the value function of the entry problem is formally verified to show that it indeed corresponds to the solution of the sequential optimal stopping problem \eqref{eq:valfun}.

We first state the solution to the exit problem \eqref{eq:exit} where a similar result can be found in Henderson~\cite{henderson12}.

\begin{lemma}
	For the exit problem \eqref{eq:exit}, if $0<\alpha\leq \beta$ then the agent will sell the asset when its price level first reaches $\frac{cH}{\gamma}$ or above where $c>1$ is a constant given by the solution to the equation
	\begin{align}
	\frac{\alpha}{\beta} c(c-1)^{\alpha-1}-(c-1)^{\alpha}-k=0.
	\label{eq:eqC}
	\end{align}
	The value function is given by
	\begin{align}
	V_1(p; H)=
	\begin{cases}
	\frac{\alpha}{\beta} c^{1-\beta}(c-1)^{\alpha-1}H^{\alpha-\beta}(\gamma p)^\beta-kH^{\alpha},&p< \frac{cH}{\gamma}; \\
	(\gamma p-H)^{\alpha},&p\geq \frac{cH}{\gamma}.
	\end{cases}
	\label{eq:exitvalufun}
	\end{align}
	\label{lem:henderson}
\end{lemma}

\begin{proof}
	Recall the notation introduced in Section \ref{sect:exit}. For $\beta>0$, the scaled payoff function of the exit-problem is given by
	\begin{align*}
	g_1(\theta)=g_1(\theta;H)=G_1(s^{-1}(\theta);H)&=U(\gamma \theta^{\frac{1}{\beta}} -H) \\
	&=
	\begin{cases}
	-k(H - \gamma \theta^{\frac{1}{\beta}})^{\alpha},& 0\leq \theta< \left(\frac{H}{\gamma}\right)^\beta;\\
	(\gamma \theta^{\frac{1}{\beta}} -H)^{\alpha},& \theta\geq \left(\frac{H}{\gamma}\right)^\beta.
	\end{cases}
	\end{align*}
	It is straightforward to work out the derivatives of $g_1$ as
	\begin{align*}
	g_1'(\theta)=
	\begin{cases}
	\frac{k\alpha\gamma}{\beta}\theta^{\frac{1}{\beta}-1}(H - \gamma \theta^{\frac{1}{\beta}})^{\alpha-1},& 0\leq \theta< \left(\frac{H}{\gamma}\right)^\beta;\\
	\frac{\alpha\gamma}{\beta}\theta^{\frac{1}{\beta}-1}(\gamma \theta^{\frac{1}{\beta}} -H)^{\alpha-1},& \theta\geq \left(\frac{H}{\gamma}\right)^\beta,
	\end{cases}
	\end{align*}
	and
	\begin{align*}
	g_1''(\theta)=
	\begin{cases}
	\frac{k\alpha\gamma}{\beta}\theta^{\frac{1}{\beta}-2}(H-\gamma\theta^{\frac{1}{\beta}})^{\alpha-2}\left[\frac{\gamma(\beta-\alpha)}{\beta}\theta^{\frac{1}{\beta}}+\frac{1-\beta}{\beta}H\right],& 0\leq \theta< \left(\frac{H}{\gamma}\right)^\beta;\\
	\frac{\alpha\gamma}{\beta}\theta^{\frac{1}{\beta}-2}(\gamma\theta^{\frac{1}{\beta}}-H)^{\alpha-2}\left[\frac{\gamma(\alpha-\beta)}{\beta}\theta^{\frac{1}{\beta}}-\frac{1-\beta}{\beta}H\right],& \theta\geq \left(\frac{H}{\gamma}\right)^\beta.
	\end{cases}
	\end{align*}
	
	Given the standing assumption $\beta\leq 1$ and the condition $0<\alpha\leq\beta$, $g_1$ is increasing concave on $\theta>(\frac{H}{\gamma})^\beta$ and is increasing convex on $0\leq \theta<(\frac{H}{\gamma})^\beta$. The smallest concave majorant of $g_1$ can be formed by drawing a straight line from $(0,g_1(0))$ which touches $g_1$ at some $\theta^{*}>(\frac{H}{\gamma})^\beta$. In particular, $\theta^*$ is a solution to $\frac{g_1(\theta)-g_1(0)}{\theta}=g_1'(\theta)$ on $\theta>(\frac{H}{\gamma})^\beta$, i.e.
	\begin{align*}
	\frac{\alpha\gamma}{\beta}\theta^{\frac{1}{\beta}-1}(\gamma \theta^{\frac{1}{\beta}}-H)^{\alpha-1}=\frac{(\gamma \theta^{\frac{1}{\beta}} -H)^{\alpha}+kH^\alpha}{\theta}.
	\end{align*}
	We conjecture that the solution is in form of $\theta^*=c^\beta (\frac{H}{\gamma})^\beta$ for some constant $c>1$. Then direct substitution shows that the constant $c$ should solve \eqref{eq:eqC}. The smallest concave majorant of $g_1$ is then
	\begin{align*}
	\bar{g}_1(\theta)&=
	\begin{cases}
	g_1(0)+\theta g_1'(\theta^*),& 0\leq \theta<\theta^*;\\
	g_1(\theta),& \theta>\theta^*;
	\end{cases}\\
	&=
	\begin{cases}
	-kH^{\alpha}+\frac{\alpha}{\beta}H^{\alpha-\beta}c^{1-\beta}(c-1)^{\alpha-1}\gamma^\beta\theta,&0\leq \theta<c^\beta \left(\frac{H}{\gamma}\right)^\beta;\\
	(\gamma \theta^{\frac{1}{\beta}} -H)^{\alpha},& \theta>c^\beta \left(\frac{H}{\gamma}\right)^\beta.
	\end{cases}
	\end{align*}
	The value function is given by $V_1(p;H)=\bar{g}_1(s(p))=\bar{g}_1(p^\beta)$ leading to \eqref{eq:exitvalufun}. The corresponding optimal stopping time is
	\begin{align*}
		\tau=\inf\left\{t\geq 0:\bar{g}_1(\Theta_t)=g_1(\Theta_t)\right\}&=\inf\left\{t\geq 0:\Theta_t\geq \theta^{*}\right\}\\
		&=\inf\biggl\{t\geq 0:P_t\geq c \Bigl(\frac{H}{\gamma}\Bigr)\biggr\}.
	\end{align*} \qed
\end{proof}

The optimal sale strategy is a gain-exit rule where the agent is looking to sell the asset when its price is sufficiently high without considering stop-loss. Note that the gain-exit target $\frac{cH}{\gamma}$ is increasing in transaction costs (i.e. decreasing in $\gamma$). It means that the agent tends to delay the sale decision in a more costly trading environment.

\begin{remark}
	Inspired by the literature, we expect that more sophisticated exit strategies can be observed under alternative model setups. For example, the agent will consider stop-loss when the asset has negative drift (see Section \ref{sect:negative} and Henderson~\cite{henderson12}) or when there are multiple trading opportunities as per the realization utility models (see Section \ref{sect:realization_uti} as well as Ingersoll and Jin~\cite{ingersoll-jin13} and He and Yang~\cite{he-yang19}). Introducing jump-diffusion price dynamics will further result in disconnected waiting regions where an agent who has recently suffered from a huge jump loss may refuse to stop-loss (Dai et al.~\cite{dai-qin-wang22} and Kong et al.~\cite{kong-qin-yue22}). We opt to work with a simpler baseline model to make the whole entry-and-exit problem as tractable as possible. 
	\end{remark}

\begin{remark}
	Since we require $\beta>0$, if the initial price of the asset is below the gain-exit target then there is a strictly positive probability that the asset is never sold. Moreover, the agent who fails to sell the asset at his target gain-exit level will suffer a total loss in the long run.
	\label{remark:bm}
\end{remark}

\begin{remark}
	The expression of the target gain-exit threshold and in turn the value function of the exit problem are available in close-form, thanks to the specialization that the degree of risk-aversion over gains is the same as that of risk-seeking over losses. This allows us to make a lot of analytical progress when studying the entry problem. We will also lose the close-form expressions in Lemma \ref{lem:henderson} if fixed transaction cost on sale is introduced: In this case the agent will only sell the asset when the utility proceed from the sale $U(\gamma p - H-\Gamma)$ is larger than that of inaction $U(-H)$ where $\Gamma\geq 0$ represents a fixed market exit fee. Then the payoff function of the exit problem will become $$G_1(p;H):=\max\Bigl\{U(\gamma p - H-\Gamma),U(-H)\Bigr\}.$$   
	\label{remark:modelassump}
\end{remark}

We now proceed to describe the optimal entry strategy of the agent. The proofs of the two propositions below are given in the appendix.
\begin{proposition}
	Suppose $\alpha\leq \beta< 1$. For the entry problem \eqref{eq:entry}, the value function is given by $V_2(p)=\bar{g}_2(p^{\beta})$ where $\bar{g}_2=\bar{g}_2(\theta)$ is the smallest concave majorant to
	\begin{align}
	g_2(\theta)&:=\max\bigl\{v_1(\theta),-kR^{\alpha}\bigr\}:=\max\Biggl\{(R+\Psi)^{\alpha}f\biggl(\Bigl(\frac{\gamma}{R+\Psi}\Bigr)^{\beta}\theta\biggr),-kR^{\alpha}\Biggr\}
	\label{eq:g2}
	\end{align}
	with
	\begin{align}
	f(x):=\frac{\frac{\alpha}{\beta}c^{1-\beta}(c-1)^{\alpha-1}x-k(\frac{\lambda}{\gamma} x^{1/\beta}+1)^{\beta}}{(\frac{\lambda}{\gamma} x^{1/\beta}+1)^{\beta-\alpha}}.
	\label{eq:f}
	\end{align}
	There are three different cases:
	\begin{enumerate}[wide, labelwidth=!, labelindent=0pt, label={(\arabic*)}]
		\item If $\frac{\lambda}{\gamma}\leq[\frac{\alpha}{\beta k}c^{1-\beta}(c-1)^{\alpha-1}]^{\frac{1}{\beta}}$, there exists $p_1^{*}\in[0,\infty)$ such that the agent will enter the trade when the asset price is at or above $p_1^*$. The value function is
		\begin{align}
		V_2(p)=
		\begin{cases}
		\frac{\alpha}{\beta} c^{1-\beta}(c-1)^{\alpha-1}(\lambda p+\Psi+R)^{\alpha-\beta}(\gamma p)^\beta-k(\lambda p+\Psi+R)^{\alpha},& p\geq p_1^*;\\
		\Bigl(\frac{\alpha}{\beta} c^{1-\beta}(c-1)^{\alpha-1}(\lambda p_1^*+\Psi+R)^{\alpha-\beta}(\gamma p_1^*)^\beta&\\
		\qquad \qquad-k(\lambda p_1^*+\Psi+R)^{\alpha}+kR^{\alpha}\Bigr)\frac{p^\beta}{(p_1^*)^{\beta}}-kR^{\alpha},&p<p_1^*.
		\end{cases}
		\label{eq:exitvalfun1}
		\end{align}
		
		\item If $\frac{\lambda}{\gamma}> [\frac{\alpha}{\beta k}c^{1-\beta}(c-1)^{\alpha-1}]^{\frac{1}{\beta}}$, there exists a constant $C\in(0,\infty)$ independent of $\Psi$ and $R$ such that:
		
		\begin{enumerate}[wide, labelwidth=!, labelindent=0pt]
			\item If $\Psi< CR$, there exists $0\leq p_1^{*}<p_2^{*}<\infty$ such that the agent will enter the trade when the asset price is within the interval $[p_1^{*},p_2^{*}]$. The value function is
			\begin{align}
			V_2(p)=
			\begin{cases}
			\frac{\alpha}{\beta} c^{1-\beta}(c-1)^{\alpha-1}(\lambda p_2^*+\Psi+R)^{\alpha-\beta}(\gamma p_2^*)^\beta&\\
			\qquad\qquad-k(\lambda p_2^*+\Psi+R)^{\alpha},& p> p_2^*;\\
			\frac{\alpha}{\beta} c^{1-\beta}(c-1)^{\alpha-1}(\lambda p+\Psi+R)^{\alpha-\beta}(\gamma p)^\beta&\\
			\qquad\qquad-k(\lambda p+\Psi+R)^{\alpha},& p_1^* \leq p \leq p_2^*;\\
				\Bigl(\frac{\alpha}{\beta} c^{1-\beta}(c-1)^{\alpha-1}(\lambda p_1^*+\Psi+R)^{\alpha-\beta}(\gamma p_1^*)^\beta&\\
				\qquad\qquad -k(\lambda p_1^*+\Psi+R)^{\alpha}+kR^{\alpha}\Bigr)\frac{p^\beta}{(p_1^*)^{\beta}}-kR^{\alpha},&p<p_1^*.
			\end{cases}
			\label{eq:exitvalfun2a}
			\end{align}
			\item If $\Psi\geq CR$, the agent will never enter the trade. The value function is $V_2(p)=-kR^{\alpha}$.
		\end{enumerate}
	\end{enumerate}
	Moreover, in case (1) and case (2)(a), $p_1^*=\frac{R+\Psi}{\gamma} (x_1^*)^{1/\beta}$ where $x^*_1$ is the smaller solution to the equation
	\begin{align}
	k&=\left(1+\frac{\Psi}{R}\right)^{\alpha}\nonumber\\
	&\quad\times \frac{k\left(\frac{\lambda}{\gamma}x^{\frac{1}{\beta}}+1\right)^\beta\left[\frac{\lambda }{\gamma}\left(1-\frac{\alpha}{\beta}\right)x^{\frac{1}{\beta}}+1\right]-\frac{\alpha}{\beta}c^{1-\beta}(c-1)^{\alpha-1}\frac{\lambda}{\gamma}\left(1-\frac{\alpha}{\beta}\right)x^{\frac{1}{\beta}+1}}{\left(\frac{\lambda}{\gamma}x^{\frac{1}{\beta}}+1\right)^{\beta-\alpha+1}}.
	\label{eq:p1eq}
	\end{align}
	In case (2)(a), $p_2^*=\frac{R+\Psi}{\gamma}(x_2^*)^{1/\beta}$ where $x_2^*$ is the unique solution to the equation
	\begin{align}
	c^{1-\beta}(c-1)^{\alpha-1}\left(x^{-\frac{1}{\beta}}+\frac{\lambda \alpha}{\gamma\beta}\right)-\frac{k\lambda}{\gamma} \left(x^{-\frac{1}{\beta}}+\frac{\lambda}{\gamma}\right)^{\beta}=0.
	\label{eq:p2eq}
	\end{align}
	
	\label{prop:entry}
\end{proposition}

In the special case of $\beta=1$ such that the asset has zero drift, the results are slightly different from those in Proposition \ref{prop:entry}.

\begin{proposition}
	
	Suppose $\alpha< \beta= 1$. For the entry problem \eqref{eq:entry}:
	\begin{enumerate}[wide, labelwidth=!, labelindent=0pt, label={(\arabic*)}]
		\item If $\frac{\lambda}{\gamma}\leq \frac{\alpha}{k}(c-1)^{\alpha-1}$, there exists $p_1^{*}\in[0,\infty)$ such that the agent will enter the trade when the asset price is at or above $p_1^*$. The value function is given by \eqref{eq:exitvalfun1} on setting $\beta=1$. 
		\item If $\frac{\alpha}{k}(c-1)^{\alpha-1}<\frac{\lambda}{\gamma}<\frac{(c-1)^{\alpha-1}}{k}$, there exists a constant $C\in(0,\infty)$ independent of $\Psi$ and $R$ such that:
		\begin{enumerate}[wide, labelwidth=!, labelindent=0pt]
			\item If $\Psi< CR$, there exists $0\leq p_1^{*}<p_2^{*}<\infty$ such that the agent will enter the trade when the asset price is within the interval $[p_1^{*},p_2^{*}]$. The value function is given by \eqref{eq:exitvalfun2a} on setting $\beta=1$. 
			\item If $\Psi\geq CR$, the agent will never enter the trade. The value function is $V_2(p)=-kR^{\alpha}$.
		\end{enumerate}
		\item If $\frac{\lambda}{\gamma}\geq \frac{(c-1)^{\alpha-1}}{k}$, the agent will never enter the trade. The value function is $V_2(p)=-kR^{\alpha}$.
	\end{enumerate}
	Moreover, $p_1^*$ has the same form in Proposition \ref{prop:entry} on setting $\beta=1$, and the expression of $p_2^*$ is available in close-form which is
	\begin{align*}
	p_2^{*}=\frac{\Bigl(R+\Psi\Bigr)\Bigl[(c-1)^{\alpha-1}-k\lambda/\gamma\Bigr]}{\lambda\Bigl[k\lambda/\gamma-\alpha(c-1)^{\alpha-1}\Bigr]}.
	\end{align*}
	\label{prop:entryspecial}
\end{proposition}

{\cred
\begin{remark}
For certain types of agents such as retail investors, fixed transaction costs are typically insignificant such that their trading behaviors might be better described by the specialization of $\Psi=0$. In this case, Proposition \ref{prop:entry} and \ref{prop:entryspecial} can be simplified. See Corollary \ref{cor:nofix} and the discussion thereafter. 
\end{remark}}

The value function of the entry problem is characterized by the smallest concave majorant to the payoff function defined in \eqref{eq:g2}. Indeed, the function $v_1$ defined in \eqref{eq:g2} is simply the scaled value function of the exit problem such that $v_1(\theta)=V_1(\theta^{1/\beta}; \lambda \theta^{1/\beta}+\Psi+R)$, as discussed in Section \ref{sect:entry}. At a mathematical level, the various cases arising in Proposition \ref{prop:entry} and \ref{prop:entryspecial} are due to the different possible shapes of $v_1$ under different combinations of parameters. Some illustrations are given in Figure \ref{fig:valfun}.

Economically, the optimal entry strategy crucially depends on the level of transaction costs relative to the market and preference parameters. A fixed market entry fee in general discourages trading when the asset price is low. Paying a flat fee of \$10 to purchase an asset at \$20 is much less appealing compared to the case that the asset is trading at \$1000, because in the former case the asset has to increase in value by 50\% just for breakeven against the fixed transaction fee paid. Meanwhile, proportional transaction costs are the most significant for asset trading at high nominal price. A 10\% transaction fee charged on a million worth of property is much more expensive in monetary terms relative to the same percentage fee charged on a penny stock.

In case (1) of both Proposition \ref{prop:entry} and \ref{prop:entryspecial}, the proportional transaction costs are relatively low. Hence the agent does not mind purchasing the asset at a high nominal price. He will just avoid purchasing the asset when its price is very low due to the consideration of fixed transaction costs and therefore the purchase region is in form of $[p_1^*,\infty)$. 

In case (2)(a), proportional transaction costs start becoming significant. On the one hand, the agent avoids initiating the trade when the asset price is too low since the fixed entry cost will be too large relative to the size of the trade. On the other hand, the agent does not want to trade an expensive asset when the proportional costs are large. Upon balancing these two factors, the agent will wait when asset price is either too low or too high, and will only purchase the asset when the price first enters an interval $[p_1^*, p_2^*]$. A very interesting feature of the optimal entry strategy is that the waiting region here is disconnected. 

In case (2)(b) of Proposition \ref{prop:entry}, or case (2)(b) and (3) of Proposition \ref{prop:entryspecial}, the overall level of transaction costs is too high and hence the agent is discouraged from entering the trade in the first place. The key difference between Proposition \ref{prop:entry} and \ref{prop:entryspecial} is that when the asset has a strictly positive drift ($\beta<1\iff \mu>0$), one must impose a strictly positive fixed entry cost in order to stop the agent from trading at all price levels (if $\Psi=0$, then either case (1) or (2)(a) in Proposition \ref{prop:entry} applies in which case the agent is willing to enter the trade at a certain price level). When the asset is a statistically fair gamble ($\beta=1\iff \mu=0$), then a high proportional transaction cost alone is sufficient to discourage the agent from trading. It is interesting to note that the trading decision also depends on the agent's aspiration level $R$. Comparing case (2)(a) and case (2)(b) in Proposition \ref{prop:entry} and \ref{prop:entryspecial}, a low value of $R$ will more often lead to the ``no trading'' case. The economic interpretation is that an agent with low aspiration level (i.e. a low target benchmark) is less likely to participate trading, especially when the (proportional) costs of trading are high. In Section \ref{sect:endoR}, we briefly discuss how the aspiration level $R$ may be endogenized.

\begin{figure}[!htbp]
	\captionsetup[subfigure]{width=0.45\textwidth}
	\centering
	\subcaptionbox{Case (1) with: $\lambda=1.01$, $\gamma=0.99$, $\Psi=1$.\label{fig:case1}}{\includegraphics[scale =0.38] {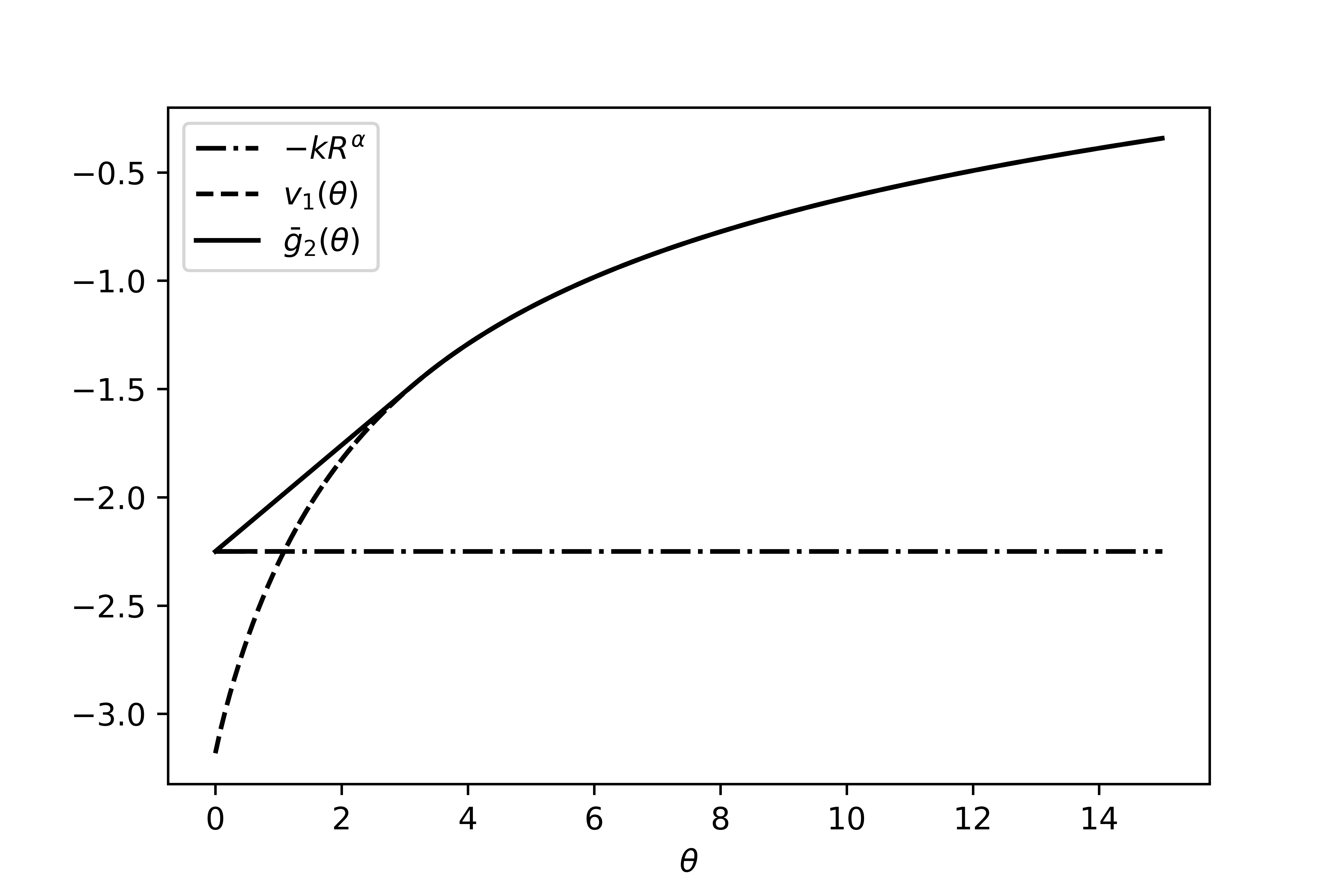}}
	\subcaptionbox{Case (2)(a) with: $\lambda=1.1$, $\gamma=0.9$, $\Psi=1$.\label{fig:case2}}{\includegraphics[scale =0.38]{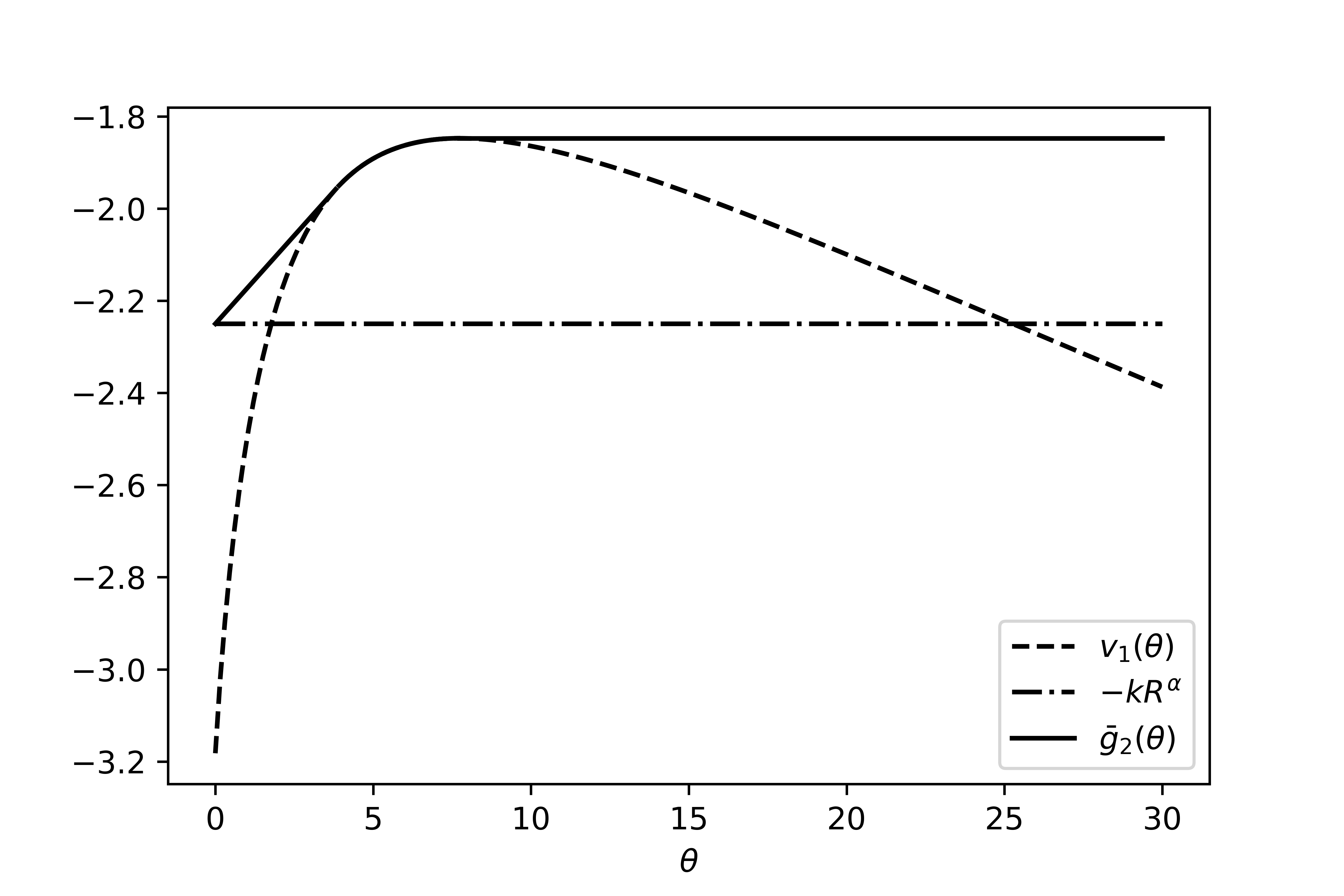}}
	\subcaptionbox{Case (2)(b) with: $\lambda=1.1$, $\gamma=0.9$, $\Psi=2.5$.\label{fig:case3}}{\includegraphics[scale =0.38]{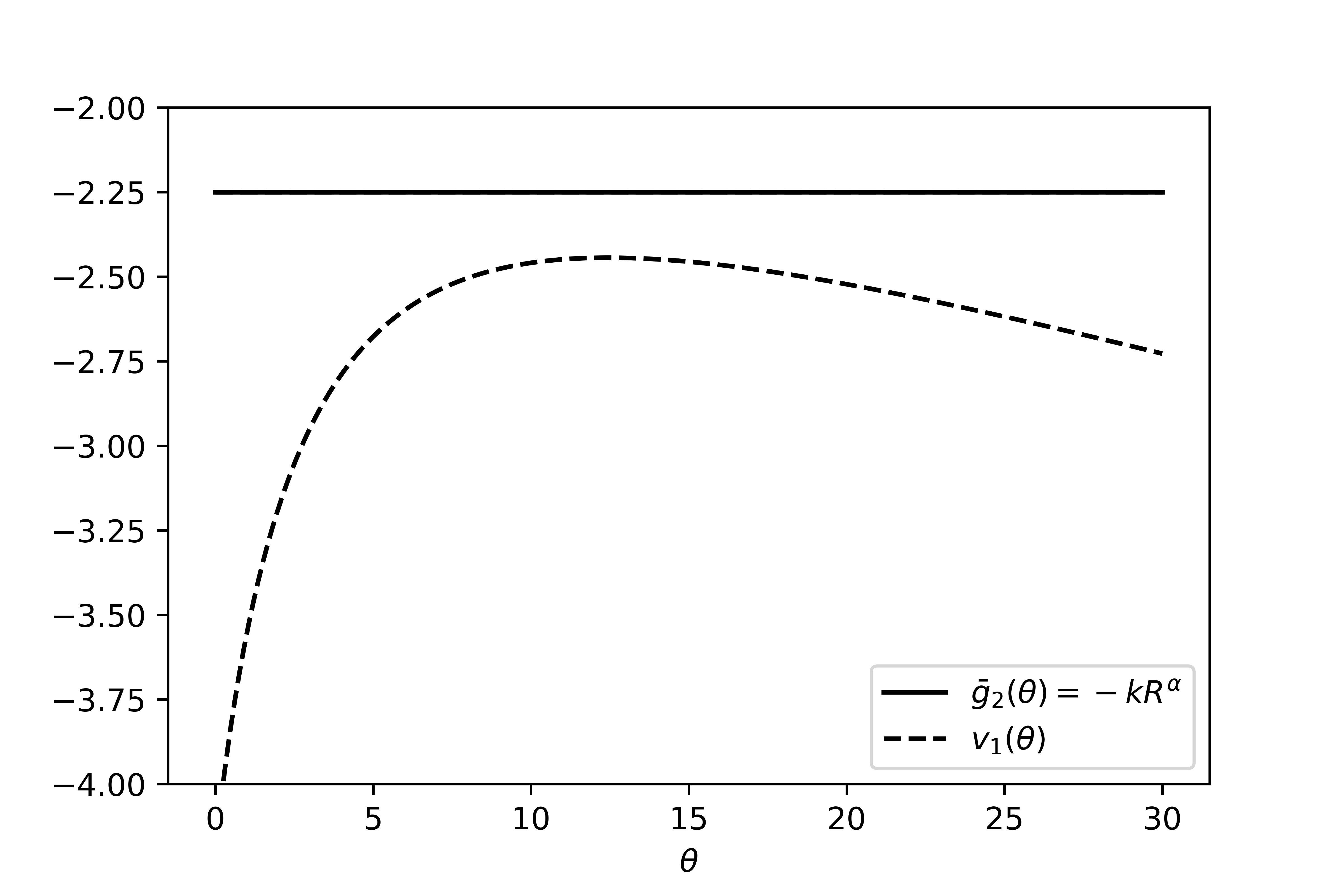}}
	
	\caption{The scaled value functions of the entry problem under different cases of Proposition \ref{prop:entry}. Base parameters used are: $\alpha=0.5$, $k=2.25$, $R=1$, $\beta=0.85$.}
	\label{fig:valfun}
\end{figure}

When viewed in conjunction with the results of well-posedness (Lemma \ref{lem:wellpose}) and the optimal exit strategy (Lemma \ref{lem:henderson}), our model can encapsulate many styles of trading behaviors. First, if $\beta\leq 0$ or $\beta<\alpha<1$ such that the whole problem is ill-posed, then we have already shown that a simple ``buy-and-hold'' strategy in form of \eqref{eq:buy_and_hold} is optimal in the sense that the attained utility level can be arbitrarily high.

In case (1) or (2)(a) of Proposition \ref{prop:entry} and \ref{prop:entryspecial}, if the asset price starts below $p_1^*$ at time zero, then the agent will purchase the asset when its price level increases to $p_1^*$. Note that, similar to Remark \ref{remark:bm}, the price process $P$ may not reach a fixed level $p_1^*>P_0$ in finite time. In this case the entry strategy will not be executed and the payoff to the agent is zero. But otherwise if a purchase is realized, then at the time of purchase the reference point is set as $H=\lambda p_1^*+\Psi+R$. Then due to Lemma \ref{lem:henderson}, the agent is looking to sell this asset later when its price level further increases to $\frac{cH}{\gamma}=\frac{c(\lambda p_1^*+\Psi+R)}{\gamma}$. Since $c>1$, $\gamma\leq 1$ and $\lambda\geq 1$, it is clear that the target sale level $\frac{c(\lambda p_1^*+\Psi+R)}{\gamma}$ is strictly larger than $p_1^*$. This trading rule is thus a momentum strategy in form of ``buy high and sell higher''.

If the asset price starts above $p_2^*$ at time zero in case (2)(a), then the agent will buy the asset when its price level drops to $p_2^*$ and later to sell the asset when it increases to $\frac{c(\lambda p_2^*+\Psi+R)}{\gamma}>p_2^*$. This is a counter-trend trading strategy in form of ``buy low sell high''.

Finally, in the high transaction cost cases (case (2)(b) of Proposition \ref{prop:entry}, and case (2)(b) or (3) of Proposition \ref{prop:entryspecial}) the agent will never participate in trading at any asset price level. 

The various cases above are generated by different level of transaction costs relative to the other model parameters. The following two corollaries further highlight the role of transaction costs in relationship to the optimal trading strategies.  
\begin{corollary}
	If $\lambda=\gamma=1$, the agent will purchase the asset when its price level is at or above $p_1^*$ for some $p_1^*\in[0,\infty)$.
	\label{cor:noprop}
\end{corollary}
\begin{proof}
	The result will follow if we can show that $[\frac{\alpha}{\beta k}c^{1-\beta}(c-1)^{\alpha-1}]^{\frac{1}{\beta}}>1$ such that case (1) of Proposition \ref{prop:entry} and \ref{prop:entryspecial} always applies when $\lambda=\gamma=1$. The required inequality is
	\begin{align*}
	\left[\frac{\alpha}{\beta k}c^{1-\beta}(c-1)^{\alpha-1}\right]^{\frac{1}{\beta}}>1 &\iff \frac{\alpha}{\beta }c^{1-\beta}(c-1)^{\alpha-1}>k \\
	&\iff  \frac{\alpha}{\beta }c^{1-\beta}(c-1)^{\alpha-1}>\frac{\alpha}{\beta} c(c-1)^{\alpha-1}-(c-1)^{\alpha} \\
	&\iff (c-1)^{\alpha}- \frac{\alpha}{\beta }c(c-1)^{\alpha-1}(1-c^{-\beta})>0
	\end{align*}
	where we have used \eqref{eq:eqC}. Using simple calculus we can show that $$F(x):=(x-1)^{\alpha}- \frac{\alpha}{\beta }x(x-1)^{\alpha-1}(1-x^{-\beta})>0$$ for all $x>1$. This concludes the proof. \qed
\end{proof}

\begin{corollary}
	Under the parameter combinations that $p_1^*$ is well-defined, if $\Psi=0$ then we have $p_1^*=0$.
	\label{cor:nofix}
\end{corollary}
\begin{proof}
	This follows immediately from Proposition \ref{prop:entry} and \ref{prop:entryspecial} by observing that $x=0$ is the solution to \eqref{eq:p1eq} when $\Psi=0$. \qed
\end{proof}

From Corollary \ref{cor:noprop}, if there is no proportional transaction cost then the agent does not care about entering the trade at a high nominal price level because he no longer needs to worry about the large magnitude of trading fee arising from the proportional nature of the transaction costs. Hence ``buy low sell high'' will not be observed as an optimal strategy. Similarly, Corollary \ref{cor:nofix} suggests that in absence of fixed market entry fee the agent is happy to purchase an asset of any arbitrarily low price (in the non-degenerate case) since now he does not need to take the size of the trade into account against any fixed cost for breakeven consideration. Thus ``buy high sell higher'' will not be an optimal strategy in this special case.

{\cb
\begin{remark}
If we further assume $R=\Psi=0$ (where the main results in Proposition \ref{prop:entry} and \ref{prop:entryspecial} still hold except \eqref{eq:g2} will have a different and simpler expression), then only case (1) or case (2)(b) of Proposition \ref{prop:entry} can arise under $\beta<1$. Alternatively, only case (1), case (2)(b) or case (3) of Proposition \ref{prop:entryspecial} can arise under $\beta=1$. But $p_1^*=0$ when $\Psi=0$ as in Corollary \ref{cor:nofix}. The entry behavior is thus trivial where the agent either is willing to purchase the asset at any price level if $\frac{\lambda}{\gamma}\leq[\frac{\alpha}{\beta k}c^{1-\beta}(c-1)^{\alpha-1}]^{\frac{1}{\beta}}$, or never enters the trade if  $\frac{\lambda}{\gamma}>[\frac{\alpha}{\beta k}c^{1-\beta}(c-1)^{\alpha-1}]^{\frac{1}{\beta}}$. From the modeling perspective, the constant $R+\Psi$ is a crucial component of the endogenized reference point which enables the model to produce non-trivial purchase behaviors. See Section \ref{sect:realization_uti} as well.
\end{remark}
}

The decomposition of the original purchase-and-sale problem \eqref{eq:valfun} into two sub-problems of optimal exit and optimal entry is based on some economic heuristics described in Section \ref{sect:exit} and \ref{sect:entry}. To close this section, we formally show that the value function of the entry problem in Proposition \ref{prop:entry} and \ref{prop:entryspecial} indeed corresponds to the value function of the sequential optimal stopping problem \eqref{eq:valfun}.
\begin{theorem}
	Recall the definition of $V_2$ which is defined in Proposition \ref{prop:entry} and \ref{prop:entryspecial} as the value function of the entry problem \eqref{eq:entry}. We have $\mathcal{V}(p)=V_2(p)$ where $\mathcal{V}$ is the value function of the sequential optimal stopping problem \eqref{eq:valfun}. The optimal purchase and sale rules are given by
	\begin{align}
	\begin{cases}
	\tau^*:=\inf\Bigl\{t\geq 0: V_2(P_t)=G_2(P_t)\Bigr\};&\\
	\nu^*:=\inf\left\{t\geq \tau^*: P_t \geq \frac{c}{\gamma}(\lambda P_{\tau^*}+\Psi+R)\right\}.
	\end{cases}
	\label{eq:optstoptime}
	\end{align}
	
\end{theorem}

\begin{proof}
Starting from \eqref{eq:obj}, we have for any $\tau,\nu\in\mathcal{T}$ with $\tau\leq \nu$ that
\begin{align*}
&J(p;\tau,\nu)\\
&=\mathbb{E}\left[U\left(\gamma P_{\nu} {\mathbbm 1}_{\{\tau<\infty,\nu<\infty\}}-\left(\lambda P_{\tau}+\Psi\right) {\mathbbm 1}_{\{\tau<\infty\}}-R\right)\Bigl | P_0=p\right]\\
&\leq \mathbb{E}\left[U\left(\left(\gamma P_{\nu}-\lambda P_{\tau}-\Psi\right) {\mathbbm 1}_{\{\tau<\infty\}}-R\right)\Bigl | P_0=p\right]\\
&=\mathbb{P}[\tau<\infty|P_0=p]\mathbb{E}\left[U\left(\gamma P_{\nu}-\lambda P_{\tau}-\Psi -R\right)\Bigl | P_0=p,\tau<\infty\right]\\
&\qquad+\mathbb{P}[\tau=\infty|P_0=p]U(-R) \\
&=\mathbb{E}\left\{\mathbb{E}\left[U\left(\gamma P_{\nu}-\lambda P_{\tau}-\Psi-R\right)\Bigl | P_\tau,\tau<\infty\right] \Bigl| P_0=p,\tau<\infty\right\}\\
&\qquad\times \mathbb{P}[\tau<\infty|P_0=p]+\mathbb{P}[\tau=\infty|P_0=p]U(-R) \\
&= \mathbb{E}\left\{\mathbb{E}\left[G_1(\gamma P_{\nu};\lambda P_{\tau}+\Psi+R)\Bigl | P_\tau,\tau<\infty\right] \Bigl| P_0=p,\tau<\infty\right\}\\
&\qquad \times \mathbb{P}[\tau<\infty|P_0=p]+\mathbb{P}[\tau=\infty|P_0=p]U(-R)\\
&\leq\mathbb{E}\left\{\sup_{\nu\in\mathcal{T}:\nu\geq \tau}\mathbb{E}\left[G_1(\gamma P_{\nu};\lambda P_{\tau}+\Psi+R)\Bigl | P_\tau,\tau<\infty\right] \Bigl| P_0=p,\tau<\infty\right\}\\
&\qquad\times \mathbb{P}[\tau<\infty|P_0=p] +\mathbb{P}[\tau=\infty|P_0=p]U(-R) .
\end{align*}
But using the Markovian structure of the problem, 
\begin{align*}
&\sup_{\nu\in\mathcal{T}:\nu\geq \tau}\mathbb{E}\left[G_1(\gamma P_{\nu};\lambda P_{\tau}+\Psi+R)\Bigl | P_\tau=s,\tau<\infty\right]\\
&=\sup_{\nu\in\mathcal{T}:\nu\geq 0}\mathbb{E}\left[G_1(\gamma P_{\nu};\lambda P_{0}+\Psi+R)\Bigl | P_0=s\right]=V_1(s;\lambda s+\Psi+R).
\end{align*}
Then we have
\begin{align*}
J(p;\tau,\nu)&\leq  \mathbb{E}\left[V_1(P_{\tau};\lambda P_{\tau}+\Psi+R) \Bigl | P_0=p,\tau<\infty\right]\\
&\qquad\times \mathbb{P}[\tau<\infty|P_0=p]+\mathbb{P}[\tau=\infty|P_0=p]U(-R)\\
&\leq \mathbb{E}\left[\max\left(V_1(P_{\tau};\lambda P_{\tau}+\Psi+R),U(-R)\right)\Bigl | P_0=p\right]\\
&\leq \sup_{\tau\in\mathcal{T}}\mathbb{E}\left[\max(V_1(P_{\tau};\lambda P_{\tau}+\Psi+R),U(-R) )\Bigl | P_0=0\right]\\
&=\sup_{\tau\in\mathcal{T}}\mathbb{E}[G_2(P_{\tau})|P_0=p]=V_2(p).
\end{align*}
Taking supremum on both sides leads to $\mathcal{V}(p)\leq V_2(p)$.

To complete the proof, it is sufficient to show $J(p;\tau^*,\nu^*)=V_2(p)$ where $\tau^*,\nu^*$ are defined in \eqref{eq:optstoptime}. This can be directly verified under the various cases covered in Proposition \ref{prop:entry} and \ref{prop:entryspecial} with different initial price level $p$.

As an example, we cover Case (2)(a) in Proposition \ref{prop:entry} and a level of $p$ such that $p<p_1^*$. We can deduce from the shape of $V_2$ in this case that $\tau^*=\inf\{t\geq 0: P_t\notin(0,p_1^*)\}$. Since $\beta>0$, there are two possibilities: the asset price reaches the purchase target level $p_1^*$ in a finite time which happens with probability $\mathbb{P}[\tau^*<\infty|P_0=p]=\frac{p^{\beta}}{(p_1^*)^\beta}$; or it never reaches $p_1^*$ in a finite time such that the agent never enters the trade where he faces a realized utility of $U(-R)$, and this happens with probability $\mathbb{P}[\tau^*=\infty|P_0=p]=1-\frac{p^{\beta}}{(p_1^*)^\beta}$. In the first scenario, after the agent purchases the asset at price $p_1^*$ he will sell the asset when its price further increases to $x^*:=\frac{c}{\gamma}(\lambda p_1^*+\Psi+R)$. The conditional probability of a successful sale is $\mathbb{P}[\nu^*<\infty|P_{\tau^*}=p_1^*]=\frac{(p_1^*)^{\beta}}{(x^*)^\beta}$ where the realized utility is $U(\gamma x^*-\lambda p_1^*-\Psi-R)$. Otherwise, the target sale level is never reached with conditional probability $\mathbb{P}[\nu^*=\infty|P_{\tau^*}=p_1^*]=1-\frac{(p_1^*)^{\beta}}{(x^*)^\beta}$ where the realized utility becomes $U(-\lambda p_1^*-\Psi-R)$. Then we can directly compute
\begin{align*}
&J(p;\tau^*,\nu^*)\\
&=\mathbb{E}\left[U\left(\gamma P_{\nu^*} {\mathbbm 1}_{\{\tau^*<\infty,\nu^*<\infty\}}-\left(\lambda P_{\tau^*}+\Psi\right) {\mathbbm 1}_{\{\tau^*<\infty\}}-R\right)\Bigl | P_0=p\right]\\
&=\mathbb{P}[\tau^*=\infty|P_0=p]U(-R)\\
&\qquad+\mathbb{P}[\tau^*<\infty,\nu^*=\infty|P_0=p]U(-\lambda p_1^*-\Psi-R)\\
&\qquad +\mathbb{P}[\tau^*<\infty,\nu^*<\infty|P_0=p]U(\gamma x^*-\lambda p_1^*-\Psi-R)\\
&=-\frac{(p_1^*)^\beta-p^{\beta}}{(p_1^*)^\beta}kR^{\alpha}-k\left(\frac{p^{\beta}}{(p_1^*)^\beta}\right)\left(\frac{(x^*)^\beta-(p_1^*)^{\beta}}{(x^*)^\beta}\right)(\lambda p_1^*+\Psi+R)^{\alpha}\\
&\qquad +\left(\frac{p^{\beta}}{(p_1^*)^\beta}\right)\left(\frac{(p_1^*)^{\beta}}{(x^*)^\beta}\right)(\gamma x^*-\lambda p_1^*-\Psi-R)^{\alpha}\\
&=\frac{p^{\beta}}{(p_1^*)^\beta}\left[\frac{(p_1^*)^{\beta}}{(x^*)^\beta}(\gamma x^*-\lambda p_1^*-\Psi-R)^{\alpha}-k\frac{(x^*)^\beta-(p_1^*)^{\beta}}{(x^*)^\beta}(\lambda p_1^*+\Psi+R)^{\alpha}\right]\\
&\qquad -\frac{(p_1^*)^\beta-p^{\beta}}{(p_1^*)^\beta}kR^{\alpha}\\
&=\frac{p^{\beta}}{(p_1^*)^\beta}\left[(\gamma p_1^*)^{\beta}c^{-\beta}((c-1)^{\alpha}+k)(\lambda p_1^*+\Psi+R)^{\alpha-\beta}-k(\lambda p_1^*+\Psi+R)^{\alpha}\right]\\
&\qquad -\frac{(p_1^*)^\beta-p^{\beta}}{(p_1^*)^\beta}kR^{\alpha}\\
&=\frac{p^{\beta}}{(p_1^*)^\beta}\left[(\gamma p_1^*)^{\beta}\frac{\alpha}{\beta}c^{1-\beta}(c-1)^{\alpha-1}(\lambda p_1^*+\Psi+R)^{\alpha-\beta}-k(\lambda p_1^*+\Psi+R)^{\alpha}\right]\\
&\qquad -\frac{(p_1^*)^\beta-p^{\beta}}{(p_1^*)^\beta}kR^{\alpha}\\
&=\frac{p^{\beta}}{(p_1^*)^\beta}\Bigl[(\gamma p_1^*)^{\beta}\frac{\alpha}{\beta}c^{1-\beta}(c-1)^{\alpha-1}(\lambda p_1^*+\Psi+R)^{\alpha-\beta}\\
&\qquad-k(\lambda p_1^*+\Psi+R)^{\alpha}+kR^{\alpha}\Bigr]-kR^{\alpha}\\
&=V_2(p)
\end{align*}
where we have used the definition that $x^*:=\frac{c}{\gamma}(\lambda p_1^*+\Psi+R)$ and the fact that $c$ is the solution to \eqref{eq:eqC}. The other cases can be handled similarly. \qed
\end{proof}

\section{Comparative statics of the optimal trading strategies}
\label{sect:compstat}

The critical trading boundaries in Proposition \ref{prop:entry} and \ref{prop:entryspecial}, although not being available in close-form in general, can be analyzed to shed lights on the comparative statics of the optimal trading strategies with respect to a few underlying model parameters. The proof of the following proposition is given in the appendix. 

\begin{proposition}
	Under the parameters combination such that $p_1^*$ and/or $p_2^*$ are well-defined. We have:
	\begin{enumerate}[wide, labelwidth=!, labelindent=0pt]
		\item $p_1^*$ is decreasing in $\gamma$ and increasing in $\Psi$.
		\item $p_2^*$ is decreasing in $\lambda$, increasing in $\gamma$ and increasing in $\Psi$.
	\end{enumerate}
	\label{prop:compstat}
\end{proposition}

Figure \ref{fig:compgamma} shows the critical purchase boundary $p_1^*$ and $p_2^*$ as $\gamma$ varies. For very large value of $\gamma$ such that the condition in case (1) of Proposition \ref{prop:entry} holds, the optimal strategy is to buy the asset when its price exceeds $p_1^*$ and that the agent is willing to enter the trade no matter how high the price is. Once $\gamma$ is smaller than a certain critical value (labeled by the vertical dotted line on the figure), parameters condition in case (2)(a) of Proposition \ref{prop:entry} applies. The optimal strategy now becomes to purchase the asset only when its price is within a bounded range $[p_1^*,p_2^*]$. As $\gamma$ further decreases, $p_1^*$ increases while $p_2^*$ decreases so that the purchase region $[p_1^*,p_2^*]$ shrinks. Once $\gamma$ reaches another critical value, $p_1^*$ and $p_2^*$ converge and the purchase region diminishes entirely. This corresponds to case (2)(b) of Proposition \ref{prop:entry} that the agent will not enter the trade at any price level. As a reminder, the constant $C$ in case (2) of Proposition \ref{prop:entry} and \ref{prop:entryspecial} depends on $\lambda$ and $\gamma$. Increasing $\frac{\lambda}{\gamma}$ will result in a switch from case (2)(a) to case (2)(b).

We do not mention in Proposition \ref{prop:compstat} the effect of $\lambda$ on $p_1^*$. While the example in Figure \ref{fig:complambda} shows that $p_1^*$ is increasing in $\lambda$, numerical results show that $p_1^*$ is not monotonic in $\lambda$ in general. See Figure \ref{fig:lam_countereg}. Hence, when viewed in conjunction with $p_2^*$ the purchase region $[p_1^*, p_2^*]$ does not necessarily shrink uniformly when proportional cost on purchase increases, i.e. the agent may not delay the purchase decision. Similar observations regarding potential non-monotonicity of trading decisions with respect to (proportional) transaction costs are made by Hobson, Tse and Zhu~\cite{hobson-tse-zhu19a, hobson-tse-zhu19b} in the context of portfolio optimization. 

\begin{figure}[!htbp]
	\captionsetup[subfigure]{width=0.45\textwidth}
	\centering
	\subcaptionbox{Comparative statics with respect to $\lambda$. The dotted vertical line marks the critical value of $\lambda$ above which $p_2^*$ exists.\label{fig:complambda}}{\includegraphics[scale =0.38] {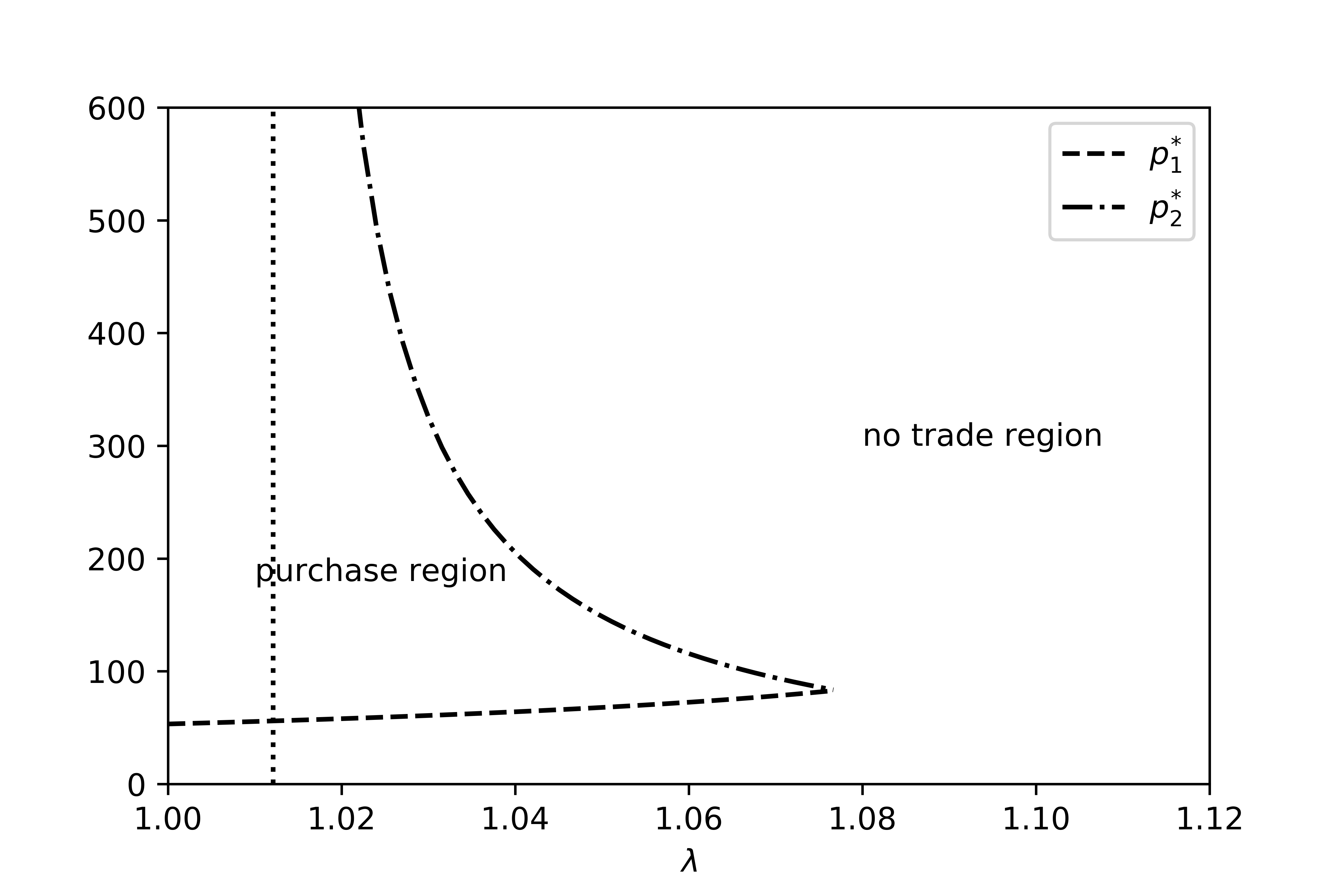}}
	\subcaptionbox{Comparative statics with respect to $\gamma$. The dotted vertical line marks the critical value of $\gamma$ below which $p_2^*$ exists.\label{fig:compgamma}}{\includegraphics[scale =0.38]{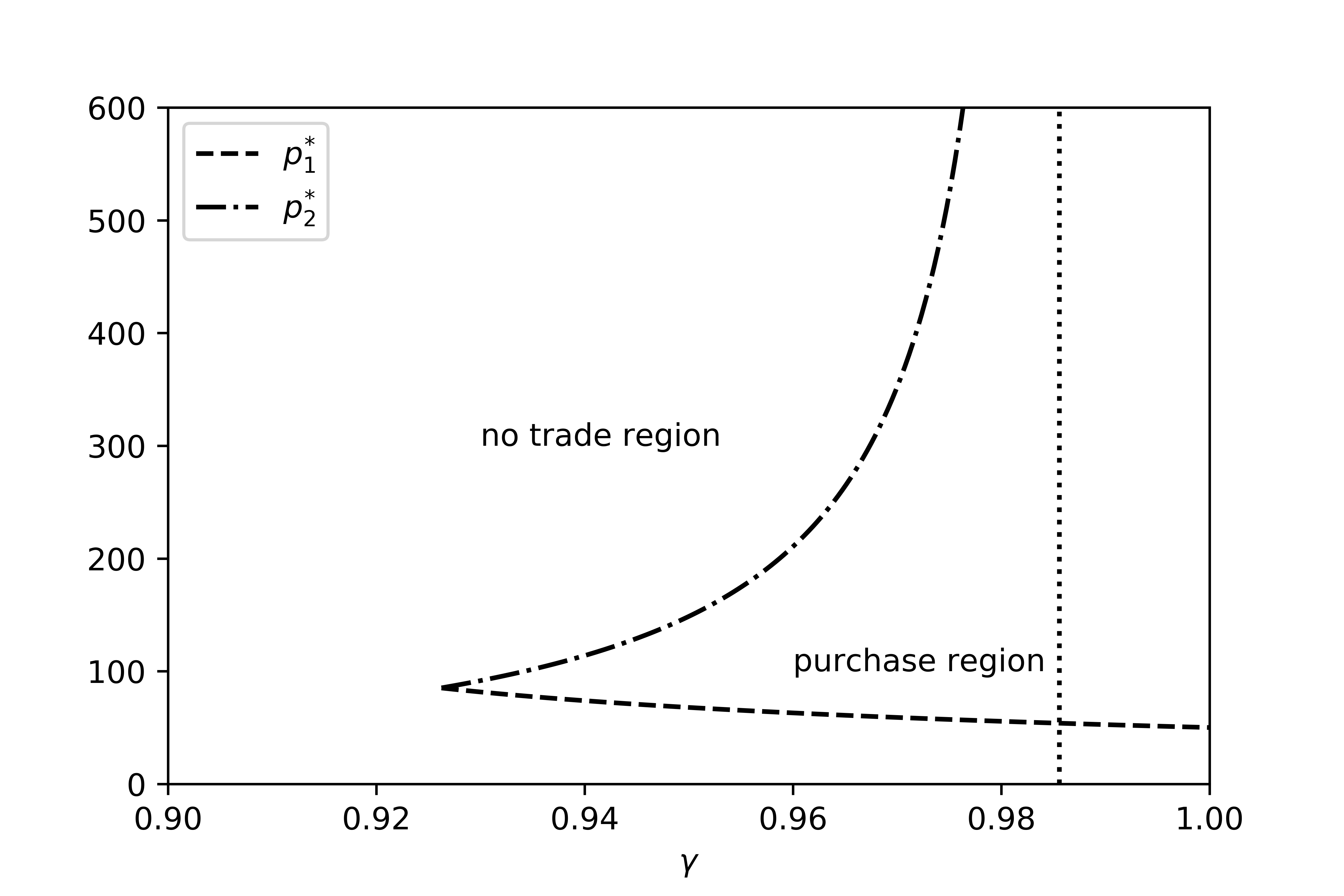}}
	\subcaptionbox{Comparative statics with respect to $\Psi$.\label{fig:comppsi}}{\includegraphics[scale =0.38]{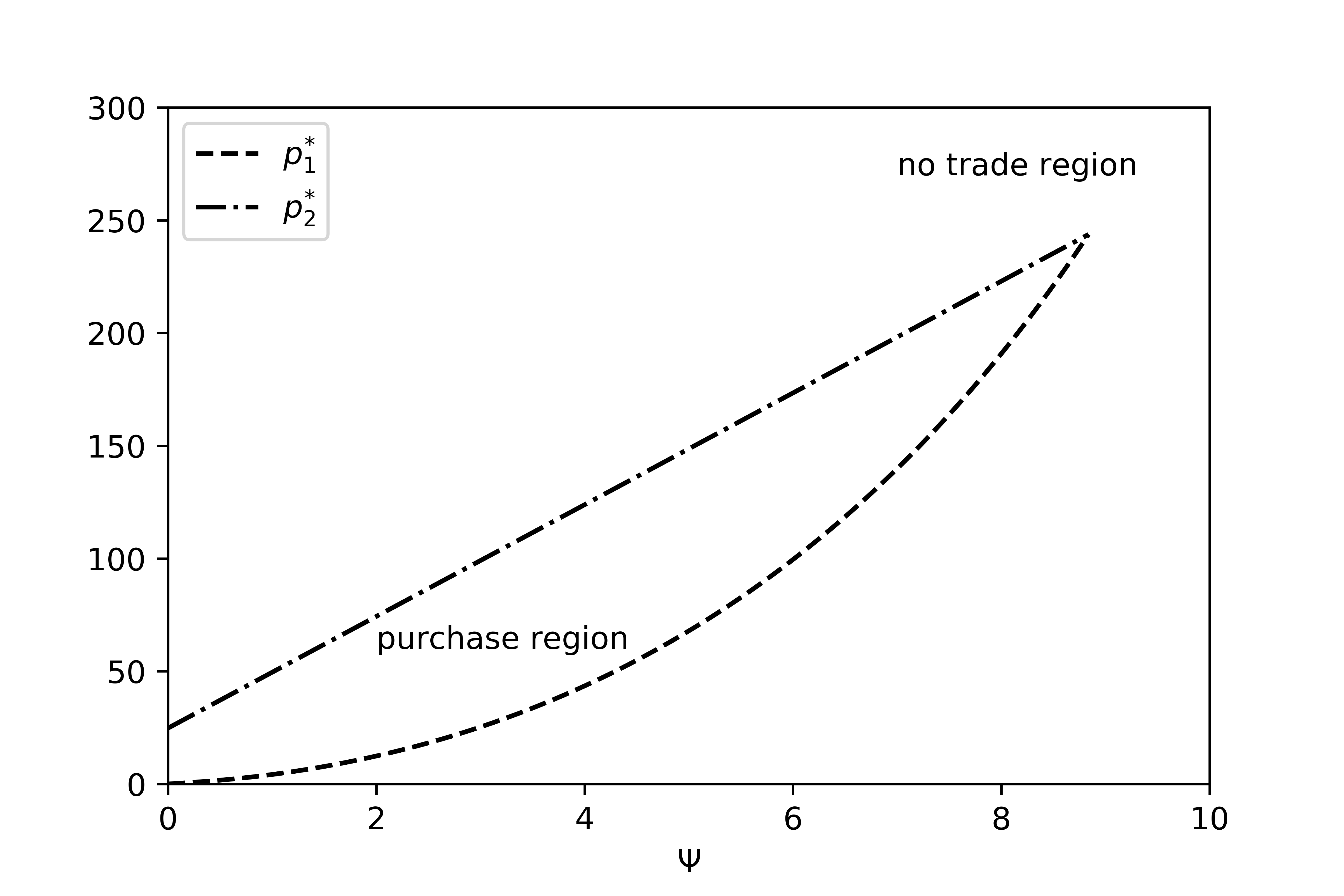}}
	
	\caption{Comparative statics of the optimal purchase boundaries $p_1^*$ and $p_2^*$. Base parameters used are: $\alpha=0.5$, $k=2.25$, $R=1$, $\beta=0.85$, $\lambda=1.05$, $\gamma=0.95$, $\Psi=5$.}
	\label{fig:compstat}
\end{figure}

\begin{figure}[!htbp]
	\centering
	\includegraphics[scale =0.4]{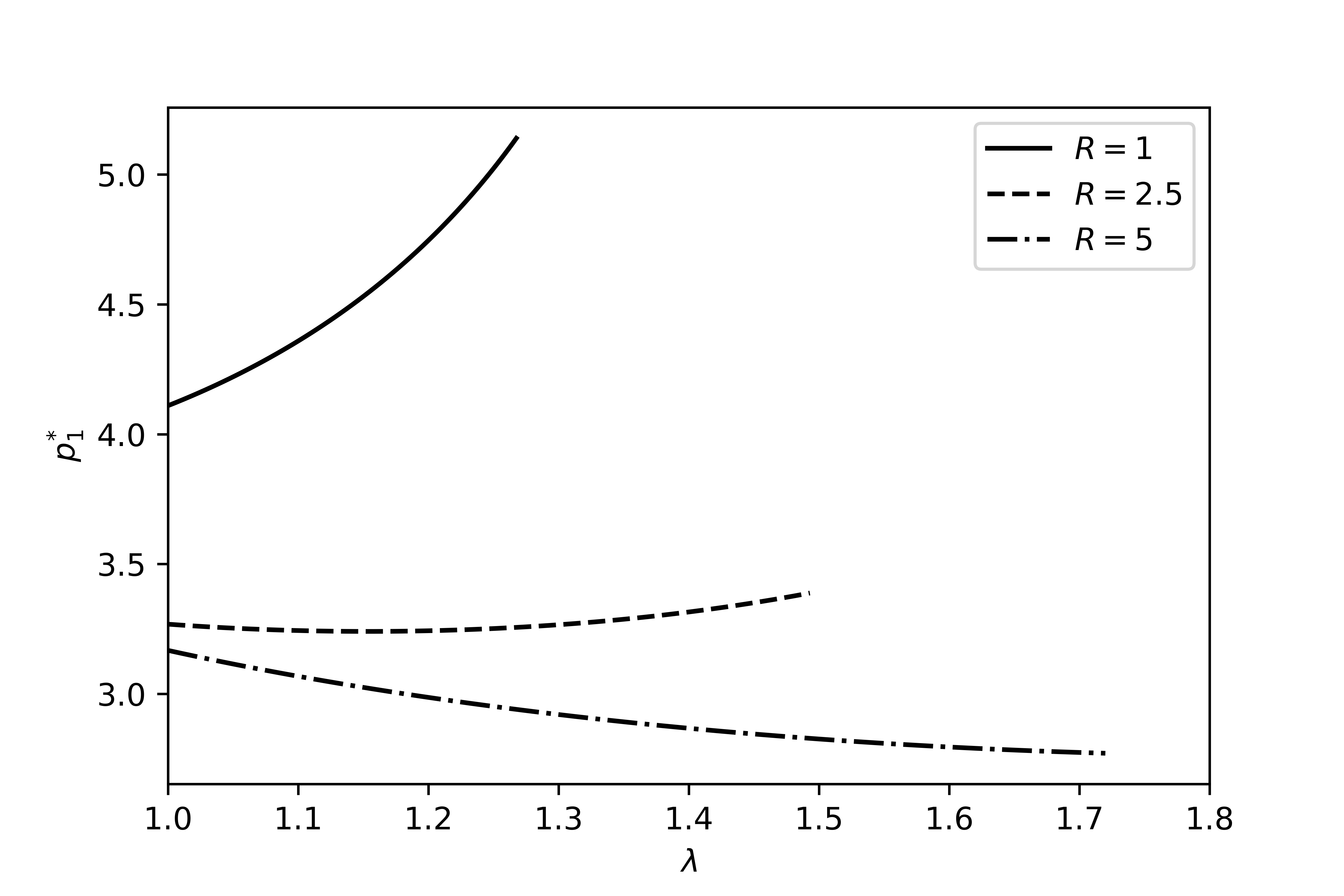}
	\caption{An illustration that $p_1^*$ the lower bound of the purchase region may not be monotonic with respect to $\lambda$. Base parameters used are: $\alpha=0.5$, $k=2.25$, $\beta=0.85$, $\lambda=1.05$, $\gamma=0.95$, $\Psi=1$.}
	\label{fig:lam_countereg}
\end{figure}

Similarly, we can also examine the impact of the fixed market entry cost on the purchase decision. As shown in Figure \ref{fig:comppsi}, $p_1^*$ and $p_2^*$ are both increasing in $\Psi$. The fact that $p_2^*$ is increasing in $\Psi$ is indeed somewhat counter-intuitive and it has a few interesting policy implications.
Suppose there is a market regulator who wants to discourage the agent from purchasing the asset (for example, as a mean to cool down a highly speculative real estate market). A natural measure to curb trading participation is to increase transaction costs. However, Figure \ref{fig:compstat} reveals that there is a subtle difference between the impact of proportional and fixed transaction cost on the agent's trading behavior. 

From Figure \ref{fig:compgamma}, the effect of increasing proportional transaction cost on sale (i.e. decreasing $\gamma$) is ``monotonic'' in terms of changing the trading decision of the agent. At any given current asset price level, decreasing $\gamma$ (while all other parameters are held fixed) can only take the agent from the purchase region to the no trade region. Increasing proportional transaction cost on sale can therefore unambiguously suppress the trading activities in the market. 

In contrast, the impact of the fixed market entry cost is somewhat unclear. Take Figure \ref{fig:comppsi} as an example and suppose the current price of the asset is \$100. If there is no fixed market entry fee initially (i.e. $\Psi=0$), the agent will not participate in trading as he is in the no trade region. However, a policy of increasing $\Psi$ from zero to \$4 will now put the agent in the purchase region such that he is willing to purchase the asset immediately (given that the current asset price stays the same at \$100). It is exactly opposite to the intended outcome of the market regulator because the increase in $\Psi$ actually encourages trading participation.

The rationale behind this phenomenon is as follows: the speculative agent is evaluating the trading opportunity by comparing the potential sale proceed against the reference point which is partly determined by the initial capital required to enter the trade, given by $\lambda p+\Psi$ if the purchase price is $p$. Increasing the fixed market entry fee $\Psi$ increases the total costs required to initiate the trade and it results in a higher effective reference point. However, under Prospect Theory framework the agent's risk attitude is tied to the level of the reference point. When the purchase price is kept as the same, a large $\Psi$ will put the agent deeper in the loss territory over which he becomes highly risk-seeking. Thus he will give a higher valuation to the opportunity to enter the speculative trade and hence is more prone to immediate trade participation.

Of course, increasing $\Psi$ will also decrease the potential profit of the trade which is economically unfavorable. As the fixed cost further increases, say from $\Psi=4$ to $\Psi=8$, the agent will eventually enter the no trade region again. Hence the precise effect of $\Psi$ on the trading decision is ambiguous governed by the two opposing forces of increasing agent's risk appetite versus decreasing profitability. When the economy is consisting of multiple agents with heterogeneous preferences, it is unclear whether increasing the fixed transaction costs can uniformly discourage trading participation for all agents. 

The non-monotonicity of $p_1^*$ with respect to $\lambda$ the proportional transaction cost on purchase also implies that an increase in $\lambda$ can potentially bring certain agents from the no trade region to the purchase region. The rationale is the same as the above that $\lambda$ partly determines the cost of purchase and in turn the reference point. Hence increasing $\lambda$ might actually make an agent find a speculative opportunity attractive.

{\cred
\section{Extensions}
\label{sect:extend}

In this section, we briefly discuss several variations of our baseline model.

\subsection{Risky asset with negative drift and voluntary stop-loss}
\label{sect:negative}

Among all the non-trivial strategies derived in our baseline setup, the agent will never voluntarily realize loss. This is not entirely realistic, and one way to enable the model to generate stop-loss behavior is to allow the excess return of the asset to be negative as inspired by Henderson~\cite{henderson12}. Negative excess return is equivalent to $\beta:=1-\frac{2\mu}{\sigma^2}>1$ in our setup. 

\begin{lemma}
Suppose the model parameters are such that $\beta>1$ instead. For the exit problem \eqref{eq:exit}, the agent will sell the asset when its price level first exits the interval $(\frac{c_1 H}{\gamma},\frac{c_2 H}{\gamma})$ where $0<c_1<1<c_2$ are the solutions to the system of equations $$\frac{k\alpha}{\beta}c_1^{1-\beta}(1-c_1)^{1-\alpha}=\frac{\alpha}{\beta}c_2^{1-\beta}(c_2-1)^{\alpha-1}=\frac{(c_2-1)^{\alpha}+k(1-c_1)^{\alpha}}{c_2^{\beta}-c_1^{\beta}}.$$ The value function is
\begin{align}
V_1(p;H)=
\begin{cases}
-k(H-\gamma p)^{\alpha},& p<\frac{c_1 H}{\gamma};\\
-kH^{\alpha}(1-c_1)^{\alpha}&\\
\qquad+\frac{H^{\alpha-\beta}\bigl[(c_2-1)^{\alpha}+k(1-c_1)^{\alpha}\bigr]\bigl[(\gamma p)^{\beta}-c_1^{\beta} H^{\beta}\bigr]}{c_2^{\beta}-c_1^{\beta}},& \frac{c_1 H}{\gamma}\leq p\leq \frac{c_2 H}{\gamma};\\
(\gamma p-H)^{\alpha},& p>\frac{c_2 H}{\gamma}.
\end{cases}
\label{eq:exitnegative}
\end{align}
\end{lemma}
\begin{proof}
It follows from a slight extension of Proposition 3 in \cite{henderson12}.\qed
\end{proof}

Given that purchase of the asset has occurred at some time $t=\tau$ which determines the reference level for the exit problem via $H=\lambda P_{\tau}+\Psi+R$, the agent is willing to impose a stop-loss level at $\frac{c_1 H}{\gamma}$ if the asset has a negative drift. Of course, it is not clear upfront whether the agent is willing to purchase an asset with negative drift in the first place. To understand the purchase behavior, one needs to solve the entry problem $$V_2(p):=\sup_{\tau\in\mathcal{T}}\mathbb{E}\Bigl[\max\bigl\{V_1(P_{\tau};\lambda P_{\tau}+\Psi+R), U(-R)\bigr\}\Bigl|P_0=p\Bigr]$$ where $V_1$ is now given by \eqref{eq:exitnegative}. Although the principle of martingale method still applies, it is now more challenging to analyze the problem thoroughly to explicitly characterize different possible shapes of the scaled payoff function under different combinations of model parameters. We hence opt to obtain numerical solutions by solving the underlying variational inequality $$\min\left(-\frac{\sigma^2 p^2}{2}V_2''(p)-\mu pV_2'(p),V_2(p)-G_2(p)\right)=0$$ by standard projected successive over-relaxation (PSOR) method and infer the optimal purchase (stopping) and no trade (continuation) region by numerically identifying the set $\{p\geq 0: V_2(p)=G_2(p)\}$ and $\{p\geq 0: V_2(p)>G_2(p)\}$ respectively.

\begin{figure}[!htbp]
	\captionsetup[subfigure]{width=0.45\textwidth}
	\centering
	\subcaptionbox{Optimal purchase (entry) boundaries as a function of $\beta$.\label{fig:negative_entry}}{\includegraphics[scale =0.38] {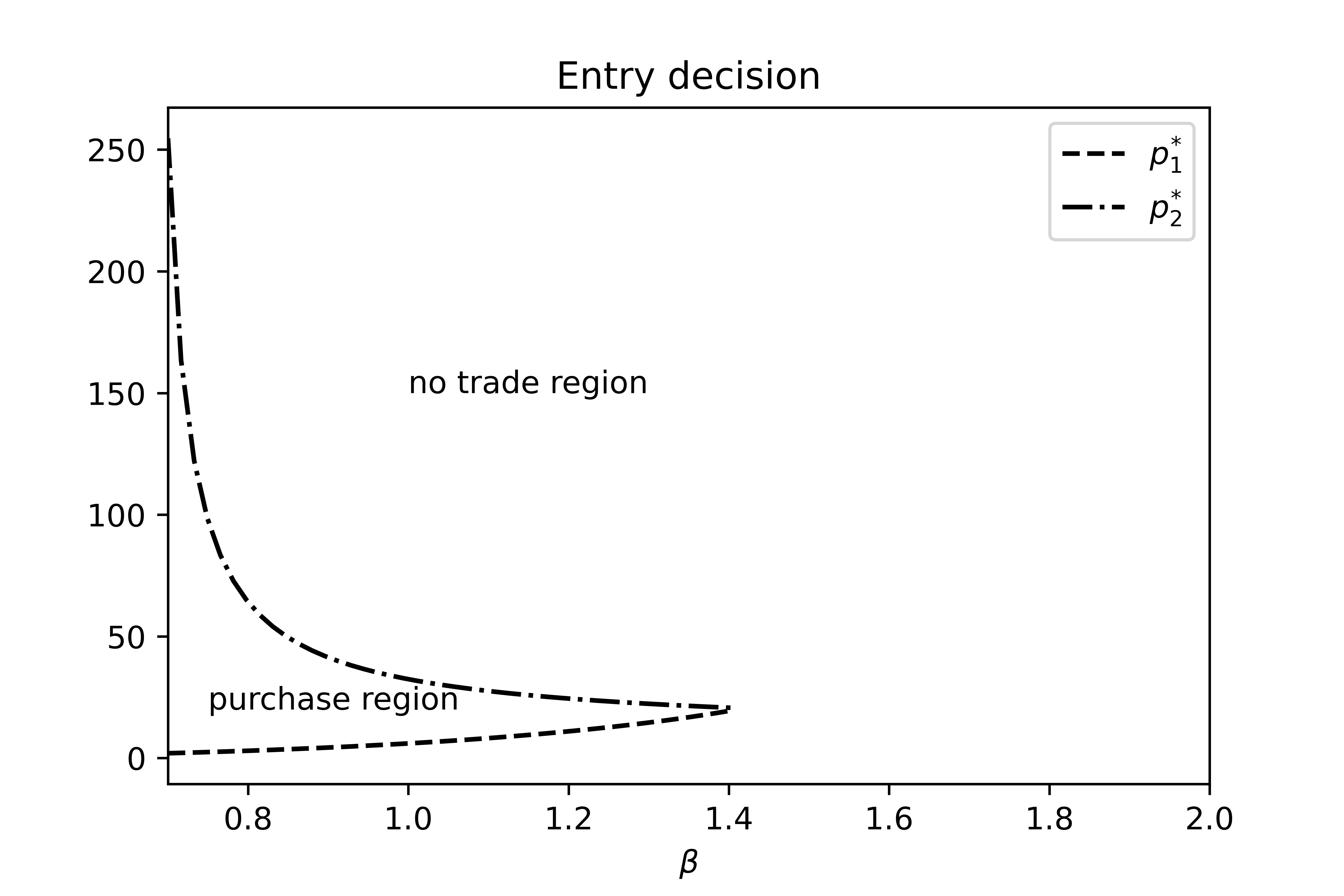}}
	\subcaptionbox{Optimal sale (exit) boundaries as a function of $\beta$.\label{fig:negative_exit}}{\includegraphics[scale =0.38]{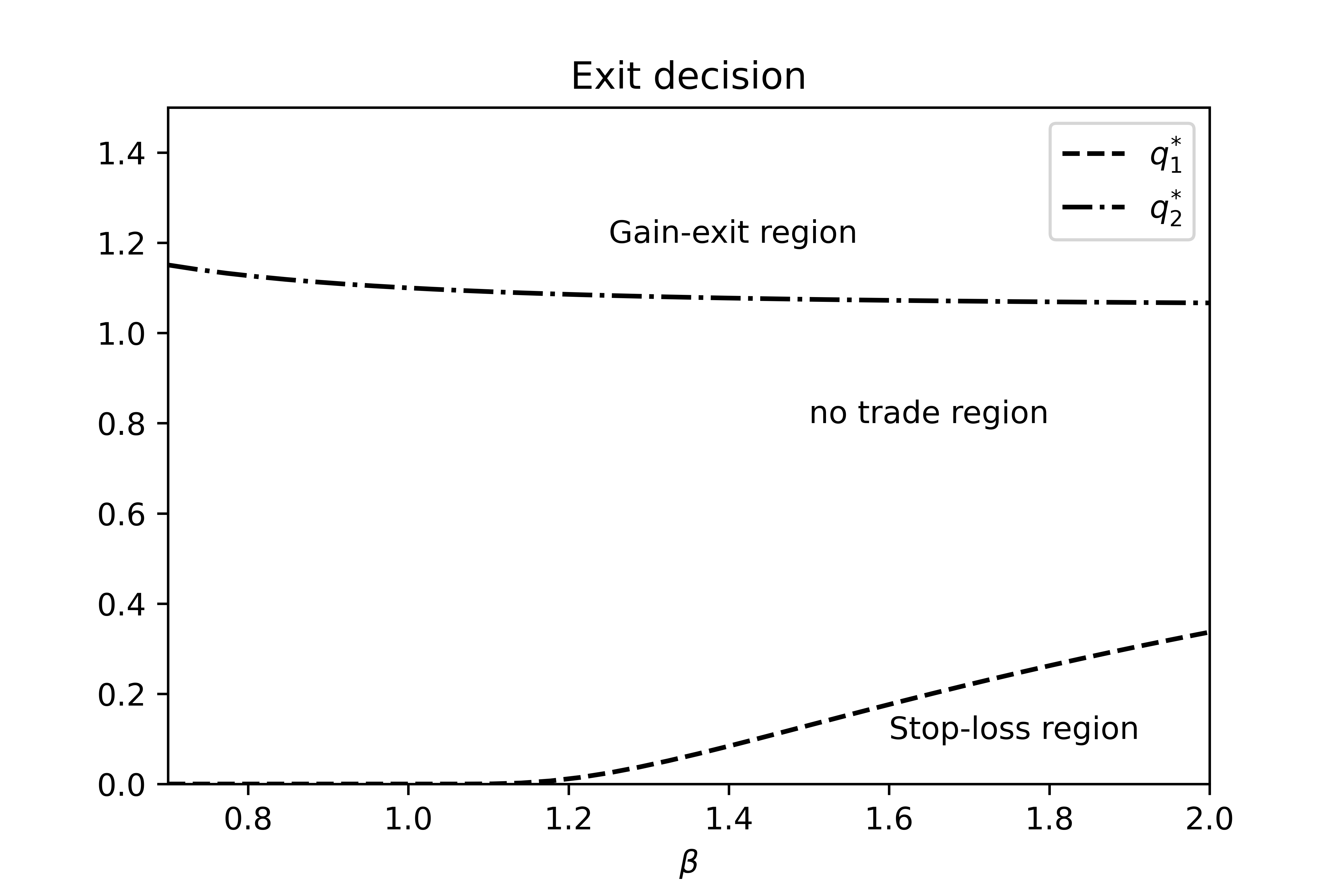}}
	
	\caption{Optimal purchase and sale decision as $\beta$ varies. $\beta< 1$ (resp. $\beta>1$) indicates that the assaet has a positive (negative) drift. Base parameters used are: $\alpha=0.5$, $k=2.25$, $R=1$, $\lambda=1.05$, $\gamma=0.95$, $\Psi=1$.}
	\label{fig:negative}
\end{figure}

Figure \ref{fig:negative_entry} shows the optimal entry decision as $\beta$ varies. When $\beta\leq 1$, the optimal purchase boundaries $p_1^*$ and $p_2^*$ are obtained semi-analytically as per Proposition \ref{prop:entry} and \ref{prop:entryspecial}, while the values under $\beta>1$ are obtained numerically by PSOR method. The agent is willing to purchase the asset if and only if the current price level is between $p_1^*$ and $p_2^*$. As the drift of the asset becomes more and more negative (i.e. $\beta$ increases), the purchase region $[p_1^*,p^*_2]$ shrinks and eventually vanishes when $\beta$ is around 1.4. Beyond this critical level of $\beta$, the agent will not purchase the asset at any price level because of its poor quality. 

In parallel, Figure \ref{fig:negative_exit} plots the optimal sale boundaries in form of $q^*_i:=\frac{c_i }{\gamma}$ such that if the asset has been purchased at level $p$, it will be sold when its price level leaves the interval $(H q_1^*,H q_2^*)$ where $H=\lambda p+\Psi+R$. If $\beta\leq 1$, it is never optimal to voluntarily realize loss which is equivalent to $q_1^*=0$. But for $\beta>1$, $q^*_1$ becomes strictly positive which represents a stop-loss level. Figure \ref{fig:negative} demonstrates that there exists some combinations of model parameters such that the agent is willing to purchase the asset at some price level and subsequently willing to liquidate the asset at loss. This happens when $\beta$ is between 1 and 1.4 in our numerical example. Explicitly characterizing the condition on the model parameters where this phenomenon occurs will be an interesting follow-up research question.

\subsection{Subjective discounting}
\label{sect:discount}

To investigate the effect of impatience on the optimal strategy, it is constructive to consider a version of the problem with discounting. Nonetheless, to the best of our knowledge there is no consensus in the literature regarding how discounting should be incorporated within a Prospect Theory framework with intertemporal cash flows. We briefly present two possible modeling choices.

\subsubsection{Profit-discounting}

The first idea is that the S-shaped utility function is applied to the net present value of the trading proceed, which we term as ``{\it profit-discounting}''. We consider the problem
\begin{align*}
\sup_{\tau,\nu\in\mathcal{T}:\tau\leq \nu}\mathbb{E}\left[U\left(\gamma e^{-\delta\nu}P_{\nu} {\mathbbm 1}_{\{\tau<\infty,\nu<\infty\}}-e^{-\delta\tau}\left(\lambda P_{\tau}+\Psi\right) {\mathbbm 1}_{\{\tau<\infty\}}-R\right)\Bigl | P_0=p\right]
\end{align*}
where $\delta>0$ is the subjective discount factor. If we further assume $\Psi=0$ (which is perhaps relevant in the context of retail trading where fixed transaction cost is insignificant), then this problem is just the same as the undiscounted one except the drift of the asset is now replaced by $\mu-\delta$ under the geometric Brownian motion assumption of the asset price. For small $\delta$ such that $\delta\leq \mu$, our baseline results under the standing assumption $\beta\leq 1$ apply. Otherwise when $\delta>\mu$ the analysis becomes similar to the one covered in Section \ref{sect:negative}. Increasing $\delta$ has the same effect of increasing $\beta$ where an endowed asset tends to be liquidated sooner (lower gain-exit level and possibly higher stop-loss level) while the purchase region shrinks. See Figure \ref{fig:profit_dis}. This result is quite interesting because impatience affects sale and purchase decision somewhat differently. Increasing impatience will cause an agent who has already owned the asset to sell sooner which is in line with common intuition, but surprisingly a higher $\delta$ will also delay the purchase decision. The economic rationale is that increasing discount rate makes the opportunity to sell the asset less valuable, and hence the agent is more reluctant to enter the trade in the first place.

\begin{remark}
	If we insist $\Psi>0$, then the objective function has an explicit dependence on $\tau$ which will make the entry problem time-inhomogeneous. Such a problem is more difficult to be analyzed analytically or numerically.
\end{remark}

\begin{figure}[!htbp]
	\captionsetup[subfigure]{width=0.45\textwidth}
	\centering
	\subcaptionbox{Optimal purchase (entry) boundaries as a function of $\delta$.\label{fig:profit_dis_entry}}{\includegraphics[scale =0.38] {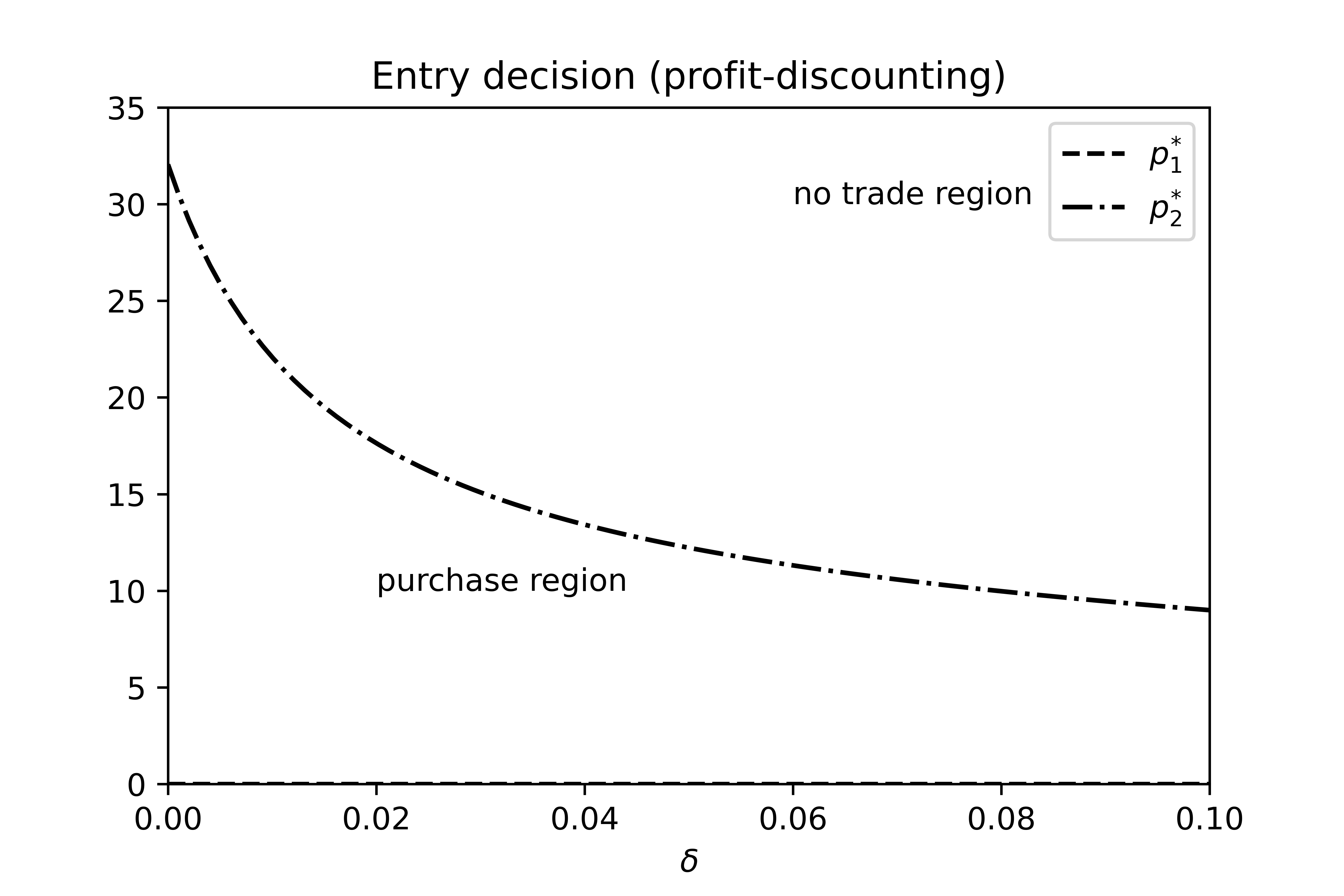}}
	\subcaptionbox{Optimal sale (exit) boundaries as a function of $\delta$.\label{fig:profit_dis_exit}}{\includegraphics[scale =0.38]{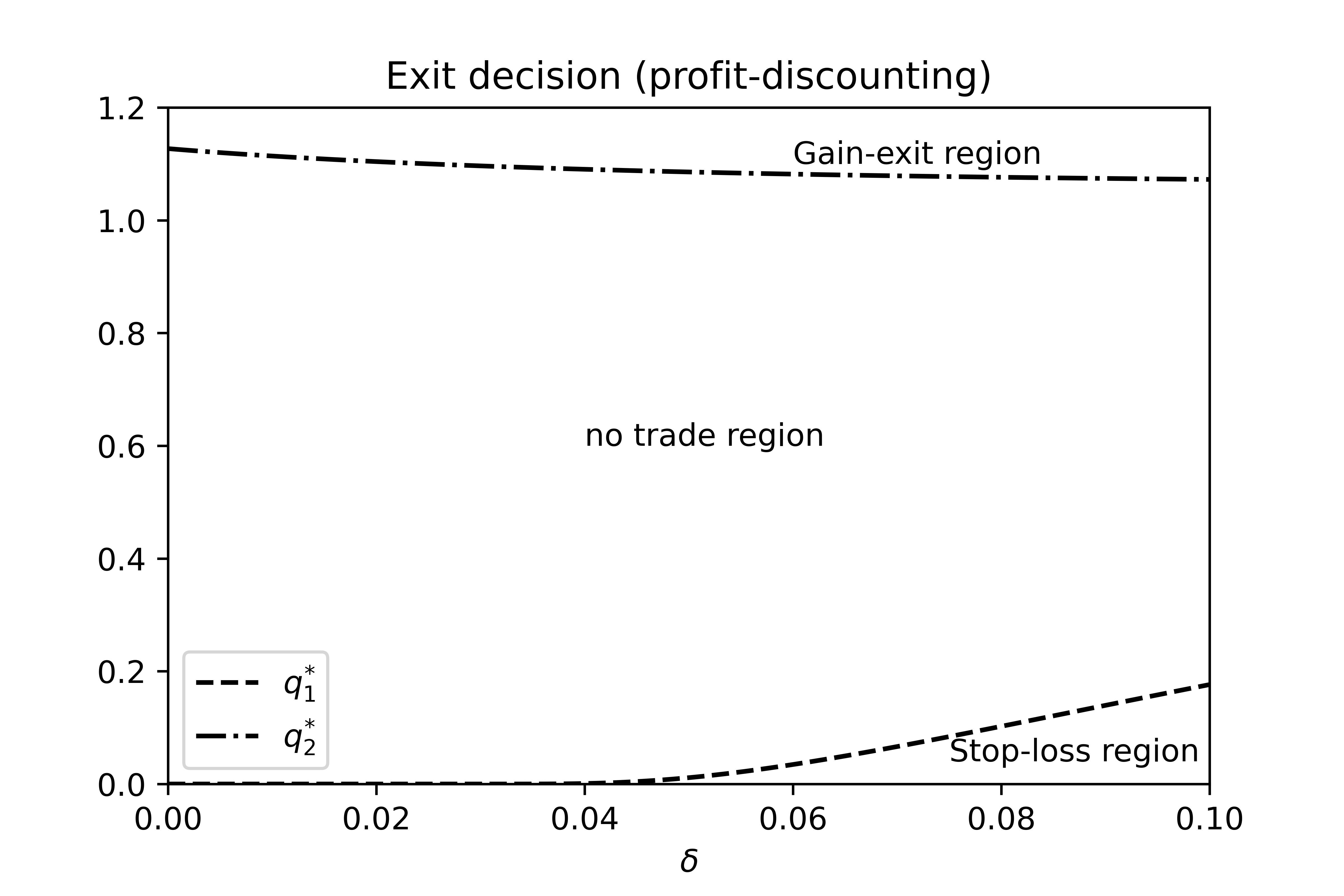}}
	
	\caption{Optimal purchase and sale decision as $\delta$ varies under profit-discounting criterion. Base parameters used are: $\alpha=0.5$, $k=2.25$, $R=1$, $\lambda=1.05$, $\gamma=0.95$, $\Psi=0$, $\mu=0.025$, $\sigma=0.5$.}
	\label{fig:profit_dis}
\end{figure}

\subsubsection{Utility-discounting}

The second possibility to incorporate discounting is to assume the utility of the round-trip proceed is discounted by a single discount factor evaluated at the liquidation date. We call this approach ``{\it utility-discounting}''. The problem is formulated as
\begin{align}
\sup_{\tau,\nu\in\mathcal{T}:\tau\leq \nu}\mathbb{E}\left[e^{-\delta \nu}U\left(P_{\nu} {\mathbbm 1}_{\{\tau<\infty,\nu<\infty\}}-\left(\lambda P_{\tau}+\Psi\right) {\mathbbm 1}_{\{\tau<\infty\}}-R\right)\Bigl | P_0=p\right].
\label{eq:uti_discount}
\end{align}
The downside of this approach is that it does not properly take into the account the timing of the cash outflow (incurred at $\tau$) and inflow (incurred at $\nu$), but an advantage is that this formulation resembles a standard discounted optimal stopping problem.

It turns out that introducing discounting in this fashion will drastically change the agent's entry behavior, as summarized by the proposition below where the proof is given in Appendix \ref{app:discount}.
\begin{proposition}
Suppose $\delta>0$ and let $\omega_1<0<\omega_2$ be the two distinct real  roots to the quadratic equation $\frac{\sigma^2}{2}x^2+(\mu-\frac{\sigma^2}{2})x-\delta=0$. Under the assumption $\alpha< \beta$, the pair of stopping times
\begin{align*}
\tau^*=0\qquad\text{and}\qquad \nu^*=\inf\left\{t\geq \tau^*: P_t\geq \frac{\omega_2(\lambda P_{\tau^*}+\Psi+R) }{\gamma(\omega_2-\alpha)}\right\}
\end{align*}
is optimal to problem \eqref{eq:uti_discount}. In other words, the agent always enters the trade immediately at any price level and subsequently adopts a gain-exit strategy.
\label{prop:uti_discount}
\end{proposition}

Proposition \ref{prop:uti_discount} applies to the case of $\beta>1$ as well. Unlike profit-discounting, the agent will never stop-loss under utility-discounting even when the asset drift is negative. More remarkably, the entry strategy now becomes trivial provided that the standing assumption $\alpha<\beta$ holds. It is also interesting to point out that the optimal strategy is not continuous at $\delta=0$, in the sense that letting $\delta \downarrow 0$ in Proposition \ref{prop:uti_discount} does not recover the no-discounting baseline results in Proposition \ref{prop:entry} or \ref{prop:entryspecial}. See Figure \ref{fig:uti_discount_exit} as an illustration where we show the optimal sale boundary as a function of $\delta$ under utility-discounting. Once discounting is applied to the utility term (no matter how small $\delta$ is), the impact of poor trading performance in form of negative utility can be mitigated by indefinitely deferring the realization of loss. The agent is then effectively protected from negative outcomes and there is no downside to take risk regardless of the asset quality or expensiveness of the transaction costs. The agent is impatient so it is optimal to enter the trade as early as possible, and there is no reason to realize loss thereafter because it can be discounted away.

In view of the above results, profit-discounting seems to yield more reasonable and flexible predictions of the agent's optimal trading behaviors. Nonetheless, a proper understanding of the implications behind the two discounting approaches shall be of interests to the field of behavioral economics.

\begin{figure}[!htbp]
	\centering
	\includegraphics[scale =0.4]{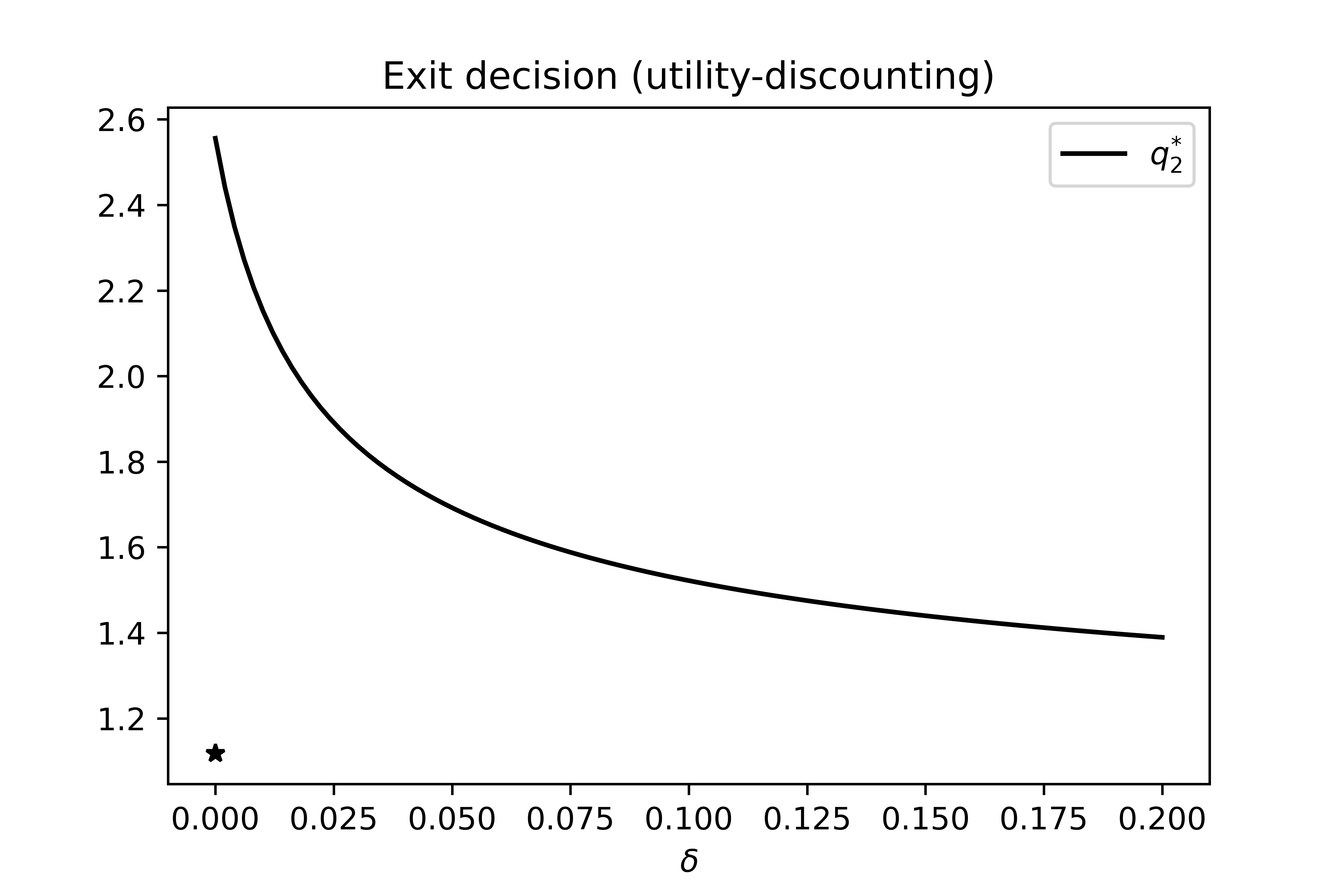}
	\caption{Optimal sale decision as $\delta$ varies under utility-discounting criterion, where the agent will sell the endowed asset whenever the price level is at or above $q^*_2 H$. There is a discontinuity at $\delta=0$ where the star marks the optimal threshold for the undiscounted problem. Base parameters used are: $\alpha=0.5$, $k=2.25$, $\beta=0.85$, $\lambda=1.05$, $\gamma=0.95$, $\Psi=1$.}
	\label{fig:uti_discount_exit}
\end{figure}

\subsection{Multiple round-trip trades}

We have exclusively focused on the case that the agent's utility is derived from a single round-trip trade. But what if the agent can repeatedly purchase and sell the asset? Similar to Section \ref{sect:discount}, we discuss two possible modeling choices regarding how utility over multiple trades can be computed.

\subsubsection{Optimization of utility over total trading profit}
\label{sect:multi_npv}

The first possibility is to consider utility derived from the total net present value of the trading proceed, which is similar to the profit-discounting idea in Section \ref{sect:discount}. Suppose the agent can perform $N$ round-trip trades, then the objective function can be written as
\begin{align}
V(p):=\sup_{\tau_i,\nu_i\in\mathcal{T}:i=1,...,N}J(p;\{\tau_i\}_{i=1,...,N},\{\nu_i\}_{i=1,...,N})
\label{eq:Ntrade}
\end{align}
with
\begin{align*}
&J(p;\{\tau_i\}_{i=1,...,N},\{\nu_i\}_{i=1,...,N})\\
&:=\mathbb{E}\Biggl[ U\biggl(\sum_{i=1}^N \Bigl(\gamma P_{\nu_i} {\mathbbm 1}_{\{\tau_i<\infty,\nu_i<\infty\}}-\left(\lambda P_{\tau_i}+\Psi\right) {\mathbbm 1}_{\{\tau_i<\infty\}}\Bigr)-R\biggr)\Biggl | P_0=p\Biggr].
\end{align*}
Here $\{\tau_i\}_{i=1,...,N}$ (resp. $\{\nu_n\}_{i=1,...,N}$) is an increasing sequence of stopping times representing the entry (resp. exit) time of the $n^{th}$ trade with $$0\leq \tau_1\leq\nu_1\leq \tau_2\leq\nu_2\leq \tau_3\leq \dots\leq\tau_N\leq \nu_N\leq +\infty.$$ In particular, the agent's utility is now derived from the total trading profits from the $N$ available round-trip trading opportunities.

The problem can still be approached by a similar backward induction principle which we outline below based on heuristics. Define
\begin{align*}
&V_b^{(n)}(p;Q):=\\
&\sup_{(\tau_i,\nu_i)_{i=1,...,n}}\mathbb{E}\Biggl[ U\biggl(\sum_{i=1}^n \Bigl(\gamma P_{\nu_i} {\mathbbm 1}_{\{\tau_i<\infty,\nu_i<\infty\}}\\
&\qquad\qquad\qquad\qquad-\left(\lambda P_{\tau_i}+\Psi\right) {\mathbbm 1}_{\{\tau_i<\infty\}}\Bigr)-Q\biggr)\Biggl | P_0=p\Biggr]\\
\end{align*}
and
\begin{align*}
&V_s^{(n)}(p;H):=\\
&\sup_{\nu_1,(\tau_i,\nu_i)_{i=2,...,n}}\mathbb{E}\Biggl[ U\biggl(\sum_{i=2}^n \Bigl(\gamma P_{\nu_i} {\mathbbm 1}_{\{\tau_i<\infty,\nu_i<\infty\}}-\left(\lambda P_{\tau_i}+\Psi\right) {\mathbbm 1}_{\{\tau_i<\infty\}}\Bigr)\\
&\qquad\qquad\qquad\qquad+\gamma P_{\nu_1} {\mathbbm 1}_{\{\nu_1<\infty\}} -H\biggr)\Biggl | P_0=p\Biggr].
\end{align*}
Here $V^{(n)}_b(p;Q)$ represents the value function when there are $n$ purchase and sale opportunities available under some reference point $Q$, and $V^{(n)}_s(p;H)$ represents the value function when the agent has already owned the asset and there are $n-1$ purchase and $n$ sale opportunities available under some reference point $H$. Then heuristically, due to the dynamic programming principle we expect
\begin{align*}
&V_b^{(n)}(p;Q)=\sup_{\tau_1}\mathbb{E}\Biggl[{\mathbbm 1}_{\{\tau_1=\infty\}}U(-Q)+{\mathbbm 1}_{\{\tau_1<\infty\}}\sup_{\nu_1,(\tau_i,\nu_i)_{i=2,...,n}}\Upsilon\Biggl | P_0=p\Biggr]
\end{align*}
where
\begin{align*}
\Upsilon&:=\mathbb{E}\biggl[ U\Bigl(\sum_{i=2}^n \bigl(\gamma P_{\nu_i} {\mathbbm 1}_{\{\tau_i<\infty,\nu_i<\infty\}}-\left(\lambda P_{\tau_i}+\Psi\right) {\mathbbm 1}_{\{\tau_i<\infty\}}\bigr)\\
&\qquad\qquad+\gamma P_{\nu_1} {\mathbbm 1}_{\{\nu_1<\infty\}} -(\lambda P_{\tau_1}+\Psi)-Q\Bigr)\biggl | \mathcal{F}_{\tau_1}\biggr].
\end{align*}
Hence we conclude
\begin{align}
&V_b^{(n)}(p;Q)\nonumber\\
&=\sup_{\tau_1}\mathbb{E}\Bigl[ {\mathbbm 1}_{\{\tau_1=\infty\}}U(-Q)+ {\mathbbm 1}_{\{\tau_1<\infty\}} V^{(n)}_{s}(P_{\tau_1};\lambda P_{\tau_1}+\Psi+Q)\Bigl | P_0=p\Bigr]\label{eq:maxopt}\\
&=\sup_{\tau_1}\mathbb{E}\Bigl[\max\bigl\{V^{(n)}_{s}(P_{\tau_1};\lambda P_{\tau_1}+\Psi+Q),U(-Q)\bigr\}\Bigl | P_0=p\Bigr].\nonumber
\end{align}
The first equality is due to the definition of $V^{(n)}_s$, and the second equality is expected to hold because for the optimal stopping problem in \eqref{eq:maxopt} it is clearly suboptimal to stop the process when $V^{(n)}_{s}(P_{t};\lambda P_{t}+\Psi+Q)<U(-Q)$ and hence the payoff function can be replaced by $\max\{V^{(n)}_{s}(P_{\tau_1};\lambda P_{\tau_1}+\Psi+Q),U(-Q)\}$.

Based on similar reasoning, for $n>1$ we expect
\begin{align*}
V_s^{(n)}(p;H)
&=\sup_{\nu_1}\mathbb{E}\Bigl[ {\mathbbm 1}_{\{\nu_1=\infty\}}U(-H)+ {\mathbbm 1}_{\{\nu_1<\infty\}}\sup_{(\tau_i,\nu_i)_{i=2,...,n}}\Theta\Bigl | P_0=p\Bigr]
\end{align*}
where
\begin{align*}
	\Theta:=\mathbb{E}\Biggl[ U\biggl(\sum_{i=2}^n \Bigl(\gamma P_{\nu_i} {\mathbbm 1}_{\{\tau_i<\infty,\nu_i<\infty\}}-\left(\lambda P_{\tau_i}+\Psi\right) {\mathbbm 1}_{\{\tau_i<\infty\}}\Bigr) +\gamma P_{\nu_1}  -H\biggr)\Biggl | \mathcal{F}_{\nu_1}\Biggr].
\end{align*}
Then
\begin{align*}
	&V_s^{(n)}(p;H)\\
	&=\sup_{\nu_1}\mathbb{E}\Bigl[ {\mathbbm 1}_{\{\nu_1=\infty\}}U(-H)+ {\mathbbm 1}_{\{\nu_1<\infty\}} V^{(n-1)}_{b}(P_{\nu_1};-\gamma P_{\nu_1}+H)\Bigl | P_0=p\Bigr]\\
	&=\sup_{\nu_1}\mathbb{E}\Bigl[ V^{(n-1)}_{b}(P_{\nu_1};-\gamma P_{\nu_1}+H)\Bigl | P_0=p\Bigr].
\end{align*}
The last equality holds because trivially $V^{(n)}_b(p;Q)$ is increasing in $p$ and decreasing in $Q$ such that $V^{(n)}_b(p;-\gamma p+H)\geq V^{(n)}_b(0;H)= U(-H)$ for all $p$ and $n$. Finally, we obviously have
\begin{align*}
 V^{(1)}_s(p;H)&=\sup_{\nu}\mathbb{E}\Bigl[U(\gamma P_{\nu} {\mathbbm 1}_{\{\nu_1<\infty\}}-H)\Bigl|P_0=p\Bigr]\\
 &=\sup_{\nu}\mathbb{E}\Bigl[U(\gamma P_{\nu}-H)\Bigl|P_0=p\Bigr].
\end{align*}

In summary, $V^{(n)}_s$ and $V^{(n)}_b$ should satisfy the recursion
\begin{align*}
V^{(n)}_s(p;H)&=\sup_{\nu}\mathbb{E}\Bigl[U(\gamma P_{\tau}-H)\Bigl | P_0=p\Bigr],& n=1;\\
V^{(n)}_b(p;Q)&=\sup_{\tau}\mathbb{E}\Bigl[\max\bigl\{V_s^{(n)}(P_{\tau};\lambda P_{\tau}+\Psi+Q),U(-Q)\bigr\}\Bigl | P_0=p\Bigr],&n\geq 1;\\
V^{(n)}_s(p;H)&=\sup_{\nu}\mathbb{E}\Bigl[V^{(n-1)}_{b}(P_{\nu};-\gamma P_{\nu}+H)\Bigl | P_0=p\Bigr],& n\geq 2.
\end{align*}
The required value function for problem \eqref{eq:Ntrade} with $N$ round-trip trading opportunities is given by $V^{(N)}_b(p;R)$. A formal verification of the above assertion as well as a thorough analysis of this recursive problem is beyond the scope of the current paper.

{\cb
\subsubsection{Optimization of total utilities from individual trading episodes}
\label{sect:realization_uti}

One may also assume a burst of utility is derived upon completion of each round-trip trade, i.e. the objective function with $N$ round-trip trading opportunities is
\begin{align*}
&J(p;\{\tau_i\}_{i=1,...,N},\{\nu_i\}_{i=1,...,N})\\
&:=\mathbb{E}\Biggl[ \sum_{i=1}^N U\biggl(\gamma P_{\nu_i} {\mathbbm 1}_{\{\tau_i<\infty,\nu_i<\infty\}}-\left(\lambda P_{\tau_i}+\Psi\right) {\mathbbm 1}_{\{\tau_i<\infty\}}-R\biggr)\Biggl | P_0=p\Biggr].
\end{align*}
Unlike the criterion in Section \ref{sect:multi_npv} where a single utility is derived from the total profit of all trades, utility is now realized upon completion of each round-trip trade and the agent's goal is to optimize the sum of those utilities. Define
\begin{align*}
&V_b^{(n)}(p):=\\
&\sup_{(\tau_i,\nu_i)_{i=1,...,n}}\mathbb{E}\Biggl[\sum_{i=1}^n U\biggl( \gamma P_{\nu_i} {\mathbbm 1}_{\{\tau_i<\infty,\nu_i<\infty\}}-\left(\lambda P_{\tau_i}+\Psi\right) {\mathbbm 1}_{\{\tau_i<\infty\}}-R\biggr)\Biggl | P_0=p\Biggr]
\end{align*}
and
\begin{align*}
&V_s^{(n)}(p;H):=\\
&\sup_{\nu_1,(\tau_i,\nu_i)_{i=2,...,n}}\mathbb{E}\Biggl[U\biggl(\gamma P_{\nu_1} {\mathbbm 1}_{\{\nu_1<\infty\}} -H\biggr)\\
&\qquad\qquad+\sum_{i=2}^n  U\biggl(\gamma P_{\nu_i} {\mathbbm 1}_{\{\tau_i<\infty,\nu_i<\infty\}}-\left(\lambda P_{\tau_i}+\Psi\right) {\mathbbm 1}_{\{\tau_i<\infty\}}-R\biggr)\Biggl | P_0=p\Biggr].
\end{align*}
Here $V^{(n)}_b(p)$ is the value function with $n$ purchase and sale opportunities remaining, and $V^{(n)}_s(p;H)$ represents the value function when agent has an endowed asset (with some given reference point $H$ for the first trading episode) and there are $n-1$ purchase and $n$ sale opportunities remaining. 

Based on the dynamic programming principle, one can heuristically write down the recursive system satisfied by the value functions as
\begin{align}
V^{(1)}_s(p;H)&:=\sup_{\nu}\mathbb{E}\Bigl[U(\gamma P_{\nu}-H)\Bigl | P_0=p\Bigr];\label{eq:realized_uti_1}\\
V^{(n)}_b(p)&:=\sup_{\tau}\mathbb{E}\biggl[\max\Bigl(V^{(n)}_s(P_{\tau};\lambda P_{\tau}+\Psi+R),n U(-R)\Bigr)\biggl | P_0=p\biggr];\label{eq:realized_uti_purchase}\\
V^{(n)}_s(p;H)&:=\sup_{\nu}\mathbb{E}\Bigl[U(\gamma P_{\nu}-H)+V^{(n-1)}_b(P_{\nu})\Bigl | P_0=p\Bigr],\label{eq:realized_uti_sale}
\end{align}
where \eqref{eq:realized_uti_purchase} holds for $n\geq 1$ and \eqref{eq:realized_uti_sale} holds for $n\geq 2$. For the ``payoff function'' in \eqref{eq:realized_uti_sale}, the first term $U(\gamma P_{\nu}-H)$ represents the utility burst when the endowed asset is sold and the second term $V_b^{(n-1)}(P_{\nu})$ reflects the maximal sum of expected utilities from the remaining $n-1$ round-trip trading opportunities. Otherwise, if the agent does not own the asset to begin with and decides to purchase at time $\tau$, the maximal expected utility attainable is given by $V^{(n)}_s(P_{\tau};\lambda P_{\tau}+\Psi+R)$ where a new trading episode is initiated with reference point set to $H=\lambda P_{\tau}+\Psi+R$ (the total cost of purchase at time $\tau$ plus the exogenous aspiration level). But the agent can also choose not to purchase at all and forgo all the $n$ remaining trading opportunities. This will result in a payoff of $U(-R)$ for each trading opportunity given up. Hence in \eqref{eq:realized_uti_purchase} the ``payoff function'' for the entry problem is given by $\max\left(V^{(n)}_s(P_{\tau};\lambda P_{\tau}+\Psi+R),n U(-R)\right)$. When $n=1$, \eqref{eq:realized_uti_1} and \eqref{eq:realized_uti_purchase} agree with the sequential optimal stopping problem deduced in Section \ref{sect:solmethod}.

This generalization is conceptually close to the realization utility model in the literature. The canonical formulation (in our notation) of such model is
\begin{align*}
\sup_{(\tau_i,\nu_i)_{i=1,2,3,..}}\mathbb{E}\left[ \sum_{i=1}^\infty e^{-\delta \nu_i}U\left(G_{\nu_i-},Q_{\nu_i-}\right) \right].
\end{align*}
$Q=(Q_t)_{t\geq 0}$ is the reference point process such that $Q_{t}$ is the benchmark to be used for performance evaluation at time $t$. $G=(G_t)_{t\geq 0}$ is the gain-and-loss process with $G_t:= \gamma P_t - Q_{t}$ representing the size of realized gain-and-loss if the agent liquidates an owned asset at time $t$. Typically, the function $U(G,Q)$ is assumed to be homogeneous in $Q$ such that $U(G,Q)=Q^{\eta}u(G/Q)$ for some $\eta\in(0,1]$ and $u(\cdot)$ is S-shaped.

There are many choices with the reference point process $Q$. Ingersoll and Jin~\cite{ingersoll-jin13} consider $Q_{t}= P_{\tau_i}$ for $t\in[\tau_{i},\tau_{i+1})$ (up to a constant multiplier) which is simply the most recent purchase price of the asset. Barberis and Xiong~\cite{barberis-xiong12} and Dai et al.~\cite{dai-qin-wang22} take $Q_{t}=  P_{\tau_i} e^{r(t-\tau_i)}$ for $t\in[\tau_{i},\tau_{i+1})$ which is the most recent purchase price growing at the risk-free rate. He and Yang~\cite{he-yang19} incorporate an additional term which asymmetrically adapts to the paper gain-and-loss. Kong et al.~\cite{kong-qin-yue22} study a path-dependent reference point which is a weighted-average of the asset prices throughout the current trading episode. A common feature among the cited papers above is that $Q_t$ is proportional to $P_{\tau_i}$ over $t\in[\tau_{i},\tau_{i+1})$ and in turn dimension reduction is possible via introducing a new state variable $X_t:=P_t/Q_t$. This greatly simplifies the entry problem but it also trivializes the optimal strategy where one either immediately enters the trade again after a sale or never enters the trade in the first place. This observation remains the same even if one introduces additional modeling elements such as a Poisson random termination time and an extra utility term over final wealth. See Proposition 3.4 of He and Yang~\cite{he-yang19}.

Our formulation can be seen as a version of the realization utility model with finite number of trading opportunities where the reference level process is (see Remark \ref{remark:refpt} at the end of this subsection as well)
\begin{align}
Q_t:=
\begin{cases}
\lambda P_{\tau_i} + \Psi +R,& t\in[\tau_i,\nu_i);\\
R,& t\in[\nu_i,\tau_{i+1}),
\end{cases}
\label{eq:refpt_process}
\end{align}
the gain-and-loss process is $G_t:=\gamma P_t  {\mathbbm 1}_{\{t<\infty\}}-Q_t$, the utility function is $U(G,Q)=Q^{\alpha}u(G/Q)$ with $u(x):=x^{\alpha} {\mathbbm 1}_{\{x\geq 0\}}-k|x|^{\alpha}{\mathbbm 1}_{\{x< 0\}} $ and the discount rate $\delta$ is set to zero. The most important distinction of our framework from the existing realization utility models is that our reference point consists of a constant component $\Psi+R$ reflecting fixed transaction cost and some baseline aspiration level. Without this component, the reference level over a particular trading episode is always proportional to the asset value at the beginning of the episode. Specifically, if $\Psi=R=0$, then using \eqref{eq:realized_uti_1} to \eqref{eq:realized_uti_sale}, the scaling property of $U$ and the geometric Brownian motion assumption of $P$, we can inductively deduce for all $n$ that $V^{(n)}_s(p;H)=H^{\alpha}V(\frac{p}{H};1)$ and
\begin{align*}
V^{(n)}_b(p)&=\sup_{\tau}\mathbb{E}\Bigl[\max\bigl(V^{(n)}_s(P_{\tau};\lambda P_{\tau}),0\bigr)\Bigl|P_0=p\Bigr]\\
&=\sup_{\tau}\mathbb{E}\Bigl[\max\bigl(\lambda^{\alpha} P_{\tau}^{\alpha}V^{(n)}_s(1/\lambda;1),0\bigr)\Bigl|P_0=p\Bigr].
\end{align*}
If $V^{(n)}_s(1/\lambda;1)<0$, then $\tau=+\infty$ is optimal and the associated value function for the entry problem is $V^{(n)}_b(p)=0$. If instead $V^{(n)}_s(1/\lambda;1)>0$, then we have $V^{(n)}_b(p)=\lambda^{\alpha}V^{(n)}_s(1/\lambda;1)\sup_{\tau}\mathbb{E}[P_{\tau}^{\alpha}|P_0=p]=\lambda^{\alpha}V^{(n)}_s(1/\lambda;1)p^{\alpha}$ under the standing assumption $\alpha\leq \beta$. The corresponding optimal entry strategy is $\tau=0$. 

The above observations in conjunction with our main theoretical results under $n=1$ suggest that incorporation of a constant component within the dynamic reference point (e.g. in form of fixed transaction cost or a default aspiration level) might enable a realization utility model to generate more realistic and non-trivial purchase behaviors. For the sale decision, recall that in the baseline model with $n=1$, the optimal sale strategy is a simple gain-exit rule (Lemma \ref{lem:henderson}). We expect this to change when $n>1$ since the effective payoff function of the exit problem \eqref{eq:realized_uti_sale} now contains an additional term $V^{(n-1)}_b(p)$ which will drastically change the convexity/concavity of the scaled exit payoff function. We leave the full analysis of the problem for future work.

\begin{remark}
	At the first sight, \eqref{eq:refpt_process} looks more complicated than the ones proposed in the existing literature where our value of $Q_t$ depends on whether the agent is inside a trading episode holding the asset ($t\in[\tau_i,\nu_i)$ for some $i$) or outside a trading episode without any asset ($t\in[\nu_i,\tau_{i+1})$ for some $i$). From an economic point of view, the reference point should not depend on $P_{\tau_i}$ anymore once the $i^{th}$ trading episode is complete and hence should be reset to the baseline aspiration level until the start of the $(i+1)^{th}$ trading episode. In the special case of $\tau_{i+1}=\nu_i$ for all $i$ such that there is no time delay between the exit of an existing trade and the entry of a new trade, \eqref{eq:refpt_process} simplifies to $Q_t=\lambda P_{\tau_i}+\Psi+R$ for $t\in[\tau_i,\tau_{i+1})$ which resembles the usual definition in the literature. In absence of a constant component within the reference point process, the entry decision is trivial where $\tau_{i+1}=\nu_i$ for all $i$ is indeed optimal because of the scaling property discussed previously so there is no need to ``correctly'' specify $Q_t$ over $t\in[\nu_i,\tau_{i+1})$. But this result is no longer true when $R$ or $\Psi$ is non-zero and hence it is necessary to model the reference point process more carefully.
\label{remark:refpt}
\end{remark}
}

\subsection{Endogenous aspiration level}
\label{sect:endoR}

Among all the model parameters, $R$ the ``aspiration level'' of the agent is the hardest one to be interpreted and estimated. It can be a pure psychological parameter representing the agent's subjective threshold which distinguishes gains and losses. In some applications such as delegated portfolio management, $R$ can also be the performance target imposed on the agent by a manager. Instead of calibrating $R$ where the exercise can be very context-specific, one may also seek to endogenize this parameter by introducing a further optimality criterion. 
	
Consider a principal-agent setup as an example. The principal imposes an aspiration level on the agent in form of a performance target. Under a given $R$, the agent's optimal trading rule $(\tau^*(R),\nu^*(R))$ can be obtained by solving \eqref{eq:valfun}. If the principal has a utility function $\tilde{U}(\cdot)$ over the trading profit, then a particular choice of $R$ will bring the principal an expected utility level of
\begin{align}
\tilde{V}(R):=\mathbb{E}\left[\tilde{U}(\gamma P_{\nu^*(R)} {\mathbbm 1}_{\{\tau^*(R)<\infty,\nu^*(R)<\infty\}}-(\lambda P_{\tau^*(R)}+\Psi) {\mathbbm 1}_{\{\tau^*(R)<\infty\}})\right].
\label{eq:principal}
\end{align}
The above can be maximized (numerically) with respect to $R$ where the solution in general depends on the initial asset price $P_0$. In Figure \ref{fig:ce}, we consider a risk-averse principal with utility function $\tilde{U}(x)=-e^{-\eta x}$ where $\eta>0$ is the constant absolute risk-aversion level and $\tilde{V}(R)$ can be maximized at some interior $R$ in this particular example. The higher the Sharpe ratio of the asset (equivalent to a lower level of $\beta$), the higher the level of endogenous aspiration. In other words, a more aggressive performance goal is set in a bullish market. However, there are also examples that the principal's maximization problem is degenerate (e.g. $\tilde{V}(R)$ being monotonically increasing or decreasing in $R$). We leave the complete analysis of such principal-agent problem for future research. 

An alternative consideration to endogenize $R$ is to modify the agent's utility function such that a round-trip profit of $x$ leads to a utility value of $U(x-R,R)$, where $U(\cdot,\cdot)$ is increasing in the first argument and decreasing in the second argument. The second argument of $U$ can reflect the agent's desire for ``self-improvement'' and ``self-enhancement'' which is achieved by choosing a high aspiration level $R$. See for example Falk and Knell~\cite{falk-knell04}. The optimal $R$ can then be determined alongside with the agent's optimal trading rule.

\begin{figure}[!htbp]
	\centering
	\includegraphics[scale =0.4]{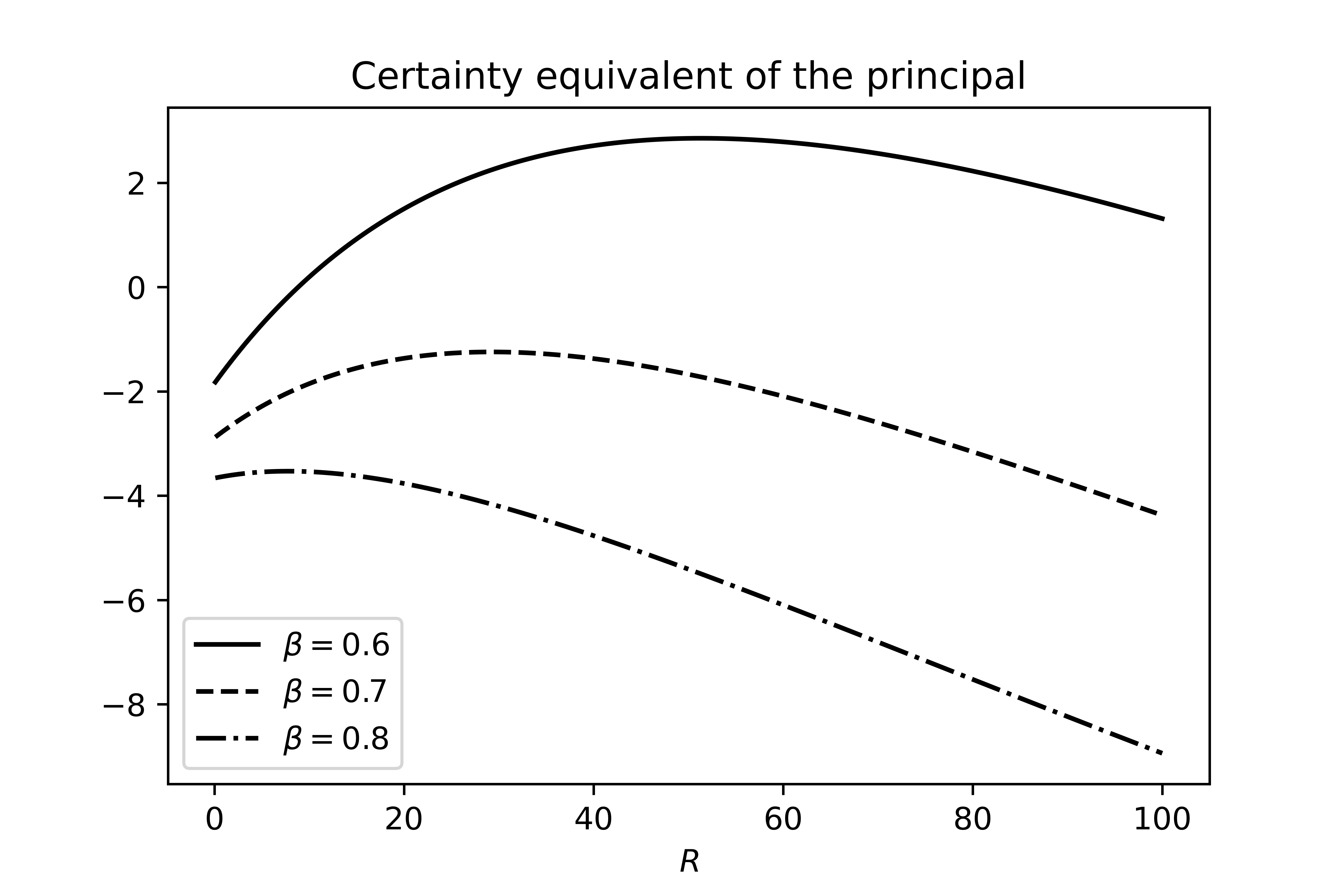}
	\caption{The principal's certainty equivalent $-\ln (- \tilde{V}(R))/\eta$ as a function of $R$ when $\tilde{U}(x)=-e^{-\eta x}$. Base parameters used are: $\alpha=0.5$, $k=2.25$, $\lambda=1.05$, $\gamma=0.95$, $\Psi=1$, $\eta=0.01$, $P_0=30$.}
	\label{fig:ce}
\end{figure}
}

\section{Concluding remarks}
\label{sect:conc}

This paper considers a dynamic trading model under Prospect Theory preference with transaction costs. By solving a sequential optimal stopping problem, we find that the optimal trading strategy can have various forms depending on the model parameters and the price level of the asset. The impact of transaction costs is subtle. In contrast to conventional wisdom, increasing transaction costs does not necessarily deter economic agents from trading participation because the agents may face a higher reference point and in turn be more risk-aggressive in an expensive trading environment. These results could potentially be useful to policy makers to better understand how undesirable speculative trading behaviors in certain markets can be effectively curbed. 

Our key mathematical results are derived under a somewhat stylized modeling specification. In particular, asymmetry of degree of risk-aversion/seeking over gains/losses, fixed transaction cost on sale and negative aspiration level are currently omitted from the baseline analysis. While these omissions allow us to derive sharp characterization and comparative statics of the optimal trading rules, it will nonetheless be constructive to extend the model to further examine the impact of other economic factors. Section \ref{sect:extend} also highlights a number of possible extensions for further rigorous mathematical analysis in presence of negative drift, subjective discounting, repeated trading opportunities and endogenous aspiration level. {\cb In particular, a natural extension is to further explore the implications of our results to the literature of realization utility where we believe a different specification of the reference point process (e.g. incorporation of a constant component) can lead to more realistic prediction of purchase behaviors.}

A more ambitious goal is to further incorporate probability weighting within our continuous-time optimal stopping model (as per Xu and Zhou~\cite{xu-zhou13} and Henderson et al.~\cite{henderson-hobson-tse18}) to fully reflect the features of the \textit{Cumulative Prospect Theory} framework of Tversky and Kahneman~\cite{tversky-kahneman92}. However, technical difficulties are likely to arise due to the time-inconsistency brought by probability weighting. Precise formulation of the problem as well as development of the appropriate mathematical techniques should prove to be an another interesting proposal for future research.

Finally, the surprising comparative statics documented in this paper also reveal that the economic interaction between market frictions and behavioral preferences can be subtle or even counter-intuitive. The model considered in this paper is just one of the many possible motivations that these two topics should not be studied in isolation, and this could potentially open up a new strand of literature to advance our understanding towards trading behaviors in a more realistic setup.

\section*{Acknowledgment}

The authors are grateful to the two anonymous reviewers whose comments and suggestions have helped improve the previous version of the paper. This research was supported in part by the EPSRC (UK) grant (EP/V008331/1).

\bibliographystyle{abbrv}

\appendix\normalsize

\section{Appendix}

\subsection{Proofs of Proposition \ref{prop:entry}, \ref{prop:entryspecial} and \ref{prop:compstat}.}

We start with two useful lemmas.

\begin{lemma}
	Write $\xi:=\frac{\lambda}{\gamma}$. For the function $f$ defined in \eqref{eq:f} we have
	\begin{align}
	\lim_{x\to +\infty}f(x)=
	\begin{cases}
	+\infty,& \xi<\left[\frac{\alpha}{\beta k}c^{1-\beta}(c-1)^{\alpha-1}\right]^{\frac{1}{\beta}};\\
	0,& \xi=\left[\frac{\alpha}{\beta k}c^{1-\beta}(c-1)^{\alpha-1}\right]^{\frac{1}{\beta}};\\
	-\infty,& \xi>\left[\frac{\alpha}{\beta k}c^{1-\beta}(c-1)^{\alpha-1}\right]^{\frac{1}{\beta}}.
	\end{cases}
	\end{align}
	Moreover:
	\begin{enumerate}[wide, labelwidth=!, labelindent=0pt]
		\item Suppose $\alpha<\beta<1$:
		
		\begin{enumerate}[wide, labelwidth=!, labelindent=0pt, label={(\arabic*)}]
			\item If $\xi\leq[\frac{\alpha}{\beta k}c^{1-\beta}(c-1)^{\alpha-1}]^{\frac{1}{\beta}}$, then $f$ is an increasing concave function.
			\item If $\xi>[\frac{\alpha}{\beta k}c^{1-\beta}(c-1)^{\alpha-1}]^{\frac{1}{\beta}}$, then $f$ is concave increasing on $[0,x_2^*]$, concave decreasing on $[x_2^*,\tilde{x}]$ and convex decreasing on $[\tilde{x},\infty)$. Here $x_2^*$ and $\tilde{x}$ are respectively the solutions to the equation
			\begin{align}
			c^{1-\beta}(c-1)^{\alpha-1}\left(x^{-\frac{1}{\beta}}+\frac{\xi \alpha}{\beta}\right)-k\xi\left(x^{-\frac{1}{\beta}}+\xi\right)^\beta=0
			\label{eq:turnpt}
			\end{align}
			and
			\begin{align}
			0&=c^{1-\beta}(c-1)^{\alpha-1}\left[-\frac{\xi \alpha}{ \beta^2}(\beta-\alpha)+\frac{1}{\beta}\left(\frac{\alpha}{\beta}-\beta+\alpha-1\right)x^{-\frac{1}{\beta}}\right]\nonumber\\
			&\qquad-k\left(\xi+x^{-\frac{1}{\beta}}\right)^\beta\left[-\xi\left(1-\frac{\alpha}{\beta}\right)+\left(\frac{1}{\beta}-1\right)x^{-\frac{1}{\beta}}\right].
			\label{eq:inflexpt}
			\end{align}
		\end{enumerate}
		
		\item Suppose $\alpha<\beta=1$:
		
		\begin{enumerate}[wide, labelwidth=!, labelindent=0pt]
			\item If $\xi\leq\frac{\alpha}{k}(c-1)^{\alpha-1}$, then $f$ is an increasing concave function.
			\item If $\frac{\alpha}{ k}(c-1)^{\alpha-1}<\xi\leq\frac{1}{ k}(c-1)^{\alpha-1}$, then $f$ is concave increasing on $[0,x_2^*]$, concave decreasing on $[x_2^*,\tilde{x}]$ and convex decreasing on $[\tilde{x},\infty)$ with
			\begin{align*}
			x^*_2:=\frac{(c-1)^{\alpha-1}-k\xi}{\xi\bigl[k\xi-\alpha(c-1)^{\alpha-1}\bigr]},\qquad 	\tilde{x}:=\frac{2(c-1)^{\alpha-1}-k\xi}{\xi\bigl[k\xi-\alpha(c-1)^{\alpha-1}\bigr]}.
			\end{align*}
			\item If $\xi>\frac{1}{ k}(c-1)^{\alpha-1}$, then $f$ is a decreasing function.
		\end{enumerate}
		
	\end{enumerate}

	\label{lem:shapef}
\end{lemma}

\begin{proof}
	
	We can rewrite $f$ as
	\begin{align*}
	f(x)=\frac{\frac{\alpha}{\beta}c^{1-\beta}(c-1)^{\alpha-1}-k(\xi +x^{-\frac{1}{\beta}})^{\beta}}{(\xi +x^{-\frac{1}{\beta}})^{\beta-\alpha}}x^{\frac{\alpha}{\beta}}
	\end{align*}
	such that $\displaystyle \lim_{x\to \infty}f(x)=\pm \infty$ when $\xi \gtrless [\frac{\alpha}{\beta k}c^{1-\beta}(c-1)^{\alpha-1}]^{\frac{1}{\beta}}$. The corner case of $\xi = [\frac{\alpha}{\beta k}c^{1-\beta}(c-1)^{\alpha-1}]^{\frac{1}{\beta}}$ can be analyzed by a simple application of L'Hospital's rule.
	
	We now derive the shapes of $f$ by first focusing on the case of $\beta\neq 1$. Direct differentiation gives
	\begin{align}
	f'(x)&=\frac{\alpha}{\beta}\frac{c^{1-\beta}(c-1)^{\alpha-1}\left(\frac{\xi\alpha}{\beta}x^{\frac{1}{\beta}}+1\right)-k\xi x^{\frac{1}{\beta}-1}\left(\xi x^{\frac{1}{\beta}}+1\right)^\beta}{\left(\xi x^{\frac{1}{\beta}}+1\right)^{\beta-\alpha+1}}\nonumber\\
	&=\frac{\alpha x^{\frac{1}{\beta}}h_1(x^{-\frac{1}{\beta}})}{\beta\left(\xi x^{\frac{1}{\beta}}+1\right)^{\beta-\alpha+1}}
	\label{eq:f_firstder}
	\end{align}
	with
	\begin{align*}
	h_1(z):=c^{1-\beta}(c-1)^{\alpha-1}\left(z+\frac{\xi \alpha}{\beta}\right)-k\xi\left(z+\xi\right)^\beta,
	\end{align*}
	and
	\begin{align*}
	f''(x)&=\frac{\xi\alpha x^{\frac{1}{\beta}-2}}{\beta\left(\xi x^{\frac{1}{\beta}}+1\right)^{\beta-\alpha+2}}\\
	&\qquad\times\Biggl\{c^{1-\beta}(c-1)^{\alpha-1}x\biggl[-\frac{\xi\alpha}{\beta^2}(\beta-\alpha)x^{\frac{1}{\beta}}+\frac{1}{\beta}\Bigl(\frac{\alpha}{\beta}-\beta+\alpha-1\Bigr)\biggr] \\
	&\qquad\qquad -k\biggl(\xi x^{\frac{1}{\beta}}+1\biggr)^\beta\biggl[-\xi\Bigl(1-\frac{\alpha}{\beta}\Bigr)x^{\frac{1}{\beta}}+\frac{1}{\beta}-1\biggr]\Biggl\} \\
	&=\frac{\xi\alpha x^{\frac{2}{\beta}-1}h_2(x^{-\frac{1}{\beta}})}{\beta\left(\xi x^{\frac{1}{\beta}}+1\right)^{\beta-\alpha+2}}
	\end{align*}
	where
	\begin{align*}
	h_2(z)&:=c^{1-\beta}(c-1)^{\alpha-1}\biggl[-\frac{\xi \alpha}{ \beta^2}(\beta-\alpha)+\frac{1}{\beta}\Bigl(\frac{\alpha}{\beta}-\beta+\alpha-1\Bigr)z\biggr]\\
	&\qquad-k(\xi+z)^\beta\biggl[-\xi\Bigl(1-\frac{\alpha}{\beta}\Bigr)+\Bigl(\frac{1}{\beta}-1\Bigr)z\biggr].
	\end{align*}
	
	We first investigate the convexity/concavity of $f$ by studying the sign of $f''(x)$, which is determined by that of $h_2(x^{-\frac{1}{\beta}})$. Check that
	\begin{align*}
	h_2(0)&=\xi\left(1-\frac{\alpha}{\beta}\right)\left[-\frac{\alpha}{\beta}c^{1-\beta}(c-1)^{\alpha-1}+k\xi^\beta\right], \\
	h_2'(0)&=c^{1-\beta}(c-1)^{\alpha-1}\frac{1}{\beta}\left(\frac{\alpha}{\beta}-\beta+\alpha-1\right)-k\xi^{\beta}\left(\frac{1}{\beta}-\beta+\alpha-1\right)
	\end{align*}
	and
	\begin{align*}
	h_2''(z)=-k\left(\xi+z\right)^{\beta-2}\bigl[(1-\beta)(1+\beta)z+\xi(1-\beta)(2+\beta-\alpha)\bigr]<0
	\end{align*}
	for all $z>0$ since $\alpha\leq \beta\leq 1$ and thus $h_2$ is concave. Then there are two possibilities.
	
	Suppose $\xi\leq[\frac{\alpha}{\beta k}c^{1-\beta}(c-1)^{\alpha-1}]^{\frac{1}{\beta}}$, then $h_2(0)\leq 0$ and
	\begin{align*}
	h_2'(0)&=c^{1-\beta}(c-1)^{\alpha-1}\frac{1}{\beta}\left(\frac{\alpha}{\beta}-\beta+\alpha-1\right)-k\xi^{\beta}\left(\frac{1}{\beta}-\beta+\alpha-1\right) \\
	&<c^{1-\beta}(c-1)^{\alpha-1}\frac{1}{\beta}\left(\frac{\alpha}{\beta}-\beta+\alpha-1\right)-k\xi^{\beta}\left(\frac{\alpha}{\beta}-\beta+\alpha-1\right) \\
	&=k\left(\frac{\alpha}{\beta}-\beta+\alpha-1\right)\left[\frac{c^{1-\beta}(c-1)^{\alpha-1}}{\beta k}-\xi^\beta\right] \\
	&\leq k\left(\frac{\alpha}{\beta}-\beta+\alpha-1\right)\left[\frac{\alpha c^{1-\beta}(c-1)^{\alpha-1}}{\beta k}-\xi^\beta\right] \\
	&\leq 0
	\end{align*}
	where we have used the facts that $\alpha<1$ and $\frac{\alpha}{\beta}-\beta+\alpha-1<\alpha-\beta\leq 0$. Since $h_2$ is concave, we must have $h_2(z)\leq 0$ for all $z>0$. Hence $f_2''(x)\leq 0$ for all $x\geq 0$, i.e. $f$ is a concave function. 
	
	Suppose instead $\xi>[\frac{\alpha}{\beta k}c^{1-\beta}(c-1)^{\alpha-1}]^{\frac{1}{\beta}}$, then $h_2(0)> 0$ and
	\begin{align*}
	\lim_{z\to\infty}\frac{h_2(z)}{z^{\beta+1}}=-k\left(\frac{1}{\beta}-1\right)<0
	\end{align*}
	such that $h_2(z)\to -\infty$ as $z\to\infty$. As $h_2$ is concave, we must have $h_2(z)$ down-crossing zero exactly once on $(0,\infty)$. Hence $f_2''(x)\propto h_2(x^{-\frac{1}{\beta}})$ has exactly one sign change from negative to positive, i.e. $f$ is concave for small $x$ and convex for large $x$ with a unique inflexion point $\tilde{x}$ which is given by the solution to $h_2(x^{-\frac{1}{\beta}})=0$. This corresponds to equation \eqref{eq:inflexpt}.
	
	Now we look at the monotonicity of $f$ via the sign of $f'(x)$ which in turn is determined by that of $h_1(x^{-\frac{1}{\beta}})$. Check that
	\begin{align*}
	h_1(0)&=\xi\left[\frac{\alpha}{\beta}c^{1-\beta}(c-1)^{\alpha-1}-k\xi^\beta\right],\\
	h_1'(z)&=c^{1-\beta}(c-1)^{\alpha-1}-\frac{k\xi\beta}{\left(z+\xi\right)^{1-\beta}}.
	\end{align*}
	Observe that $h_1'$ is increasing and thus $h_1$ is convex. There are two cases.
	
	Suppose $\xi\leq[\frac{\alpha}{\beta k}c^{1-\beta}(c-1)^{\alpha-1}]^{\frac{1}{\beta}}$, then $h_1(0)\geq 0$ and 
	\begin{align*}
	h_1'(0)&=c^{1-\beta}(c-1)^{\alpha-1}-k\beta\xi^\beta \geq \alpha c^{1-\beta}(c-1)^{\alpha-1}-k\beta\xi^\beta
	\geq 0.
	\end{align*}
	As $h_1$ is convex, we must have $h_1(z)\geq 0$ for all $z>0$. Hence $f'(x)\geq 0$ for all $x\geq 0$, i.e. $f$ is an increasing function. Together with the consideration of $f''$ in this parameter regime, $f$ is an increasing concave function.
	
	Suppose $\xi>[\frac{\alpha}{\beta k}c^{1-\beta}(c-1)^{\alpha-1}]^{\frac{1}{\beta}}$, then we have $h_1(0)<0$ instead. We also have
	\begin{align}
	\lim_{z\to\infty}\frac{h_1(z)}{z}=c^{1-\beta}(c-1)^{\alpha-1}>0
	\label{eq:h1infty}
	\end{align}
	and hence $h_1(z)\to \infty$ as $z\to\infty$. Since $h_1$ is convex, $h_1$ must up-cross zero exactly once on $(0,\infty)$. Therefore $f'(x)\propto h_1(x^{-\frac{1}{\beta}})$ changes sign exactly once, from which we conclude $f$ is first increasing and then decreasing with a unique turning point $x_2^*$. Moreover, $x_2^*$ is the solution to $h_1(x^{-\frac{1}{\beta}})=0$ which is equivalent to \eqref{eq:turnpt}. Taking the behavior of $f''$ into consideration, we conclude that $f$ is increasing concave on $[0,x_2^*]$, decreasing concave on $[x_2^*,\tilde{x}]$ and decreasing convex on $[\tilde{x},\infty)$. 
	
	The case of $\beta=1$ can be handled similarly. The key difference is that \eqref{eq:h1infty} no longer holds when $\beta=1$ but rather we will have
	\begin{align*}
	\lim_{z\to\infty}\frac{h_1(z)}{z}=(c-1)^{\alpha-1}-k\xi
	\end{align*}
	instead which can be either positive or negative. In the case of $\xi>\frac{(c-1)^{\alpha-1}}{k}$, we have $h_1(z)\to -\infty$ as $z\to \infty$. We can then deduce $f'(x)$ is negative for all $x$ and thus $f$ is decreasing. \qed
\end{proof}

\begin{lemma}
	If $\xi:=\frac{\lambda}{\gamma}>[\frac{\alpha}{\beta k}c^{1-\beta}(c-1)^{\alpha-1}]^{\frac{1}{\beta}}$, then $f(x)<0$ for all $x$ where $f$ is defined in \eqref{eq:f}.
	\label{lem:zerobound}
\end{lemma}

\begin{proof}
	The result follows directly from the definition of $f$ that
	\begin{align*}
	f(x)&=\left[\frac{\alpha}{\beta}c^{1-\beta}(c-1)^{\alpha-1}x\left(\frac{\lambda}{\gamma} x^{1/\beta}+1\right)^{-\beta}-k\right]\left(\frac{\lambda}{\gamma} x^{1/\beta}+1\right)^{\alpha} \\
	&<\left[k\left(\frac{\lambda}{\gamma}\right)^\beta x\left(\frac{\lambda}{\gamma} x^{1/\beta}+1\right)^{-\beta}-k\right]\left(\frac{\lambda}{\gamma} x^{1/\beta}+1\right)^{\alpha}\\
	&<\left[k\left(\frac{\lambda}{\gamma}\right)^\beta x\left(\frac{\lambda}{\gamma} x^{1/\beta}+0\right)^{-\beta}-k\right]\left(\frac{\lambda}{\gamma} x^{1/\beta}+1\right)^{\alpha}=0.
	\end{align*} \qed
\end{proof}

With the help of Lemma \ref{lem:shapef} and \ref{lem:zerobound}, we now prove Proposition \ref{prop:entry}, \ref{prop:entryspecial} and \ref{prop:compstat}.

\begin{proof}[Proof of Proposition \ref{prop:entry}]
	From the discussion in Section \ref{sect:solmethod}, we have to identify $\bar{g}_2(\cdot)$ the smallest concave majorant of the function 
	\begin{align*}
	g_2(\theta)&:=\max\Bigl\{V_1\bigl(s^{-1}(\theta); \lambda s^{-1}(\theta)+\Psi+R\bigr),U(-R)\Bigr\}\\
	&=\max\Bigl\{V_1(\theta^{\frac{1}{\beta}}; \lambda \theta^{\frac{1}{\beta}}+\Psi+R),U(-R)\Bigr\}
	\end{align*}
	where $V_1$ is the value function of the exit problem given in Lemma \ref{lem:henderson}. Since $c>1$ and we assume that $R> 0$, $\Psi\geq 0$ and $\gamma\leq 1\leq \lambda$, we have $c(\frac{\lambda \theta^{\frac{1}{\beta}}+\Psi+R}{\gamma})\geq \theta^{\frac{1}{\beta}}$ and hence the first regime in \eqref{eq:exitvalufun} will always apply when evaluating $V_1(\theta^{\frac{1}{\beta}}; \lambda \theta^{\frac{1}{\beta}}+\Psi+R)$, i.e.
	\begin{align}
	v_1(\theta)&:=V_1(\theta^{\frac{1}{\beta}}; \lambda \theta^{\frac{1}{\beta}}+\Psi+R) \nonumber\\
	&=-k(\lambda \theta^{\frac{1}{\beta}}+\Psi+R)^{\alpha}+\frac{\alpha}{\beta}(\lambda \theta^{\frac{1}{\beta}}+\Psi+R)^{\alpha-\beta}c^{1-\beta}(c-1)^{\alpha-1}\gamma^\beta\theta \nonumber\\
	&=R^\alpha\left(1+\frac{\Psi}{R}\right)^\alpha\Biggl[\frac{\alpha}{\beta}c^{1-\beta}(c-1)^{\alpha-1}\biggl(\frac{\lambda}{\gamma}\frac{\gamma \theta^{\frac{1}{\beta}}}{R+\Psi}+1\biggr)^{\alpha-\beta}\biggl(\frac{\gamma}{R+\Psi}\biggr)^\beta\theta\nonumber\\
	&\qquad\qquad\qquad\qquad -k\biggl(\frac{\lambda}{\gamma}\frac{\gamma \theta^{\frac{1}{\beta}}}{R+\Psi}+1\biggr)^\alpha\Biggr] \nonumber \\
	&=R^\alpha\left(1+\frac{\Psi}{R}\right)^{\alpha}f\left(\left(\frac{\gamma}{R+\Psi}\right)^{\beta}\theta\right)
	\label{eq:v1rep}
	\end{align}
	where $f$ is defined in \eqref{eq:f}. The shape of $f$ under different parameters combination is given by Lemma \ref{lem:shapef} and thus we have the following cases.
	
	When $\xi:=\frac{\lambda}{\gamma}\leq[\frac{\alpha}{\beta k}c^{1-\beta}(c-1)^{\alpha-1}]^{\frac{1}{\beta}}$, $f$ is increasing concave and we have $\displaystyle \lim_{x\to \infty}f(x)= +\infty$. These properties are inherited by $v_1$. Furthermore, $$v_1(0)=R^{\alpha}(1+\Psi/R)^{\alpha}f(0)=-kR^{\alpha}(1+\Psi/R)^{\alpha}\leq -kR^{\alpha}$$ and $$\lim_{\theta\to \infty}v_1(\theta)> 0>-kR^{\alpha}.$$ Thus $g_2$ is constructed by truncating an increasing concave function from below at $-kR^{\alpha}$. The smallest concave majorant of $g_2$ is formed by drawing a tangent line passing through $(0,-kR^{\alpha})$ which touches $v_1$ at some $\theta^*_1$. See Figure \ref{fig:case1}. The exact form of the smallest concave majorant is
	\begin{align*}
	\bar{g}_2(\theta)=
	\begin{cases}
	\frac{v_1(\theta_1^*)+kR^{\alpha}}{\theta_1^*}\theta-kR^{\alpha},& \theta<\theta_1^*;\\
	v_1(\theta),& \theta\geq \theta_1^*,
	\end{cases}
	\end{align*}
	which is equivalent to \eqref{eq:exitvalfun1} upon observing that $V_2(p)=\bar{g}_2(p^\beta)$.
	
	The optimal strategy is to sell the asset when its transformed price $\Theta_t$ first reaches $\theta_1^*$ or above. The corresponding threshold in the original price scale is given by $p_1^*:=s^{-1}(\theta_1^*)=(\theta_1^*)^{1/\beta}$. Since $\theta_1^*$ is the point of contact of the tangent line to $v_1$ which passes $(0,-kR^{\alpha})$, $\theta_1^*$ should solve
	\begin{align}
	v_1'(\theta)-\frac{v_1(\theta)+kR^{\alpha}}{\theta}=0.
	\label{eq:theta1eq}
	\end{align}
	Furthermore, we can deduce from a graphical inspection that the solution to \eqref{eq:theta1eq} is a down-crossing. For the large proportional transaction costs case $\xi=\frac{\lambda}{\gamma}>[\frac{\alpha}{\beta k}c^{1-\beta}(c-1)^{\alpha-1}]^{\frac{1}{\beta}}$, the straight line passing $(0,-kR^{\alpha})$ can touch $v_1$ at two distinct locations. A simple geometric inspection will tell us that the required root is the smaller one. Using the representation of $v_1(\theta)$ in \eqref{eq:v1rep}, \eqref{eq:theta1eq} can be rewritten as
	\begin{align*}
	0&=R^\alpha\left(1+\frac{\Psi}{R}\right)^{\alpha}\left(\frac{\gamma}{R+\Psi}\right)^{\beta}f'\left(\left(\frac{\gamma}{R+\Psi}\right)^{\beta}\theta\right)\\
	&\qquad-\frac{R^\alpha\left(1+\frac{\Psi}{R}\right)^{\alpha}f\left(\left(\frac{\gamma}{R+\Psi}\right)^{\beta}\theta\right)+kR^{\alpha}}{\theta}.
	\end{align*}
	A further substitution of $x=(\frac{\gamma}{R+\Psi})^\beta \theta$ leads to
	\begin{align}
	\left(1+\frac{\Psi}{R}\right)^{\alpha}[xf'(x)-f(x)]=k.
	\label{eq:x1}
	\end{align}
	Then $p_1^*=(\theta_1^*)^{1/\beta}=\frac{R+\Psi}{\gamma}(x_1^*)^{1/\beta}$ where $x_1^*$ is defined as the solution to \eqref{eq:x1} which is equivalent to \eqref{eq:p1eq}. 
	
	When $\xi=\frac{\lambda}{\gamma}>[\frac{\alpha}{\beta k}c^{1-\beta}(c-1)^{\alpha-1}]^{\frac{1}{\beta}}$, Lemma \ref{lem:shapef} implies that $v_1$ is first concave increasing, reaching a global maximum at some $\theta_2^*$, concave decreasing and finally convex decreasing with  $\displaystyle \lim_{\theta\to \infty}v_1(\theta)= -\infty$. There are two further possibilities. 
	
	If $v_1(\theta_2^*)> -kR^{\alpha}$, then there must exist $0\leq \hat{\theta}_1<\hat{\theta}_2$ such that we have $g_2(\theta)=-kR^{\alpha}$ on $[0,\hat{\theta}_1]\cup [\hat{\theta}_2,\infty)$ and $g_2(\theta)=v_1(\theta)$ on $[\hat{\theta}_1,\hat{\theta}_2]$. The smallest concave majorant of $g_2(\theta)$ is formed by a chord passing $(0,-kR^\alpha)$ which touches $v_1$ at some $\theta_1^*<\theta_2^*$ on $\theta< \theta_1^*$, a horizontal line at level $g(\theta_2^*)$ on $\theta>\theta_2^*$, and the function $g_2$ itself on $\theta_1^*\leq \theta\leq \theta_2^*$. See Figure \ref{fig:case2}. The smallest concave majorant is
	\begin{align*}
	\bar{g}_2(\theta)=
	\begin{cases}
	\frac{v_1(\theta_1^*)+kR^{\alpha}}{\theta_1^*}\theta-kR^{\alpha},& \theta<\theta_1^*;\\
	v_1(\theta),& \theta_1^*\leq \theta\leq \theta_2^*; \\
	v_1(\theta_2^*),& \theta>\theta_2^*.
	\end{cases}
	\end{align*}
	This gives the form of the value function in \eqref{eq:exitvalfun2a}.
	
	The optimal strategy is to purchase the asset when its transformed price $\Theta_t$ first enters the interval $[\theta_1^*,\theta_2^*]$. The boundary of the purchase regions in the naive scale can be recovered via $p_i^*=(\theta_i^*)^{1/\beta}$ for $i=1,2$. Given that $\theta_2^*$ is the maximizer of $v_2(\theta)$, using the representation of \eqref{eq:v1rep} $\theta_2^*$ should then solve $f'((\frac{\gamma}{R+\Psi})^{\beta}\theta)=0$. Using \eqref{eq:f_firstder}, $x_2^*:=(\frac{\gamma}{R+\Psi})^\beta \theta_2^*$ is a solution to
	\begin{align*}
	h_1(x^{-\frac{1}{\beta}})=c^{1-\beta}(c-1)^{\alpha-1}\left(x^{-\frac{1}{\beta}}+\frac{\xi \alpha}{\beta}\right)-k\xi\left(x^{-\frac{1}{\beta}}+\xi\right)^\beta=0.
	\end{align*}
	Then $p_2^*=(\theta_2^*)^{1/\beta}=\frac{R+\Psi}{\gamma}(x_2^*)^{1/\beta}$ where $x_2^*$ is given by the solution to \eqref{eq:p2eq}.
	
	If $v_1(\theta_2^*)\leq -kR^{\alpha}$ instead, then $v_1(\theta)\leq -kR^{\alpha}$ for all $\theta$. Thus we have $g_2(\theta)=-kR^{\alpha}$ which is a flat horizontal line, and it is also the smallest concave majorant of itself, i.e. $\bar{g}_2(\theta)=-kR^{\alpha}$. The optimal strategy is not to trade at all at any price level such that the utility received is always $U(-R)=-kR^{\alpha}$. See Figure \ref{fig:case3}. The ``never purchase'' case arises if and only if $v_1(\theta_2^*)\leq -kR^{\alpha}$ or equivalently
	\begin{align*}
	R^\alpha\left(1+\frac{\Psi}{R}\right)^{\alpha}f\left(\left(\frac{\gamma}{R+\Psi}\right)^{\beta}\theta_2^*\right) \leq -kR^{\alpha} \iff \left(1+\frac{\Psi}{R}\right)^{\alpha}f(x_2^*) \leq -k
	\end{align*}
	where $x_2^*$ is the maximizer of $f$ introduced in Lemma \ref{lem:shapef} and it is independent of $\Psi$ and $R$. Using the fact that $f(0)=-k$ and Lemma \ref{lem:zerobound}, we deduce that $-k<f(x^*)<0$ and hence there must exist $C:=[-\frac{k}{f(x^*)}]^{1/\alpha}-1>0$ such that $\left(1+\frac{\Psi}{R}\right)^{\alpha}f(x^*) \leq -k$ if and only if $\Psi/R\geq C$. \qed
\end{proof}

\begin{proof}[Proof of Proposition \ref{prop:entryspecial}]
	Omitted since it is largely than same as the proof of Proposition \ref{prop:entry}. \qed
\end{proof}

\begin{proof}[Proof of Proposition \ref{prop:compstat}]
	
	From the proof of Proposition \ref{prop:entry}, the required solution to equation \eqref{eq:p1eq} is a down-crossing. Then given that the left hand side of \eqref{eq:p1eq} is increasing in $\Psi$ (when evaluated at $x=x_1^*$) we can deduce $x_1^*$ and in turn $p_1^*$ are both increasing in $\Psi$.
	
	To show that $p_1^*$ is decreasing in $\gamma$, consider a substitution of $q=\frac{x^{1/\beta}}{\gamma}$. Then $p_1^*=(R+\Psi)q_1^*$ where $q_1^*$ is the solution to
	\begin{align}
	k&=\left(1+\frac{\Psi}{R}\right)^{\alpha}\nonumber\\
	&\qquad\times\frac{k\left(\lambda q+1\right)^\beta\left[\lambda\left(1-\frac{\alpha}{\beta}\right)q+1\right]-\frac{\alpha}{\beta}c^{1-\beta}(c-1)^{\alpha-1}\lambda \gamma^{\beta}\left(1-\frac{\alpha}{\beta}\right)q^{\beta+1}}{\left(\lambda q+1\right)^{\beta-\alpha+1}}
	\label{eq:qeq1}
	\end{align}
	where the right hand side of \eqref{eq:qeq1} is decreasing in $\gamma$. Hence $q_1^*$ and in turn $p_1^*$ are both decreasing in $\gamma$.
	
	The monotonicity of $p_2^*$ with respect to $\Psi$ is trivial because equation \eqref{eq:p2eq} which defines $x_2^*$ does not depend on $\Psi$. To check the monotonicity with respect to $\gamma$, consider a substitution of $q=\frac{x^{1/\beta}}{\gamma}$ again so that $p_2^*=(R+\Psi)q_2^*$ where $q_2^*$ is defined as the solution to
	\begin{align}
	c^{1-\beta}(c-1)^{\alpha-1}\left(\frac{1}{q}+\frac{\lambda \alpha}{\beta}\right)-\frac{k\lambda}{\gamma^{\beta}}\left(\frac{1}{q}+\lambda\right)^\beta=0.
	\label{eq:qeq}
	\end{align}
	From the proof of Lemma \ref{lem:shapef}, the solution to $h_1(x^{-\frac{1}{\beta}})=0$ is a down-crossing. This property is inherited by \eqref{eq:qeq}. Moreover, the left hand side of \eqref{eq:qeq} is increasing in $\gamma$. Hence $q_2^*$ and in turn $p_2^*$ is increasing in $\gamma$.
	
	Similarly, consider a substitution of $y=\lambda^\beta x$. Then $p_2^*=\frac{R+\Psi}{\lambda \gamma }(y_2^*)^{1/\beta}$ where $y_2^*$ is defined as the solution to
	\begin{align}
	c^{1-\beta}(c-1)^{\alpha-1}\left(y^{-\frac{1}{\beta}}+\frac{ \alpha}{\beta\gamma}\right)-\frac{k\lambda^\beta}{\gamma}\left(y^{-\frac{1}{\beta}}+\frac{1}{\gamma}\right)^\beta=0.
	\label{eq:y}
	\end{align}
	The left hand side of \eqref{eq:y} is decreasing in $\lambda$ and hence $y_2^*$ is decreasing in $\lambda$. Therefore $p_2^*$ is decreasing in $\lambda$ as well. \qed
\end{proof}

{\cred
\subsection{Proof of Proposition \ref{prop:uti_discount}}
\label{app:discount}

We first show that the solution method in Section \ref{sect:mgmethod} can be extended to the problem with utility-discounting. A general exposition can be found in Dayanik and Karatzas~\cite{dayanik-karatzas03} but we will outline the key ideas under our specific model to introduce some notation to be used.

Let $\mathcal{A}:=\frac{\sigma^2}{2}\frac{d^2}{dp^2}+\mu\frac{d}{dp}$ be the infinitesimal generator of $P$. The second order ordinary differential equation $\mathcal{A}u(p)=\delta u(p)$ with $\delta>0$ admits $r_1(p):=p^{\omega_1}$ and $r_2(p):=p^{\omega_2}$ as two linearly independent solutions where $\omega_1<0<\omega_2$ are the distinct real roots to the quadratic equation 
\begin{align}
\frac{\sigma^2}{2}r^2+\left(\mu-\frac{\sigma^2}{2}\right)r-\delta=0.
\label{eq:quad}
\end{align}
Moreover, for $p\in[a,b]\subseteq \mathcal{J}$, $\varphi(p;a,b):=\mathbb{E}\left[e^{-\delta \tau_a} {\mathbbm 1}_{\{\tau_a<\tau_b\}}|P_0=p\right]$ is the solution to $\mathcal{A}\varphi=\delta \varphi$ with boundary conditions $\varphi(a;a,b)=1$ and $\varphi(b;a,b)=0$ where $\tau_{\ell}:=\inf\{t\geq 0: P_t=\ell\}$ (likewise, $\vartheta(p;a,b):=\mathbb{E}\left[e^{-\delta \tau_b} {\mathbbm 1}_{\{\tau_b<\tau_a\}}|P_0=p\right]$ has similar properties). From this, we obtain
\begin{align*}
\varphi(p;a,b)=\frac{r_2(b)r_1(p)-r_1(b)r_2(p)}{r_1(a)r_2(b)-r_2(a)r_1(b)},\qquad \vartheta(p;a,b)=\frac{r_1(a)r_2(p)-r_2(a)r_1(p)}{r_1(a)r_2(b)-r_2(a)r_1(b)}.
\end{align*}

Now consider a discounted optimal stopping problem in form of $$V(p)=\sup_{\tau\in\mathcal{T}}\mathbb{E}\bigl[e^{-\delta \tau}G(P_\tau)\bigl|P_0=p\bigr].$$ As before, it is sufficient to search for an optimal stopping time in form of $\tau_{a,b}:=\tau_a\wedge \tau_b$. Then 
\begin{align*}
J(p;\tau_{a,b})&:=\mathbb{E}\bigl[e^{-\delta \tau_{a,b}}G(P_{\tau_{a,b}})\bigl|P_0=p\bigr]\\
&=G(a)\mathbb{E}\bigl[e^{-\delta \tau_a} {\mathbbm 1}_{\{\tau_a<\tau_b\}}\bigl|P_0=p\bigr]+G(b)\mathbb{E}\bigl[e^{-\delta \tau_b} {\mathbbm 1}_{\{\tau_b<\tau_a\}}\bigl|P_0=p\bigr]\\
&=G(a)\varphi(p;a,b)+G(b)\vartheta(p;a,b)\\
&=G(a)\frac{r_2(b)r_1(p)-r_1(b)r_2(p)}{r_1(a)r_2(b)-r_2(a)r_1(b)}+G(b)\frac{r_1(a)r_2(p)-r_2(a)r_1(p)}{r_1(a)r_2(b)-r_2(a)r_1(b)}\\
&=r_1(p)\biggl[\frac{r_2(b)/r_1(b)-r_2(p)/r_1(p)}{r_2(b)/r_1(b)-r_2(a)/r_1(a)}\times\frac{G(a)}{r_1(a)}\\
&\qquad+\frac{r_2(p)/r_1(p)-r_2(a)/r_1(a)}{r_2(b)/r_1(b)-r_2(a)/r_1(a)}\times \frac{G(b)}{r_1(b)}\biggr]\\
&=r_1\bigl(s^{-1}(\theta)\bigr)\biggl[\frac{s(b)-s(p)}{s(b)-s(a)}\times \phi\bigl(s(a)\bigr)+\frac{s(p)-s(a)}{s(b)-s(a)}\times \phi\bigl(s(b)\bigr)\biggr]
\end{align*}
where $s(x):=r_2(x)/r_1(x)=x^{\omega_2-\omega_1}$, $\theta:=s(p)$ and $$\phi(x):=\bigl(G/r_1\bigr)\bigl(s^{-1}(x)\bigr)=x^{-\frac{\omega_1}{\omega_2-\omega_1}}G(x^{1/(\omega_2-\omega_1)}).$$ The optimal stopping rule can be deduced by maximizing the above with respect to $a$ and $b$. Upon replacing the dummy variables via $a'=s(a)$ and $b'=s(b)$, we have
\begin{align*}
V(p)=\sup_{a,b:a\leq p\leq b}J(p;\tau_{a,b})&=r_1\bigl(s^{-1}(\theta)\bigr)\sup_{a',b:a'\leq\theta\leq b'}\left[\frac{b'-\theta}{b'-a'}\phi(a')+\frac{\theta-a'}{b'-a'}\phi(b')\right]\\
&=:v(\theta).
\end{align*}
The supremum in the second last term can be characterized by $\bar{\phi}$ the smallest concave majorant to the scaled payoff function $\phi=\frac{G}{r_1}\circ s^{-1}$, and the value function in the original coordinate is given by $V(p)=v(s(p))=r_1(p)\bar{\phi}(s(p))$.

\begin{proof}[Proof of Proposition \ref{prop:uti_discount}]

Similar to the baseline problem, \eqref{eq:uti_discount} which features utility-discounting can be solved by decomposing the problem into the sub-problems of exit and entry. The discounted exit problem is 
\begin{align*}
V_1(p;H):=\sup_{\nu\in\mathcal{T}}\mathbb{E}\bigl[e^{-\delta \nu}U(\gamma P_{\nu}-H)\bigl|P_0=p\bigr]
\end{align*}
where $H\geq 0$ is some given constant. The scaled-payoff function is given by
\begin{align*}
g_1(\theta)&:=\theta^{-\frac{\omega_1}{\omega_2-\omega_1}}U(\gamma \theta^{\frac{1}{\omega_2-\omega_1}}-H)\\
&=
\begin{cases}
-k\theta^{-\frac{\omega_1}{\omega_2-\omega_1}}(H-\gamma \theta^{\frac{1}{\omega_2-\omega_1}})^{\alpha},& 0\leq \theta<\left(\frac{H}{\gamma}\right)^{\omega_2-\omega_1};\\
\theta^{-\frac{\omega_1}{\omega_2-\omega_1}}(\gamma \theta^{\frac{1}{\omega_2-\omega_1}}-H)^{\alpha},&\theta\geq \left(\frac{H}{\gamma}\right)^{\omega_2-\omega_1}.
\end{cases}
\end{align*}
Note that $g_2(0)=g_2((\frac{H}{\gamma})^{\omega_2-\omega_1})=0$ and $g_2(\theta)< 0$ on $0<\theta< (\frac{H}{\gamma})^{\omega_2-\omega_1}$. Furthermore, on $\theta> (\frac{H}{\gamma})^{\omega_2-\omega_1}$, we have
\begin{align*}
g_1'(\theta)=\theta^{-\frac{\omega_2}{\omega_2-\omega_1}}(\gamma \theta^{\frac{1}{\omega_2-\omega_1}}-H)^{\alpha-1}\left[-\frac{\omega_1(\gamma \theta^{\frac{1}{\omega_2-\omega_1}}-H)}{\omega_2-\omega_1}+\frac{\alpha \gamma \theta^{\frac{1}{\omega_2-\omega_1}}}{\omega_2-\omega_1}\right]> 0
\end{align*}
and
\begin{align*}
g_1''(\theta)=\frac{\theta^{-\frac{\omega_1}{\omega_2-\omega_1}-2}(\gamma\theta^{\frac{1}{\omega_2-\omega_1}})^2(\gamma \theta^{\frac{1}{\omega_2-\omega_1}}-H)^{\alpha-2}}{(\omega_2-\omega_1)^2}h(1-H \gamma^{-1} \theta^{-\frac{1}{\omega_2-\omega_1}})
\end{align*}
where $$h(z):=\omega_1\omega_2 z^2+\alpha(1-\omega_1-\omega_2)z-\alpha(1-\alpha).$$

We now show that on $z\in[0,1]$ the quadratic function $h(z)$ is strictly negative which in turn will imply $g_1$ is a strictly increasing and concave function on $\theta\geq (\frac{H}{\gamma})^{\omega_2-\omega_1}$. Note that $\omega_1+\omega_2=1-\frac{2\mu}{\sigma^2}=\beta$ by considering the sum of roots of the quadratic equation \eqref{eq:quad}. Then we have $h(0)=-\alpha(1-\alpha)<0$ and $$h(1)=\omega_1\omega_2 +\alpha(1-\omega_1-\omega_2)-\alpha(1-\alpha)=(\alpha-\omega_1)(\alpha-\omega_2)<0$$ since $\omega_1<0$ and $0<\alpha<\beta=\omega_1+\omega_2<\omega_2$. If $\beta\geq 1$, then $$0\geq \alpha( 1-\beta)= \alpha(1-\omega_1-\omega_2)=h'(0)> h'(1)$$ and we must have $h$ being decreasing and in turn negative on $z\in[0,1]$. If $\beta< 1$ and $\alpha\geq -\frac{2\omega_1\omega_2}{1-\beta}$, then $$h'(0)>h'(1)=2\omega_1\omega_2+\alpha(1-\omega_1-\omega_2)=2\omega_1\omega_2+\alpha(1-\beta)\geq 0$$ and $h$ must be increasing and thus negative on $z\in[0,1]$. Finally, if  $\beta<1$ and $\alpha< -\frac{2\omega_1\omega_2}{1-\beta}$, we can compute the discriminant of the quadratic function $h$ as
\begin{align*}
\triangle:=\alpha^2(1-\omega_1-\omega_2)^2+4\omega_1\omega_2\alpha(1-\alpha)&=\alpha^2(1-\beta)^2+4\omega_1\omega_2\alpha(1-\alpha)\\
&<\alpha\left\{-2\omega_1\omega_2(1-\beta)+4\omega_1\omega_2(1-\alpha)\right\}\\
&=-2\omega_1\omega_2\alpha\left\{(\alpha-1)-(\alpha-\beta)\right\}<0
\end{align*}
and hence $h(z)<0$ for all $z$.

Now, the smallest concave majorant of $g_1$ can be formed by drawing a straight line from $(0,0)$ which touches $g_1$ at some $\theta^*>(\frac{H}{\gamma})^{\omega_2-\omega_1}$. The point of contact is given by the unique $\theta^*$ satisfying $\frac{g_1(\theta^*)}{\theta^*}=g_1'(\theta^*)$. The required $\theta^*$ is thus given by the solution to the equation
\begin{align*}
\frac{\theta^{-\frac{\omega_1}{\omega_2-\omega_1}}(\gamma \theta^{\frac{1}{\omega_2-\omega_1}}-H)^{\alpha}}{\theta}&=(\gamma \theta^{\frac{1}{\theta_2-\theta_1}}-H)^{\alpha-1}\theta^{-\frac{\omega_2}{\omega_2-\omega_1}}\\
&\qquad\times\left[\frac{\alpha\gamma \theta^{\frac{1}{\omega_2-\omega_1}}}{\omega_2-\omega_1}-\frac{\omega_1}{\omega_2-\omega_1}(\gamma \theta^{\frac{1}{\omega_2-\omega_1}}-H)\right]
\end{align*}
 which admits an explicit solution $$\theta^*=\left(\frac{\omega_2 H}{\gamma(\omega_2-\alpha)}\right)^{\omega_2-\omega_1}.$$ The smallest concave majorant of $g_1$ is then
\begin{align*}
\bar{g}_1(\theta)&=\bar{g}_1(\theta;H)\\
&=
\begin{cases}
\left(\frac{\omega_2}{\gamma(\omega_2-\alpha)}\right)^{-\omega_2}\left(\frac{\alpha}{\omega_2-\alpha}\right)^{\alpha}H^{\alpha-\omega_2}\theta,&\theta<\left(\frac{\omega_2 H}{\gamma(\omega_2-\alpha)}\right)^{\omega_2-\omega_1};\\
\theta^{-\frac{\omega_1}{\omega_2-\omega_1}}(\gamma \theta^{\frac{1}{\omega_2-\omega_1}}-H)^{\alpha},&\theta\geq \left(\frac{\omega_2 H}{\gamma(\omega_2-\alpha)}\right)^{\omega_2-\omega_1}.
\end{cases}
\end{align*}
The value function of the exit problem is thus $$V_1(p;H)=r_1(p)\bar{g}_1\bigl(s(p)\bigr)=p^{\omega_1}\bar{g}_{1}(p^{\omega_2-\omega_1};H).$$ From the form of the value function, the optimal exit strategy is to sell the asset whenever its price level is at or above $p^*=s^{-1}(\theta^*)=\frac{\omega_2 H}{\gamma(\omega_2-\alpha)}$.

Now we look at the entry problem. Note that unlike the problem without discounting which objective function is \eqref{eq:obj}, the strategy of $\tau=\infty$ now yields a value of zero rather than $U(-R)$ under utility-discounting. Hence the payoff function for the entry problem is $\max\left\{V_1(P_{\tau};\lambda P_{\tau}+\Psi+R),0\right\}$ rather than $\max\left\{V_1(P_{\tau};\lambda P_{\tau}+\Psi+R),U(-R)\right\}$. The entry problem is thus
\begin{align*}
V_2(p):=\sup_{\tau}\mathbb{E}\biggl[e^{-\delta \tau}\max\Bigl\{V_1\bigl(P_{\tau};\lambda P_{\tau}+\Psi+R\bigr),0\Bigr\}\biggr].
\end{align*}
It is also clear that $V_1$ is non-negative. Then the scaled payoff function of the entry problem is
\begin{align*}
g_2(\theta)&:=\theta^{-\frac{\omega_1}{\omega_2-\omega_1}}V_1(\theta^{\frac{1}{\omega_2-\omega_1}};\lambda \theta^{\frac{1}{\omega_2-\omega_1}}+\Psi+R)\\
&=\bar{g_1}(\theta;\lambda \theta^{\frac{1}{\omega_2-\omega_1}}+\Psi+R )\\
&=\left(\frac{\omega_2}{\gamma(\omega_2-\alpha)}\right)^{-\omega_2}\left(\frac{\alpha}{\omega_2-\alpha}\right)^{\alpha}(\lambda \theta^{\frac{1}{\omega_2-\omega_1}}+\Psi+R)^{\alpha-\omega_2}\theta.
\end{align*}
The last equality holds because
\begin{align*}
\left(\frac{\omega_2 }{\gamma(\omega_2-\alpha)}\right)^{\omega_2-\omega_1} (\lambda \theta^{\frac{1}{\omega_2-\omega_1}}+\Psi+R )^{\omega_2-\omega_1}&\geq \left(\frac{\omega_2 }{\gamma(\omega_2-\alpha)}\right)^{\omega_2-\omega_1} \lambda ^{\omega_2-\omega_1}\theta\\
&=\left(\frac{\omega_2 }{\omega_2-\alpha}\right)^{\omega_2-\omega_1} \left(\frac{\lambda}{\gamma}\right)^{\omega_2-\omega_1}\theta\\
&> \theta
\end{align*}
due to the facts that $R>0$, $\Psi\geq 0$, $\omega_1<\omega_2$ and $\gamma\leq 1\leq \lambda$, and thus the linear regime of $\bar{g}_1$ always applies when evaluating $\bar{g_1}(\theta;\lambda \theta^{\frac{1}{\omega_2-\omega_1}}+\Psi+R )$. It remains to identify $\bar{g}_2$ the smallest concave majorant of $g_2$. But by following similar (and indeed less tedious) calculus exercise as the one for the exit problem, one can verify that $\theta\to (\lambda \theta^{\frac{1}{\omega_2-\omega_1}}+\Psi+R)^{\alpha-\omega_2}\theta$ is an increasing concave function for all $\theta\geq 0$ given the condition $\alpha<\beta$. We hence must have $g_2=\bar{g}_2$ and the strategy of $\tau=0$ is optimal. \qed
\end{proof}

}
\end{document}